%% file: manuscript.tex
\documentclass{article}

\usepackage{amsmath, amsthm, amssymb, dsfont, mathtools}
\usepackage{mathrsfs}

\usepackage{graphicx}
\usepackage{verbatim}
\usepackage{natbib}
\usepackage{caption}
\usepackage{subcaption}
\usepackage{fancyvrb}
\usepackage{enumerate}
\usepackage[inline]{enumitem}
\usepackage{relsize}
\usepackage{hyperref}
\usepackage[margin=1.4in]{geometry}
\hypersetup{colorlinks,citecolor=blue,urlcolor=blue,linkcolor=blue}
\usepackage{diagbox}
\usepackage{nikos_tex}
\usepackage{float}
\usepackage{mdframed}
\usepackage{multirow}
\usepackage{booktabs}
\usepackage[utf8]{inputenc}
\usepackage[english]{babel}

\theoremstyle{definition}
\newtheorem{prop}{Proposition}

\newtheorem{lemm}[prop]{Lemma}
\newtheorem{theo}[prop]{Theorem}
\newtheorem{rema}[prop]{Remark}
\newtheorem{defi}[prop]{Definition}

\newcounter{procedure}
\newcounter{procedurevariant}[procedure]
\renewcommand{\theprocedurevariant}{\theprocedure\ensuremath{'}}

\usepackage{tikz}
\usetikzlibrary{matrix,calc,arrows.meta,backgrounds}
\definecolor{naiveblue}{RGB}{36,105,151}
\definecolor{tregred}{RGB}{164,58,77}
\definecolor{restgreen}{RGB}{29,125,94}
\definecolor{actpurple}{RGB}{119,79,151}

\date{Draft manuscript: September, 2026}

\title{Empirical Bayes prepivoting under group invariance:\\ false discovery rate control and moderated t-tests}

\author{
Nikolaos Ignatiadis\\
\texttt{ignat@uchicago.edu}
\and 
Etienne Roquain\\
\texttt{etienne.roquain@upmc.fr}
}

\begin{document}

\maketitle

\begin{abstract}
We consider simultaneously testing hypotheses about thousands of units, e.g., genes or proteins, where each unit yields a handful of replicate measurements and we test whether its mean is zero. A widely used approach in genomics, implemented in the limma software, borrows strength across units to learn the distribution of the unit-specific variances, then computes moderated t-statistics. Here we develop the first procedures using moderated t-statistics that (i) control the false discovery rate (FDR) in finite samples under independence across units and null group invariance, without assuming limma's hierarchical model, and (ii) match the power of an oracle local false discovery rate procedure in a sparse asymptotic regime under limma's working model. Benjamini-Hochberg (BH) with standard t-test p-values has asymptotically zero power in the same regime. Our approach learns the variance distribution from statistics that depend on each unit's data only through its orbit under a compact group of transformations that preserves the null distributions, such as sign flips or orthogonal rotations. We then calibrate the resulting moderated t-statistics against group-transformed statistics pooled across units to obtain compound p-values, which we use with BH and a close variant. For small finite groups, we also construct a Selective SeqStep+ procedure using the same learned scores. Our approach extends to two-sample tests and tests of linear model coefficients.
\end{abstract}

\section{Introduction} 
\label{sec:introduction}

In this paper, we consider the following multiple testing problem. For each of $n$ units, indexed by $i$, we observe $K$ normally distributed replicates,
\begin{equation}
Z_{ij} \simindep \mathrm{N}(\mu_i, \sigma_i^2),\;\; \text{ for } j=1,\ldots,K,\;\;\; i=1,\ldots,n,
\label{eq:normal_samples}
\end{equation}
with both $\mu_i$ and $\sigma_i^2$ unknown.  We are interested in testing the null hypotheses, $H_{0,i}: \mu_i = 0$ for $i=1,\ldots,n$, while controlling the false discovery rate (FDR). 

We endow the parameters $(\mu_i, \sigma_i^2)$ with the following working model due to~\citet{lonnstedt2002replicated} and~\citet{smyth2004linear}, widely used in genomics,
\begin{equation}
\sigma_i^2 \simindep \frac{\nu_0 s_0^2}{ \chi^2_{\nu_0}},\;\;\;\;
\mu_i \mid \sigma_i^2 \simindep \pi_0 \delta_0 \,+\, (1-\pi_0) \mathrm{N}(0, \lambda\sigma_i^2/K),
\label{eq:lonnsted_prior}
\end{equation}
where $\pi_0 \in [0,1]$, $\lambda>0$, $\nu_0>0$ and $s_0^2>0$ are unknown hyperparameters, and $\chi^2_{\nu_0}$ denotes the chi-square distribution with $\nu_0$ degrees of freedom.

We develop procedures with the following properties:
\begin{enumerate}[leftmargin=*, labelsep=4pt, itemindent=0pt]
\item[(i)] they attain a finite sample FDR control guarantee when only~\eqref{eq:normal_samples} holds (i.e., without assuming~\eqref{eq:lonnsted_prior}), or more generally when the error terms $Z_{ij} - \mu_i,\; j=1,\ldots,K$ are rotationally invariant, or even under the weaker assumption that their distributions are invariant under coordinatewise sign flips (as is the case, for example, for independent errors symmetric about zero);
\item[(ii)] under the working model~\eqref{eq:normal_samples}--\eqref{eq:lonnsted_prior}, in a sparse regime with $\pi_0 \to 1$ and $\lambda \to \infty$, 
they asymptotically match the performance of an oracle local false discovery rate procedure that knows the hyperparameters in~\eqref{eq:lonnsted_prior}. 
\end{enumerate}
The multiple testing setting described in~\eqref{eq:normal_samples} is long standing (Section~\ref{sec:existing} reviews existing approaches) and it has many useful applications in high-throughput biology~\citep{ritchie2015limma}; see Section~\ref{subsec:yves} for an application to proteomics.

The specific model illustrates a broader statistical idea: if we were interested only in desideratum (i) above, then the natural benchmark procedure would be to use a separate group invariance randomization test for each hypothesis~\citep{lehmann2005testing,ritzwoller2026randomization} and use the resulting p-values as input to the Benjamini-Hochberg (BH) procedure~\citep{benjamini1995controlling}. However, we will argue that this approach can have low power in some settings for two main reasons. First, if the randomization test uses a finite group (such as sign flips) and the number of replicates $K$ in~\eqref{eq:normal_samples} is small, the resulting p-values may be too coarse to allow the BH procedure to reject any hypothesis. A second limitation concerns adaptivity: standard implementations typically construct each test statistic using only the data from that unit, without borrowing information across units. Moreover, even if such borrowing is introduced, e.g., through empirical Bayes methods, separate randomization can undo the resulting power gains  (Proposition~\ref{prop:sep_pvalue_limma}). 

In this paper, we develop an alternative approach in the case of independence across units.\footnote{\citet{RW2005,arlot2010nonasymptotic} develop methods based on randomization tests for family-wise error control that account for dependence between test statistics; \citet{hemerik2018false} do the same for the median of the false discovery proportion. Here we assume independence across units and focus on the power gains from borrowing information across them.} Our starting point is the approach of~\citet{barber2026false}, who calibrate each unit's test statistic against a reference distribution obtained by pooling simulated or permuted test statistics across units. They show that this yields compound p-values~\citep{armstrong2022false,ignatiadis2025asymptotic}, which relax individual p-value validity to a single condition on the sum of the null rejection probabilities. They also show that applying BH gives a finite sample FDR bound of $1.93\alpha$. In Section~\ref{sec:compound_random}, we extend this approach to more general groups beyond permutations, but perhaps more importantly, we allow the test statistic to be learned across units using quantities that are unchanged by the group transformations. In particular, this allows us to learn powerful scores using shared properties of the null distributions when the null hypotheses are composite. We show that the same $1.93\alpha$ FDR bound holds with learned scores (Theorem~\ref{theo:compound-p}), and give a modification that controls the FDR at level $\alpha$ (Theorem~\ref{theo:compound-ddr}).

For sufficiently small groups, even compound BH cannot make any discoveries (Proposition~\ref{prop:no_discoveries}). In Section~\ref{subsec:seqstep}, we therefore also construct a procedure based on Selective SeqStep+~\citep{barber2015controlling}, using the same learned scores. This procedure controls the FDR at level $\alpha$ (Theorem~\ref{theo:seqstep_multi}) and can make discoveries even with a group of order two.

We pursue these ideas by robustly estimating the variance-prior parameters
in~\eqref{eq:lonnsted_prior} to construct the test statistics, which we interpret as ``empirical Bayes prepivots.'' The resulting construction satisfies desideratum (i) of finite sample FDR control (as explained above), and moreover, we establish in Theorem~\ref{theo:sparse_regime} of Section~\ref{sec:power_analysis} that it also satisfies desideratum (ii) of asymptotic power optimality when \eqref{eq:lonnsted_prior} is true. In the same asymptotic regime, Theorem~\ref{theo:sparse_regime} also shows that BH applied to ordinary t-test p-values or separate sign flip p-values has asymptotically zero power. The results are corroborated by the simulation study of  Section~\ref{sec:numerical_results} and illustrated in a proteomics study in 
Section~\ref{subsec:yves}.

To illustrate the generality of the framework, we also describe
applications to two-sample testing and tests of linear model coefficients in Section~\ref{sec:other_settings}; with an application to differential methylation in Section~\ref{subsec:methylation}.

\section{Existing approaches}
\label{sec:existing}
We begin by providing a (non-exhaustive) review of existing procedures for the problem. This serves two purposes: it motivates our construction, and it provides benchmark methods to which we will compare our proposal; see Section~\ref{sec:conclusion} for further related work.

To describe existing approaches, it will be convenient to write $\boldZ_i = (Z_{i1},\ldots,Z_{iK})^\top \in \RR^K$ for the $K$ replicates of the $i$-th unit. Given $\boldz = (z_1,\ldots,z_K)^\top \in \mathbb R^K$, let us also define the following summary statistics:
\begin{equation}
\hat{\mu}(\boldz):= \frac{1}{K}\sum_{j=1}^K z_j, \quad \hat{\sigma}^2(\boldz):= \frac{1}{K-1}\sum_{j=1}^K (z_j - \hat{\mu}(\boldz))^2.
\label{eq:summary_statistics}
\end{equation}
We also write $\hat{\mu}_i = \hat{\mu}(\boldZ_i)$, $\hat{\sigma}_i^2 = \hat{\sigma}^2(\boldZ_i)$ and $\nu=K-1$.

\paragraph{Standard t-tests with BH ($t\textnormal{-}\mathrm{BH}(\alpha)$).} A textbook approach is to define the t-statistic function for $\boldz \in \mathbb R^K$ as
\begin{equation}
T(\boldz) := \frac{\sqrt{K}\hat{\mu}(\boldz)}{\hat{\sigma}(\boldz)},
\label{eq:standard_t}
\end{equation}
and then compute p-values as $P_i^{t} := 2 \bar{F}_{t,\nu}(|T(\boldZ_i)|)$, where $\bar{F}_{t,\nu}$ is the survival function of a $t$-distribution with $\nu$ degrees of freedom. Finally, one applies the Benjamini-Hochberg procedure to $P_1^{t},\ldots,P_n^{t}$ at level $\alpha$. This procedure controls the FDR at level $\alpha$ under~\eqref{eq:normal_samples}.

\paragraph{Oracle local false discovery rate thresholding ($\mathrm{LFDR}^{\text{OR}}(\alpha)$).} If we assume that the working model in~\eqref{eq:lonnsted_prior} holds, then we can compute the local false discovery rate~\citep{efron2001empirical, efron2010largescale} function for $\boldz \in \RR^K$,
\begin{equation}
\mathrm{lfdr}(\boldz) := \PP{\mu_i = 0 \mid \boldZ_i=\boldz}
\label{eq:local_fdr}
\end{equation}
under the above model (which also depends on the hyperparameters in~\eqref{eq:lonnsted_prior}). To turn the above into a procedure controlling the FDR, we can threshold the local false discovery rates at a level $c_\alpha$ chosen to ensure that the FDR is controlled at level $\alpha$. Specifically, we can define $c_{\alpha}$ as the largest solution of the following equation (setting $c_{\alpha} =0$ if no solution exists):
\begin{equation}
\frac{\pi_0\PP{\mathrm{lfdr}(\boldZ_i)\leq c_{\alpha} \mid \mu_i=0}}
{\pi_0\PP{\mathrm{lfdr}(\boldZ_i)\leq c_{\alpha} \mid \mu_i=0}
+\pi_1\PP{\mathrm{lfdr}(\boldZ_i)\leq c_{\alpha} \mid \mu_i \neq0}}
\stackrel{!}{=}\alpha.
\label{eq:mfdr_equilibrium}
\end{equation}
where $\stackrel{!}{=}$ emphasizes that $c_\alpha$ is chosen so that equality holds.
One then rejects the $i$-th null hypothesis if $\mathrm{lfdr}(\boldZ_i) \leq c_{\alpha}$.
Under both~\eqref{eq:normal_samples} and~\eqref{eq:lonnsted_prior}, \citet{storey2003positive} shows that the above procedure has FDR exactly equal to the display in~\eqref{eq:mfdr_equilibrium} times the probability of making at least one rejection, i.e., equal to $\alpha \cdot(1- \PP{\mathrm{lfdr}(\boldZ_i)> c_{\alpha}}^n)  \leq \alpha$, and so it controls the FDR at level $\alpha$. It does not, however, provide a finite sample guarantee for the FDR under~\eqref{eq:normal_samples} alone. 

\citet{abraham2022empirical} call such a procedure a q-value procedure following~\citet{storey2003positive}. An alternative oracle procedure is the cumulative local fdr (Clfdr) approach of~\citet{sun2007oracle}; see~\citet{abraham2022empirical}.
Empirical Bayes implementations of the q-value and Clfdr approaches estimate the hyperparameters in~\eqref{eq:lonnsted_prior} from the data, and apply the above to plug-in estimates $\widehat{\mathrm{lfdr}}(\boldz)$; such procedures only have FDR guarantees as $n \to \infty$.

\paragraph{Limma with BH.} 

The limma procedure~\citep{smyth2004linear} is perhaps the most widely-used approach in genomics for solving the above problem. Its starting point is to only posit part of the prior specification in~\eqref{eq:lonnsted_prior}, namely only the prior on the variances:
\begin{equation}
\sigma_i^2 \simindep \frac{\nu_0 s_0^2}{ \chi^2_{\nu_0}}.
\label{eq:limma_variance_prior}
\end{equation}
An oracle version of the limma procedure
proceeds as follows. First, it defines the following moderated variance and moderated t-statistic functions for $\boldz \in \RR^K$:
\begin{equation}
\hat{\sigma}^2(\boldz; \nu_0, s_0) := \frac{  \nu_0 s_0^2 + \nu \hat{\sigma}^2(\boldz)}{\nu_0 + \nu},\;\;\; T(\boldz; \nu_0, s_0) := \frac{\sqrt{K} \hat{\mu}(\boldz)}{\hat{\sigma}(\boldz; \nu_0, s_0)},
\label{eq:moderated_stats}
\end{equation}
where we make the dependence on $\nu_0$ and $s_0^2$ explicit. (Notionally, $\hat{\sigma}^2(\boldz) = \hat{\sigma}^2(\boldz; 0, 1)$ and $T(\boldz) = T(\boldz; 0, 1)$.) The key observation is that if~\eqref{eq:normal_samples} and~\eqref{eq:limma_variance_prior} hold, then 
\begin{equation}
T(\boldZ_i; \nu_0, s_0) \mid \mu_i =0 \sim t_{\nu + \nu_0},\,\; \text {or more generally, }\,\;\frac{\sqrt{K} \{\hat{\mu}(\boldZ_i)-\mu_i\}}{\hat{\sigma}(\boldZ_i; \nu_0, s_0)} \sim t_{\nu + \nu_0}, 
\label{eq:moderated_t_pivot}
\end{equation}
i.e., the moderated t-statistic 
has a $t$-distribution with $\nu + \nu_0$ degrees of freedom; see Proposition~\ref{prop:moderated_two_groups} in Supplement~\ref{sec:t_formulas} for a self-contained proof. Thus one can compute p-values as $P_i^{\mathrm{limma}} := 2 \bar{F}_{t,\nu+\nu_0}(|T(\boldZ_i; \nu_0, s_0)|)$.\footnote{ 
\citet{smyth2004linear} writes that ``the
added degrees of freedom [of $T(\boldZ_i; \nu_0, s_0)$ over $T(\boldZ_i)$] reflect the extra information which is borrowed
on the basis of the hierarchical model, from the ensemble of genes for inference about each individual gene.''
} These can then be used to control the FDR at level $\alpha$ by applying the Benjamini-Hochberg procedure to $P_1^{\mathrm{limma}},\ldots,P_n^{\mathrm{limma}}$.

An important feature of oracle limma is that it does not require knowledge of $\pi_0$ and $\lambda$ in~\eqref{eq:lonnsted_prior}. Nevertheless, the ranking induced by the moderated t-statistics $T(\boldZ_i; \nu_0, s_0)$ is equivalent to the ranking induced by the local false discovery rates $\mathrm{lfdr}(\boldZ_i)$ in~\eqref{eq:local_fdr}. The reason is that, as shown by~\citet{smyth2004linear} (also see the self-contained proof in Proposition~\ref{prop:localfdr_formula} of Supplement~\ref{sec:t_formulas}), 
\begin{equation}
\mathrm{lfdr}(\boldz)  = \left[
1+
\frac{\pi_1}{\pi_0}
(1+\lambda)^{(\nu_0+\nu)/2}
\left\{
\frac{\nu_0+\nu+T^2(\boldz;\nu_0,s_0)}
{(1+\lambda)(\nu_0+\nu)+T^2(\boldz;\nu_0,s_0)}
\right\}^{(\nu_0+K)/2}
\right]^{-1},
\label{eq:localfdr_smyth}
\end{equation}
and so $\mathrm{lfdr}(\boldz)$ is a decreasing function of $|T(\boldz;\nu_0,s_0)|$. We note that if we apply BH to oracle $P_i^{\mathrm{limma}}$ the resulting procedure is not exactly the same as the oracle $q$-value or Clfdr procedures described above, but they are all asymptotically equivalent under sparsity (i.e., when $\pi_0 \to 1$).

In practice, limma estimates the hyperparameters $\nu_0$ and $s_0^2$ from $\hat{\sigma}^2(\boldZ_1),\ldots,\hat{\sigma}^2(\boldZ_n)$, yielding estimates $\hat{\nu}_0$ and $\hat{s}_0^2$. These are plugged in, in lieu of the true hyperparameters in~\eqref{eq:moderated_stats}. This plug-in approach can be justified asymptotically (though we are not aware of such justification in the literature).

We call $T(\boldZ_i; \hat{\nu}_0, \hat{s}_0)$ an \emph{empirical Bayes pivot}. The reason is that (ignoring error in estimates of $\hat{\nu}_0, \hat{s}_0$) it follows the pivotal distribution in~\eqref{eq:moderated_t_pivot} when also integrating over the variance distribution in~\eqref{eq:limma_variance_prior}.
This is different from classical pivots for group invariance tests~\citep{janssen1997studentized, chung2013exact, fogarty2021prepivoted} wherein we would ask for a pivotal distribution for fixed $\sigma_i, \mu_i$. One of the main questions we address is how to use an empirical Bayes pivot in a way that yields finite-sample FDR control for fixed $\sigma_i$ and $\mu_i$, and higher power when these parameters are indeed drawn from the working hierarchical model.

\paragraph{Significance analysis of microarrays (SAM).} The SAM method~\citep{tusher2001significance, storey2003sam} uses a test statistic of the following form:\footnote{Our presentation of SAM is simplified and we omit details. We also note that the publications describe the two-sample case (see our Section~\ref{subsec:twosample}), however, the SAM software also allows for one-sample comparisons as described in the user's guide~\citep{chu_sam_guide}. 
This is the version of SAM we describe herein.
}
\begin{equation*}
T^{\mathrm{SAM}}(\boldz; s_0) := \frac{\hat{\mu}(\boldz)}{\hat{\sigma}(\boldz) + s_0},
\label{eq:SAM_T}
\end{equation*}
for some $s_0>0$. We point out the similarity to the moderated t-statistic in~\eqref{eq:moderated_stats}. The second ingredient of SAM is to estimate the FDR associated with a rejection rule of the form $|T^{\mathrm{SAM}}(\boldz; s_0)| \geq t$ for some $t>0$. The proposal is to do so via a pooled null distribution,\footnote{Later versions of SAM replaced outer averaging over $H \in \{-1,+1\}^K$ by taking a median. This is the SAM version studied by~\citet{hemerik2018false}. Nevertheless, the purpose of \cite{hemerik2018false} is to show finite sample confidence bounds for the false discovery proportion, which is a different aim from ours.}
\begin{equation}
\widehat{\mathrm{FDR}}^{\mathrm{SAM}}\!(t) := \frac{1}{2^K}\sum_{H \in \{-1,+1\}^K} \frac{\sum_{k=1}^n \ind\{  |T^{\mathrm{SAM}}(H\boldZ_k; s_0)|       \geq t\}}{\sum_{k=1}^n \ind\{  |T^{\mathrm{SAM}}(\boldZ_k; s_0)|       \geq t\}},
\label{eq:samfdr}
\end{equation}
where $H\boldZ_k = (H_1 Z_{k1},\ldots, H_K Z_{kK})$ is a copy of $\boldZ_k$ with some entries (the ones for which $H_j=-1$) having their sign flipped. 
The FDR estimate in~\eqref{eq:samfdr} can be used to  determine which hypotheses to reject and
will be central to our proposal below.

\paragraph{SeqStep+ methods.}

Several approaches exploit symmetry to control the FDR using the Selective SeqStep+ procedure of~\citet{barber2015controlling}; see~\citet{arias-castro2017distributionfree, ge2021clipper, zou2020new, tian2025conformalized}. We revisit such an approach in Section~\ref{subsec:seqstep}.

\section{BH with compound randomization p-values}
\label{sec:compound_random}
In this section we present our compound randomization framework and establish FDR guarantees in a general group invariance setting.

\subsection{Preliminaries on multiple testing with compound p-values}

In the general multiple testing setting, 
we observe data $Z$ drawn from an unknown distribution $\mathbb P$ on an underlying sample space, and consider testing $n$ hypotheses $H_{0,i}: \mathbb P \in \mathcal{P}_i$ for $i \in [n] := \{1,\dots,n\}$, where each $\mathcal{P}_i$ is a set of distributions. A p-value for $H_{0,i}$ is a $[0,1]$-valued random variable $P_i=\mathrm{P}_i(Z)$ satisfying $\mathbb P[P_i\le t]\le t$ for all $t \in (0,1)$ and all $\mathbb P \in \mathcal{P}_i$.  
Compound p-values~\citep{armstrong2022false, ignatiadis2025asymptotic} relax the above per-hypothesis conditions to a single condition that averages over the true nulls; we will see below why this notion is central to our methodological development.
\begin{defi}[Compound p-values]
\label{defi:compound_pvalues}
Fix the null hypotheses $(\cP_1,\dots,\cP_n)$ and a set $\cP $ of distributions.
Let $P_1,\dots,P_n$ be nonnegative random variables. We say that 
$P_1,\dots,P_n$ are \emph{compound} p-values for $(\cP_1,\dots,\cP_n)$ under the model $\cP$ if 
$$
\sum_{i : \mathbb P \in \cP_i} \PP{P_i\leq t} \leq nt \qquad \mbox{for all $t\in (0,1)$ and all $\mathbb P\in \mathcal P$.}
$$
We omit ``under $\cP$'' in case $\cP$ is the set of all distributions. 
\end{defi}

A multiple testing procedure $\mathcal{D}$ maps data $Z$ to the indicator vector $\{0,1\}^n$, where $\mathcal{D}_i=1$ indicates that the $i$-th null hypothesis is rejected. The hypotheses rejected by  $\mathcal{D}$ are called discoveries. We are interested in procedures controlling the false discovery rate (FDR) of~\citet{benjamini1995controlling}, which is defined as
$$
\mathrm{FDR}[\mathcal{D}] \equiv \mathrm{FDR}[\mathcal{D}; \cP_1,\ldots,\cP_n; \mathbb P] := \EE[\mathbb P]{ \frac{\sum_{i=1}^n \ind\!\cb{\mathcal{D}_{i} = 1,\, \mathbb P \in \cP_i}}{1 \lor \sum_{i=1}^n \ind\cb{\mathcal{D}_{i} = 1}} },
$$
where we will usually keep the dependence on $\cP_1,\ldots,\cP_n$ and $\mathbb P$ implicit.

The most common procedure for controlling the FDR is the Benjamini-Hochberg (BH) procedure~\citep{benjamini1995controlling}.
\begin{defi}[BH]
\label{defi:pBH}
Let $P_1,\dots, P_n$ be nonnegative random variables. Let $P_{(i)}$ be the $i$-th order statistic of $P_1,\ldots,P_n$, from the smallest to the largest. Fix $\alpha, \tau \in (0,1)$. Then, the BH procedure at level $\alpha$ with $\tau$-censoring\footnote{
The $\tau$-censoring extension of BH has found uses for multiple testing with data-driven weights~\citep{li2019multiple, ignatiadis2021covariate} and for null-proportion adaptivity~\citep{storey2004strong, gao2026minstorey}. Here we require it for a different reason; see below.
} rejects hypothesis $H_{0,i}$ if $P_i$ is among the smallest $i_\alpha^*$ values $P_1,\dots,P_n$, where 
\begin{equation*} 
i_\alpha^*:=\max\left\{i \in [n]\,:\, P_{(i)} \le \p{\frac{i \alpha}{n}}\land \tau \right\},
\end{equation*}     
with the convention $\max(\emptyset) = 0$. We write $\mathrm{BH}(\alpha, \tau)$ for the resulting multiple testing procedure and by default we set $\tau=\alpha$ and write $\mathrm{BH}(\alpha) \equiv \mathrm{BH}(\alpha, \alpha)$.
\end{defi} 
FDR control of BH when its input consists of independent p-values is well-known~\citep{benjamini1995controlling}. Recently,~\citet[Theorem 1]{barber2026false} establish a finite-sample FDR control guarantee for BH when applied to independent compound p-values.
\begin{theo}[\citet{barber2026false}]
Suppose that $P_1,\dots,P_n$ are independent compound p-values for $\mathcal{P}_1,\ldots,\mathcal{P}_n$ under $\mathcal{P}$. Then the BH procedure at level $\alpha$ applied to $P_1,\dots,P_n$ controls the FDR at level $1.93\alpha$, that is,
$$
\mathrm{FDR}[\mathrm{BH}; \cP_1,\ldots,\cP_n; \mathbb P] \leq 1.93 \alpha \quad \text{for all } \mathbb P \in \mathcal{P}.
$$
\label{theo:bh}
\end{theo}

We note that under an asymptotic regime as that of~\citet{storey2004strong} with weak dependence, the BH procedure with compound p-values controls the FDR at level $\alpha$ asymptotically~\citep{armstrong2022false,ignatiadis2025empirical}. Consequently, a practical compromise (adopted, e.g., in the data analysis of~\citet{barber2026false}) is to run BH at level $\alpha$ anyway, treating $1.93\alpha$ as a worst-case bound that is unlikely to be tight.

For finite sample FDR control at level $\alpha$, one can instead run BH at level $\alpha/1.93$. When bounds on the individual null tail probabilities are available,
a result of~\citet*{dohler2018new} provides an alternative to running
BH at level $\alpha/1.93$ that can be more powerful in some settings.
Below we adapt the reformulation of~\citet{durand2019discretefdr} and express it in terms of compound p-values.

The construction starts from test statistics with known null distributions. Let $S_1=S_1(Z),\ldots,S_n=S_n(Z)$, where larger values indicate more evidence against the null, and suppose that the right-tailed distribution functions
$
\xi_i(s):=\mathbb P[S_i\ge s]$ 
for $\mathbb P \in \mathcal{P}_i$
are known and do not depend on the choice of $\mathbb P\in\mathcal P_i$.
Then,~\citet[Proposition 7]{barber2026false} establish that
\begin{equation}
P_i := \frac{1}{n}\sum_{k=1}^n \xi_k(S_i),
\label{eq:tail_averaging} 
\end{equation}
are compound p-values for $\mathcal{P}_1,\ldots,\mathcal{P}_n$. Instead, consider the following more conservative construction for some $\tau \in (0,1)$:
\begin{equation}
P_i^{\mathrm{DDR}} := \frac{1}{n}\sum_{k=1}^n \frac{\xi_k(S_i)}{1-\xi_k(s_{\tau})}, \;\text{ where }\;
s_{\tau} := \inf\cb{s\,:\, \frac{1}{n}\sum_{k=1}^n \frac{\xi_k(s)}{1-\xi_k(s)} \leq \tau}.
\label{eq:compound_pvalues_ddr}
\end{equation}
We allow for $P_i^{\mathrm{DDR}}=\infty$. Since $P_i^{\mathrm{DDR}} \geq P_i$, the $P_i^{\mathrm{DDR}}$ are also compound p-values. The advantage is that now the BH procedure controls the FDR at level $\alpha$ instead of $1.93\alpha$.
\begin{theo}[\citet*{dohler2018new}]
\label{theo:ddr}
Let $P_1^{\mathrm{DDR}},\ldots,P_n^{\mathrm{DDR}}$ be defined as in~\eqref{eq:compound_pvalues_ddr} for some $\tau \in (0,1)$ and suppose they are independent under $\mathcal{P}$. Then the BH procedure at level $\alpha$ with $\tau$-censoring
applied to $P_1^{\mathrm{DDR}},\ldots,P_n^{\mathrm{DDR}}$ controls the FDR at level $\alpha$, that is,
$$
\mathrm{FDR}[\mathrm{BH}(\alpha, \tau); \cP_1,\ldots,\cP_n; \mathbb P] \leq \alpha \quad \text{for all } \mathbb P \in \mathcal{P}.
$$
\end{theo}
To compare Theorem~\ref{theo:bh} and the above result, take $\tau=\alpha$. If $\min_{k}\{1-\xi_k(s_{\alpha})\} \geq 1/1.93$, then $\mathrm{BH}(\alpha)$ applied to $P_1^{\mathrm{DDR}},\ldots,P_n^{\mathrm{DDR}}$ makes at least as many discoveries as $\mathrm{BH}(\alpha/1.93)$ applied to $P_1,\ldots,P_n$.

\citet{durand2019discretefdr} only consider the case $\tau=\alpha$. Allowing $\tau<\alpha$ has been proposed recently in \cite{perier2026confidence} and introduces a useful trade-off: decreasing $\tau$ can make the modified compound p-values ($P_i^{\mathrm{DDR}}$) smaller, but also restricts BH to rejecting hypotheses with $P_i^{\mathrm{DDR}}\leq\tau$.
One possible default is $\tau=\alpha/10$.

\subsection{Procedures based upon compound randomization p-values}

We now specialize the above general framework to the special case wherein $Z= (\boldZ_1,\ldots,\boldZ_n)$, where $\boldZ_i$ takes values in a space $\mathcal{Z}$. 
We assume that under any distribution $\mathbb P \in \mathcal{P}$, the $\boldZ_i$ are independent across $i$ (that is, whenever we write $\mathbb P \in \mathcal{P}$, we assume that $\boldZ_1,\ldots,\boldZ_n$ are independent under $\mathbb P$).
Moreover, we fix $\mathcal{H}$, a compact group acting continuously on $\mathcal{Z}$, and let $\mu_{\mathcal{H}}$ be the Haar measure on $\mathcal{H}$. This measure is right invariant: if $H\sim\mu_{\mathcal{H}}$, then $HG\sim\mu_{\mathcal{H}}$ for every fixed $G\in\mathcal{H}$.

We specify the $i$-th null hypothesis as
\begin{equation}
\mathcal{P}_i\equiv \mathcal{P}_i(\mathcal{H}):= \{\mathbb P \in \mathcal{P}\,:\, \boldZ_i \stackrel{\mathcal{D}}{=} H \boldZ_i \,\text{ for all } \, H \in \mathcal{H}\}.
\label{eq:group_null}
\end{equation}
Let $\mathcal{O}_i := \mathcal{H} \cdot \boldZ_i = \cb{H\boldZ_i : H \in \mathcal{H}}$ be the group orbit of $\boldZ_i$ and $\mathcal{O}_{1:n} = (\mathcal{O}_1,\ldots,\mathcal{O}_n)$.

The main proposal in this paper is the following three-step procedure:
\begin{mdframed}[nobreak=true]
\refstepcounter{procedure}\label{proc:compound_pvalues}
\textbf{Procedure \theprocedure\ (Compound BH: p-value formulation).}

\begin{enumerate}
\item Use $\mathcal{O}_{1:n}$ to learn a score function $\widehat{S}(\cdot)$ that takes the $\boldZ_i$ as input.
\item Compute compound randomization p-values,
\begin{equation}
\begin{aligned}
&\pcomp_i := \pcompfun(\boldZ_i; \widehat{S}(\cdot), \mathcal{O}_{1:n}, \mathcal{H}),\, \text{ where } \\ 
&\pcompfun(\boldz; \widehat{S}, \mathcal{O}_{1:n}, \mathcal{H}) := \frac{1}{n} \sum_{k=1}^n \int_{\mathcal{H}} \ind\cb{ \widehat{S}(H \boldZ_k) \geq \widehat{S}(\boldz)} \dd\mu_{\mathcal{H}}(H).
\end{aligned}
\label{eq:compound_pvalue_fun}
\end{equation}
\item Apply the $\mathrm{BH}(\alpha)$ procedure to $\pcomp_1,\ldots,\pcomp_n$.
\end{enumerate}
\end{mdframed}
When the group is the group of permutations and $\widehat{S}(\cdot) \equiv S(\cdot)$ is a fixed score, then the above procedure reduces to one of the proposals of~\citet{barber2026false}. 
For limma (Section~\ref{sec:compound_limma}), this lets us borrow information across genes to estimate their variances, without losing the finite sample FDR guarantee.

Our main result is as follows.
\begin{theo}
Suppose we are testing the nulls 
$\mathcal{P}_1(\mathcal{H}),\ldots,\mathcal{P}_n(\mathcal{H})$
in~\eqref{eq:group_null}. Also suppose that $\mathcal{P}$ is such that $\boldZ_i$ are jointly independent across $i$. 
Then conditional on the orbits $\mathcal{O}_{1:n}$, the compound randomization p-values $\pcomp_1,\ldots,\pcomp_n$ in~\eqref{eq:compound_pvalue_fun} are jointly independent compound p-values for $\mathcal{P}_1(\mathcal{H}),\ldots,\mathcal{P}_n(\mathcal{H})$ under $\mathcal{P}$. In particular, the BH procedure at level $\alpha$ applied to $\pcomp_1,\ldots,\pcomp_n$ controls the FDR at level $1.93\alpha$, that is,
$$
\mathrm{FDR}[\mathrm{BH}(\alpha); \cP_1(\mathcal{H}),\ldots,\cP_n(\mathcal{H}); \mathbb P] \leq 1.93 \alpha \quad \text{for all } \mathbb P \in \mathcal{P}.
$$
\label{theo:compound-p}
\end{theo}

We can rewrite Procedure~\ref{proc:compound_pvalues} equivalently in terms of a score threshold, see Procedure~\ref{proc:compound_threshold} in Supplement~\ref{subsec:compound_threshold}. In view of the latter representation (and using the group $\mathcal{H} = \{-1,+1\}^K$ which we discuss in more detail in Section~\ref{sec:compound_limma}), Theorem~\ref{theo:compound-p} also establishes FDR control for the SAM procedure in~\eqref{eq:samfdr} as long as $s_0$ is estimated from the orbits.

We can also apply the DDR construction in~\eqref{eq:compound_pvalues_ddr} conditional on the orbits. For a fixed $\tau\in(0,1)$, this gives the following procedure.

\begin{mdframed}[nobreak=true]
\refstepcounter{procedure}\label{proc:compound_ddr}
\textbf{Procedure \theprocedure\ (DDR compound BH).}

\begin{enumerate}
\item Use $\mathcal{O}_{1:n}$ to learn a score function $\widehat{S}(\cdot)$ that takes the $\boldZ_i$ as input.
\item Compute the modified compound randomization p-values,
\begin{equation}
P_i^{\mathrm{DDR}} := \frac{1}{n}\sum_{k=1}^n
\frac{\int_{\mathcal H}\ind\cb{\widehat{S}(H\boldZ_k)\geq\widehat{S}(\boldZ_i)}\dd\mu_{\mathcal H}(H)}
{1-\int_{\mathcal H}\ind\cb{\widehat{S}(H\boldZ_k)\geq s_\tau}\dd\mu_{\mathcal H}(H)},
\label{eq:compound_randomization_ddr}
\end{equation}
where
$$
s_\tau:=\inf\Bigg\{s:\frac{1}{n}\sum_{k=1}^n
\frac{\int_{\mathcal H}\ind\cb{\widehat{S}(H\boldZ_k)\geq s}\dd\mu_{\mathcal H}(H)}
{1-\int_{\mathcal H}\ind\cb{\widehat{S}(H\boldZ_k)\geq s}\dd\mu_{\mathcal H}(H)}\leq\tau\Bigg\}.
$$
\item Apply the $\mathrm{BH}(\alpha,\tau)$ procedure to $P_1^{\mathrm{DDR}},\ldots,P_n^{\mathrm{DDR}}$.
\end{enumerate}
\end{mdframed}

Applying Theorem~\ref{theo:ddr} conditional on the orbits gives the following result.
\begin{theo}
\label{theo:compound-ddr}
Under the assumptions of Theorem~\ref{theo:compound-p}, Procedure~\ref{proc:compound_ddr} with any fixed $\alpha,\tau\in(0,1)$ controls the FDR at level $\alpha$.
\end{theo}

\subsection{Separate, per-hypothesis randomization p-values}

It is useful to contrast the compound randomization p-values in~\eqref{eq:compound_pvalue_fun} with the following separate, per-hypothesis randomization p-values (see~\citet{ritzwoller2026randomization} for a recent review):
\begin{equation}
\begin{aligned}
&\psep_i := \psepfun(\boldZ_i; \widehat{S}(\cdot), \mathcal{O}_i, \mathcal{H}),\, \text{ where } \\ 
&\psepfun(\boldz; \widehat{S},\mathcal{O}_i, \mathcal{H}) := \int_{\mathcal{H}} \ind\cb{ \widehat{S}(H \boldZ_i) \geq \widehat{S}(\boldz)} \dd\mu_{\mathcal{H}}(H).
\end{aligned}
\label{eq:sep_pvalue_fun}
\end{equation}
Note that under the null hypothesis in~\eqref{eq:group_null}, we have that,
$$
\PP{\psep_i \leq t \mid \mathcal{O}_i} \leq t \quad \text{for all } t \in (0,1) \text{ and all } \mathbb P \in \mathcal{P}_i(\mathcal{H}),
$$
that is, the $\psep_i$ are valid p-values for each null hypothesis $\mathcal{P}_i(\mathcal{H})$. Conditional on $\mathcal{O}_{1:n}$, the $\psep_i$ are independent across $i$ and remain valid null p-values. Thus, the BH procedure applied to $\psep_1,\ldots,\psep_n$ controls the FDR at level $\alpha$ by conditioning on $\mathcal{O}_{1:n}$, even if the score function $\widehat{S}(\cdot)$ is learned from these orbits.

Compound randomization has two advantages over separate randomization:
\begin{enumerate}[leftmargin=*, labelsep=4pt, itemindent=0pt]
\item \textbf{Resolution:} As noted by~\citet{barber2026false}, compound randomization p-values have better resolution in the sense that if the group $\mathcal{H}$ is finite, then the smallest possible value of $\psep_i$ is $1/|\mathcal{H}|$, whereas the smallest possible value of $\pcomp_i$ is $1/(n|\mathcal{H}|)$; with multiplicity correction this can lead to more discoveries when the number of hypotheses $n$ is large.
\item \textbf{Ranking:} The compound p-values maintain the ranking of the scores $\widehat{S}(\boldZ_i)$, whereas the separate p-values do not. In particular, if \smash{$\widehat{S}(\boldZ_i) \geq \widehat{S}(\boldZ_j)$}, then \smash{$\pcomp_i \leq \pcomp_j$}, but it is possible that \smash{$\psep_i > \psep_j$}. If the score function \smash{$\widehat{S}(\cdot)$} is powerful for detecting non-null hypotheses, then the compound randomization p-values may be substantially more powerful than the separate randomization p-values.
\end{enumerate}

Our theory and simulations will seek to emphasize the above two points.\footnote{The advantage of separate randomization is that it provides individually valid p-values for each $i$.} In discussing SAM's pooling of null distributions as in~\eqref{eq:samfdr},~\citet{storey2003sam} and~\citet{kerr2009comments} also bring up these disadvantages of separate randomization p-values.

\subsection{Selective SeqStep+ for small groups}
\label{subsec:seqstep}

We now briefly discuss the important case wherein the group $\mathcal{H}$ is finite with small cardinality, that is, $\kappa:=|\mathcal{H}|-1 < \infty$ and $\mathcal{H}=\{\mathrm{id}, H_1,\ldots, H_{\kappa}\}$. In that case, it is clear that $P_i^{\mathrm{sep}}$ in~\eqref{eq:sep_pvalue_fun} must be $\geq 1/(\kappa+1)$, and so BH at level  $\alpha < 1/(\kappa+1)$ cannot make any discoveries. 
It turns out that even with $\pcomp_i$ (which could be as small as $1/\{n(\kappa+1)\}$), BH cannot make any discoveries if $\alpha < 1/(\kappa+1)$. 

Although this is an important caveat of compound BH, it does not contradict
our main claim that compound randomization can substantially improve power
over separate randomization. First, one of our main examples, orthogonal rotations, uses an infinite group. Moreover, for finite groups, compound BH can be very useful for increasing power once $\alpha$ exceeds $1/(\kappa+1)$. In particular,
Theorem~\ref{theo:sparse_regime}(iv) shows that compound BH with sign flips
asymptotically matches the power of the oracle local false discovery rate
procedure at a reduced testing level (with the reduction becoming exponentially small as the number of replicates $K$ increases), while separate sign flip BH has asymptotically zero power.

\begin{prop}
Suppose $\kappa:=|\mathcal{H}|-1 < \infty$ and that $\alpha < 1/(\kappa+1)$. Then both Procedures~\ref{proc:compound_pvalues} and~\ref{proc:compound_ddr} at level $\alpha$ make zero discoveries.
\label{prop:no_discoveries}
\end{prop}

\begin{proof}
We show the result for Procedure~\ref{proc:compound_pvalues} which always makes at least as many discoveries as Procedure~\ref{proc:compound_ddr}.
Let us write $\pcomp_{(1)} \leq  \ldots  \leq \pcomp_{(n)}$ for the order statistics of $\pcomp_1, \ldots, \pcomp_n$. The self-comparison inherent in~\eqref{eq:compound_pvalue_fun} (i.e., the fact that for each $i$ we also include $H=\mathrm{id}$ in the comparison) implies that:
$$
\pcomp_{(i)} \geq \frac{i}{n(\kappa+1)} > \frac{i \alpha}{n}, 
$$
where the second strict inequality follows from the assumption that  $\alpha < 1/(\kappa+1)$. Thus \smash{$\pcomp_{(i)}$} all lie above the critical values of the BH procedure, and so BH makes no discoveries whatsoever.
\end{proof}
Given the above, for situations with small $\kappa=|\mathcal{H}|-1$ we propose an alternative procedure that is based on the Selective SeqStep+ procedure of~\citet{barber2015controlling}. The formulation is analogous to the multi-knockoff construction of~\citet{gimenez2019improving} with the maximum contrast score of~\citet{ge2021clipper}.   As with compound BH, we allow the score to be learned from the orbits $\mathcal{O}_{1:n}$. This will allow us to use moderated t-statistics that borrow information about the variances across hypotheses, while retaining finite sample FDR control.

\begin{mdframed}[nobreak=true]
\refstepcounter{procedure}\label{proc:seqstep_multi}
\textbf{Procedure \theprocedure\ (Selective SeqStep+ for a finite group of size $\kappa+1$).}
\begin{enumerate}
\item Use $\mathcal{O}_{1:n}$ to learn a nonnegative score function $\widehat{S}(\cdot)$ that takes the $\boldZ_i$ as input.
\item Compute the signed scores
$$
\tilde{S}_i:=\max_{H\in\mathcal{H}}\widehat{S}(H\boldZ_i)\cdot\mathrm{sign}\Big(\widehat{S}(\boldZ_i)-\max_{H\in\mathcal{H}\setminus\{\mathrm{id}\}}\widehat{S}(H\boldZ_i)\Big),
$$
with the convention $\mathrm{sign}(0)=0$.
For each $s>0$, let
\begin{equation}
\widehat{\mathrm{FDR}}(s):=
\frac{\kappa^{-1}+\kappa^{-1}\sum_{i=1}^n\ind\{\tilde{S}_i\leq-s\}}
{1\lor\sum_{i=1}^n\ind\{\tilde{S}_i\geq s\}},
\label{eq:fdr_seqstep_multi}
\end{equation}
and compute the threshold $\hat{s}:=\inf\!\big\{s\in\big\{|\tilde{S}_1|,\ldots,|\tilde{S}_n|\big\}\setminus\{0\}\,:\,\widehat{\mathrm{FDR}}(s)\leq\alpha\big\}$.
\item Reject all $i$ with $\tilde{S}_i\geq\hat{s}$.
\end{enumerate}
\end{mdframed}
Among hypotheses $i$ that could be rejected, $|\tilde{S}_i|$ depends only on $\mathcal{O}_{1:n}$. This gives the following guarantee.
\begin{theo}
\label{theo:seqstep_multi}
Suppose $\kappa:=|\mathcal{H}|-1 < \infty$.
Conditional on the orbits $\mathcal{O}_{1:n}$, the $\tilde S_i$ are jointly independent and satisfy $\mathbb P[\tilde S_i>0\mid\mathcal{O}_{1:n}]\leq1/(\kappa+1)$ for all null $i$ with $\mathbb P\in\mathcal{P}_i(\mathcal{H})$. Moreover, Procedure~\ref{proc:seqstep_multi} controls the FDR at level $\alpha$.
\end{theo}
We note that for $\kappa=1$, Procedure~\ref{proc:seqstep_multi} reduces to Procedure~\ref{proc:seqstep_order_two} in Supplement~\ref{subsec:seqstep_order_two}.

An interesting case emerges when one actually starts with a larger group $\mathcal{H}$ (with $|\mathcal{H}| \gg 2$), and then runs the above procedure with the subgroup $\{\mathrm{id}, H\}$ for a chosen $H \in \mathcal{H}$ with $H \neq \mathrm{id}$ and $H^2 = \mathrm{id}$.\footnote{In principle, one could also reduce to some other small $\kappa > 1$ and run Procedure~\ref{proc:seqstep_multi} with $\{\mathrm{id}, H_1,\ldots,H_{\kappa}\}$. Moreover, if $H_1,\ldots,H_{\kappa}$ are chosen from the larger group at random (in a suitable way), then finite-sample FDR control will persist even if they do not form a group alongside $\mathrm{id}$~\citep{hemerik2018exact}.
} 
Such an approach provides an alternative to our main proposal of using compound randomization p-values alongside BH. Moreover, as we explain in Supplement~\ref{subsec:ress_and_SENS}, the SENS method of~\citet{tian2025conformalized} and the earlier RESS method of~\citet{zou2020new} can be interpreted in exactly that way for certain choices of score $\widehat{S}(\cdot)$ (different from the score we advocate for). Briefly, disadvantages of such an approach compared to compound BH are the following:
\begin{enumerate}[leftmargin=*, labelsep=4pt, itemindent=0pt]
\item \textbf{Variability: }The rejection set crucially depends on the (often arbitrary) choice of $H \in \mathcal{H}$. Rerunning with $H' \neq H$ may lead to a very different set of discoveries; see our proteomics application in Section~\ref{subsec:yves} for an example where this occurs.\footnote{Meanwhile, attempting to derandomize over several choices of $H$ is possible via the techniques of~\citet{ren2024derandomised}, but may lead to a decrease in power.} This point is most concerning to us; and so in general our recommendation is to run compound BH by default, and only use Procedure~\ref{proc:seqstep_multi} when $|\mathcal{H}|$ is small to begin with. Our methylation application in Section~\ref{subsec:methylation} illustrates such a scenario with small $|\mathcal{H}|$.
\item \textbf{$+\kappa^{-1}$ correction:} The $+\kappa^{-1}$ in~\eqref{eq:fdr_seqstep_multi} means that the procedure can never make a nonzero number of discoveries below $1/(\kappa\alpha)$. This restriction is particularly severe for $\kappa=1$ and is a common criticism of Selective SeqStep+ (e.g., for $\alpha=0.1$, it can make either $0$ discoveries, or $10$ or more, but never in between). By comparison, Proposition~\ref{prop:no_discoveries} shows that compound BH applied to the same group cannot make any discoveries at all when $\alpha<1/(\kappa+1)$, i.e., when $\alpha < 1/2$ for $\kappa=1$.
\end{enumerate}

In what follows, we also construct a SeqStep+ procedure with a learned moderated t-statistic score, borrowing information across hypotheses about shared properties of the null distributions.
In Theorem~\ref{theo:sparse_regime}, we show that this procedure asymptotically matches the power of the oracle local false
discovery rate procedure in our sparse regime.
(For comparison, SENS, which we discussed above as another instance of Selective SeqStep+ for this problem, has substantially lower power than our proposed methods in our simulations. In Supplement~\ref{subsec:sens_limma}, we show that even an oracle version of SENS has asymptotically negligible power in our setting.)

\section{Compound randomization with moderated t-statistics}
\label{sec:compound_limma}
We are ready to instantiate our framework to provide a solution to the multiple testing problem in~\eqref{eq:normal_samples}.
To do so, we must first choose the group $\mathcal{H}$, and second, the data-driven score function $\widehat{S}(\cdot)$. We will consider three choices of the group $\mathcal{H}$, and we will choose the score function to be the absolute value of the moderated t-statistic in~\eqref{eq:moderated_stats}, with hyperparameters $\nu_0$ and $s_0^2$ estimated from the orbits.

\subsection{Three choices of the invariance assumption and group \texorpdfstring{$\mathcal{H}$}{H}}

Let us write $\boldZ_i = (Z_{i1},\ldots,Z_{iK})^\top \in \mathcal{Z}:=\RR^K$ for the $K$ replicates of the $i$-th unit in~\eqref{eq:normal_samples}. We can write the $i$-th null hypothesis as
\begin{equation}
\mathcal{P}_i^{\mathrm{N}} := \cb{\mathbb P \in \mathcal{P}\, :\, \boldZ_i \sim \mathrm{N}(0, \sigma_i^2 I_K) \text{ for some } \sigma_i^2>0}.
\label{eq:normalnull}
\end{equation}
We will embed this null hypothesis in a broader null hypothesis that is invariant under a compact group $\mathcal{H}$ of transformations. We will consider three choices of $\mathcal{H}$. 

\paragraph{Sign flips.}  As our first group, we take
$$\mathcal{H}_{\mathrm{symm}}:=  \cb{\pm 1}^K,$$
which acts on $\mathcal{Z}$ by componentwise multiplication.
In this case, the broader hypothesis we are testing is whether the null distribution of $\boldZ_i$ is invariant under coordinatewise sign flips. This holds, for instance, when the replicates are independent and symmetric about $0$ under the null. We may identify the orbit $\mathcal{O}_i$ with $(|Z_{i1}|,\ldots, |Z_{iK}|)$. In this case,~\eqref{eq:compound_pvalue_fun} reduces to the following discrete sum:
\begin{equation}
\pcompfun(\boldz; \widehat{S}(\cdot), \mathcal{O}_{1:n}, \mathcal{H}_{\mathrm{symm}}) := \frac{1}{n 2^K} \sum_{i=1}^n \sum_{H \in \cb{\pm 1}^K} \ind\cb{ \widehat{S}(H \boldZ_i) \geq \widehat{S}(\boldz)}.
\label{eq:compound_pvalue_fun_sign_flip}
\end{equation}
Following notation from~\eqref{eq:group_null}, the above yield compound p-values for the null hypotheses $\mathcal{P}_i(\mathcal{H}_{\mathrm{symm}})$.

\paragraph{Orthogonal rotations.}
As our second group, we take
$$\mathcal{H}_{\mathrm{orth}} := O(K),$$ 
the orthogonal group in dimension $K$, which acts on $\mathcal{Z}$ by matrix multiplication. In this case, the broader hypothesis we are testing is whether the null distribution of each $\boldZ_i$ is rotationally invariant.
Here we may identify the orbit $\mathcal{O}_i$ with $\Norm{\boldZ_i}_2$. Again recalling the notation from~\eqref{eq:group_null}, the above yield compound p-values for the null hypotheses $\mathcal{P}_i(\mathcal{H}_{\mathrm{orth}})$.

We note that we have the following inclusions: 
\begin{equation}
\mathcal{P}_i^{\mathrm{N}} \subseteq \mathcal{P}_i(\mathcal{H}_{\mathrm{orth}}) \subseteq \mathcal{P}_i(\mathcal{H}_{\mathrm{symm}}),
\label{eq:null_inclusions}
\end{equation}
and this means that the first choice of group $\mathcal{H}_{\mathrm{symm}}$ provides a stronger guarantee for FDR control (since it requires the fewest assumptions). 

\paragraph{Order-two subgroup.}   Following the strategy discussed in Section~\ref{subsec:seqstep} of choosing an order-two subgroup, a natural (but arbitrary) choice for either $\mathcal{H}_{\mathrm{orth}}$ or $\mathcal{H}_{\mathrm{symm}}$ is $ \{\mathbf{1},H_{\mathrm{half}}\}$, where\footnote{For $\mathcal{H}_{\mathrm{orth}}$, we identify $H_{\mathrm{half}}$ with the corresponding diagonal matrix and $\mathbf{1}$ with the identity matrix.}
\begin{equation}
H_{\mathrm{half}}:=(\underbrace{1,\ldots,1}_{\lceil K/2\rceil},\underbrace{-1,\ldots,-1}_{\lfloor K/2\rfloor}).
\label{eq:h_half}
\end{equation}
Thus we flip half the coordinates when $K$ is even, and set the middle sign to $+1$ when $K$ is odd. The resulting method depends on the arbitrary indexing of $Z_{i1},\ldots,Z_{iK}$.

\subsection{The limma score function} 

Our next task is to construct the score $\widehat{S}(\cdot)$ using the absolute moderated t-statistic in~\eqref{eq:moderated_stats}, with hyperparameters $\nu_0$ and $s_0^2$ estimated from $\Norm{\boldZ_1}_2,\ldots,\Norm{\boldZ_n}_2$. These norms are measurable with respect to the orbits $\mathcal{O}_{1:n}$ for the sign flip group $\mathcal{H}_{\mathrm{symm}}$, the orthogonal rotation group $\mathcal{H}_{\mathrm{orth}}$ and the order-two group $\cb{\mathbf{1}, H_{\mathrm{half}}}$.  We underline that $\hat{\sigma}^2(\boldz)$ in~\eqref{eq:moderated_stats} is still the usual centered sample variance defined in~\eqref{eq:summary_statistics}, which is unbiased for $\sigma_i^2$ even for the non-nulls. Only the hyperparameters $\nu_0$ and $s_0^2$ will be estimated from the orbits, which may introduce some bias into their estimates, as we will see momentarily.

We briefly recall the limma procedure, which estimates the hyperparameters in~\eqref{eq:limma_variance_prior} as
$$
(\hat{\nu}_0^{\mathrm{limma}}, \hat{s}_0^{2,\mathrm{limma}}) := h^{\mathrm{limma}}(\hat{\sigma}_1^2,\ldots,\hat{\sigma}_n^2; \nu)
$$
for a function $h^{\mathrm{limma}}$ that depends on all the sample variances and the degrees of freedom $\nu = K-1$. The choice of $h^{\mathrm{limma}}$ can for instance be the one given by the original limma paper~\citep{smyth2004linear} that applies the method of moments on $\log(\hat{\sigma}_i^2)$, or the more robust Winsorized method of moments by~\citet{phipson2016robust}. Such estimates of $\nu_0$ and $s_0^2$
cannot be used to construct our score function \smash{$\widehat{S}(\cdot)$}, because they do not only depend on the orbits $\mathcal{O}_i$. Instead, we will construct estimates that depend only on the orbits $\mathcal{O}_{1:n}$. Our idea is as follows. Let us define
\begin{equation}
\hat{\tau}^2(\boldz) := \frac{1}{K} \Norm{\boldz}_2^2,
\label{eq:hat_tau_squared}
\end{equation}
and $\hat\tau_i^2 := \hat{\tau}^2(\boldZ_i)$. Note that $\hat\tau_i^2$ is a function of the orbit $\mathcal{O}_i$ for all three choices of the group $\mathcal{H}$ above. Moreover, letting $\chi^2_K(u)$ be the non-central chi-square distribution with $K$ degrees of freedom and non-centrality parameter $u$, we have that,
\begin{equation}
 \hat{\tau}_i^2 \mid \sigma_i^2, \mu_i \sim \frac{\sigma_i^2}{K}\chi^2_K\p{K \mu_i^2 / \sigma_i^2 },\; \text{ and so, }\; \hat{\tau}_i^2 \mid \sigma_i^2, \mu_i =0 \sim \frac{\sigma_i^2}{K}\chi^2_K.
 \label{eq:hat_tau_i_bn}
\end{equation}
This suggests computing
\begin{equation}
(\hat{\nu}_0(\mathcal{O}_{1:n}), \hat{s}_0^2(\mathcal{O}_{1:n})) := h^{\mathrm{limma}}(\hat{\tau}_1^2,\ldots,\hat{\tau}_n^2; K),
\label{eq:hlimma_suff}
\end{equation}
where it may be preferable to use a robust procedure $h^{\mathrm{limma}}$ such as the one by~\citet{phipson2016robust} to estimate the hyperparameters (to reduce the influence of non-nulls). Here we propose a new procedure for doing so for which we can provide asymptotic guarantees in the sparse regime. (Indeed, in our analysis of the asymptotic power in Section~\ref{sec:power_analysis}, we will use the estimates of $\nu_0$ and $s_0^2$ developed below.)

Let $\widehat Q_n(p)$ denote the empirical $p$-quantile of
$\hat\tau_1^2,\ldots,\hat\tau_n^2$, and let $Q_\nu(p)$ denote the
$p$-quantile of $F_{K,\nu}$, the $F$-distribution with $K$ and $\nu$ degrees of freedom.
Under~\eqref{eq:normal_samples} and~\eqref{eq:lonnsted_prior} and conditional on $\mu_i=0$ (see Proposition~\ref{prop:second_moment_f} in Supplement~\ref{sec:t_formulas} for a proof),
$$
\frac{\hat\tau_i^2}{s_0^2} \mid \mu_i=0\,\sim \, F_{K,\nu_0}.
$$
Hence the ratio of two null quantiles depends only on $\nu_0$. We define
$
R(\nu):=Q_\nu(3/4)/Q_\nu(1/4),
$
and show in Lemma~\ref{lemm:quantile_ratio_monotone} that $R$ is invertible on its range. Thus  we propose to estimate $\nu_0$ by
\begin{equation}
\hat{\nu}_0(\mathcal{O}_{1:n})
:=
R^{-1}\p{
\widehat Q_n(3/4)\big / 
{\widehat Q_n(1/4)}
}.
\label{eq:nu_0_quantile}
\end{equation}
We then estimate the scale from the median,
\begin{equation}
\hat s_0^2(\mathcal{O}_{1:n})
:=
\widehat Q_n(1/2)\big / 
Q_{\hat\nu_0(\mathcal{O}_{1:n})}(1/2).
\label{eq:s_0_median}
\end{equation}
Finally, our score is the absolute value of the estimated moderated t-statistic:
\begin{equation}
\widehat{S}(\boldz) := \abs{T(\boldz; \hat{\nu}_0(\mathcal{O}_{1:n}), \hat{s}_0(\mathcal{O}_{1:n}))}.
\label{eq:moderated_t_score}
\end{equation}
We emphasize that $T(\cdot; \hat{\nu}_0, \hat{s}_0)$ is still the moderated t-statistic in~\eqref{eq:moderated_stats}, whose denominator uses the centered sample variance $\hat{\sigma}^2(\boldz)$. The quantities $\hat{\tau}_i^2$ are used only to estimate the hyperparameters $\nu_0, s_0$.

It is useful at this point to contrast the compound randomization p-values $\pcomp_i$ in~\eqref{eq:compound_pvalue_fun} with the separate randomization p-values $\psep_i$ in~\eqref{eq:sep_pvalue_fun} for the limma score function $\widehat{S}(\cdot)$ above.

\begin{prop}
For both $\mathcal{H} \in \{\mathcal{H}_{\mathrm{symm}}, \mathcal{H}_{\mathrm{orth}}\}$ and any $\boldz \in \mathcal{O}_i$, we have that:
$$
\psepfun(\boldz; \abs{T(\cdot; \hat{\nu}_0(\mathcal{O}_{1:n}), \hat{s}_0(\mathcal{O}_{1:n}))}, \mathcal{O}_i, \mathcal{H})  = \psepfun(\boldz; \abs{T(\cdot; 0, 1)}, \mathcal{O}_i, \mathcal{H}).
$$
Moreover, for $\mathcal{H} = \mathcal{H}_{\mathrm{orth}}$, we have that
$$
P_i^{\mathrm{sep}} = \psepfun(\boldZ_i; \abs{T(\cdot; 0, 1)}, \mathcal{O}_i, \mathcal{H}_{\mathrm{orth}}) = 2 \bar{F}_{t,\nu}(|T(\boldZ_i)|) = P_i^t\; \text{ almost surely},
$$
with $P_i^t$ being the usual two-sided t-test p-value for testing $\mu_i=0$ defined after~\eqref{eq:standard_t}.
\label{prop:sep_pvalue_limma}
\end{prop}

In other words, the score used for separate randomization p-values automatically reverts to the usual t-statistic score and it was pointless to learn the score $\widehat{S}(\cdot)$ in~\eqref{eq:moderated_t_score} with data-driven choices of $\hat{\nu}_0$ and $\hat{s}_0^2$. For $\mathcal{H}_{\mathrm{orth}}$ this becomes even more striking, as the separate randomization p-value is exactly the usual two-sided t-test p-value.\footnote{The second part of Proposition~\ref{prop:sep_pvalue_limma}, i.e., the fact that $\psepfun(\boldZ_i; \abs{T(\cdot; 0, 1)}, \mathcal{O}_i, \mathcal{H}_{\mathrm{orth}}) = 2 \bar{F}_{t,\nu}(|T(\boldZ_i)|)$ also appears in~\citet{koning2024more}.} (For $\mathcal{H}_{\mathrm{symm}}$, we may expect the separate randomization p-value to be approximately equal to the usual two-sided t-test p-value for large $K$, see e.g., Example 15.2.4 in~\citet{lehmann2005testing}.)

\begin{rema}[Interpretation as prepivoting limma's p-values]
Recall from Section~\ref{sec:existing}
that limma uses p-values of the form $2\bar{F}_{t,\nu+\hat{\nu}_0}(|T(\boldZ_i; \hat{\nu}_0, \hat{s}_0)|)$. If $\hat{\nu}_0$ and $\hat{s}_0^2$ are computed as in~\eqref{eq:nu_0_quantile} and~\eqref{eq:s_0_median}, then an alternative score would be to use the (negative) p-value itself as a score, i.e., $\widehat{S}(\boldz) := -2\bar{F}_{t,\nu+\hat{\nu}_0(\mathcal{O}_{1:n})}(|T(\boldz; \hat{\nu}_0, \hat{s}_0)|)$. However, due to monotonicity of the survival function $\bar{F}_{t,\nu}$, this is equivalent to using the absolute value of the moderated t-statistic as a score. In other words, our proposal can be interpreted as prepivoting limma's p-values.
\end{rema}

\subsection{Computational considerations}

For $ \mathcal{H}_{\mathrm{symm}}$, we can compute $\pcompfun(\boldz;  \widehat{S}(\cdot), \mathcal{O}_{1:n}, \mathcal{H}_{\mathrm{symm}})$ exactly by enumerating all sign flips in~\eqref{eq:compound_pvalue_fun_sign_flip}. This is tractable for small $K$, say $K \leq 16$ which is the regime in which the proposed methods are most useful. As a small shortcut, we note that by symmetry, we can reduce the number of sign flips to enumerate to $2^{K-1}$, since the sign flips $H$ and $-H$ yield the same value of $\widehat{S}(H\boldZ_i)$. For larger $K$, we could instead work with a subgroup of $\mathcal{H}$~\citep{koning2024more} or a random subset of $\mathcal{H}$~\citep{hemerik2018exact}.

For the orthogonal rotations $\mathcal{H}_{\mathrm{orth}}$, we can compute 
$\pcompfun(\cdot)$ analytically, as the next result shows.

\begin{prop} Let $\bar{F}_{B, 1/2, \nu/2}$ denote the survival function of a $\mathrm{Beta}(1/2, \nu/2)$ random variable. Then for any $\boldz \in \RR^K$,
$$\pcompfun(\boldz; \widehat{S}(\cdot), \mathcal{O}_{1:n}, \mathcal{H}_{\mathrm{orth}}) = \frac{1}{n} \sum_{i=1}^n \bar{F}_{B, 1/2, \nu/2}\cb{ \frac{\widehat{S}^2(\boldz)/\{\nu + \hat{\nu}_0\}}{ 1 +\widehat{S}^2(\boldz)/\{\nu + \hat{\nu}_0\}}\p{ 1 + \frac{\hat{\nu}_0 \hat{s}_0^2}{\Norm{\boldZ_i}^2_2}}}.
$$
\label{prop:orthogonal_rotation_cpvalue}
\end{prop}

\section{Asymptotic power in sparse regime}
\label{sec:power_analysis}
In this section, we will study the asymptotic power of our proposal (as well as some competing methods) in a sparse regime with $\pi_1 := 1-\pi_0 \to 0$.  Specifically, we will consider data-generation governed by~\eqref{eq:normal_samples} and~\eqref{eq:lonnsted_prior}, with the following asymptotic regime:
$$
n\to\infty,\;\;\pi_1 \to 0,\;\;n\pi_1 \to \infty,\;\;\lambda \to \infty,\;\;\,     
K, s_0^2, \nu_0 \text{ fixed}.
$$
As in prior work~\citep{tang2017bayesian} (also see~\citet{bogdan2011asymptotic, neuvial2012false}), we will choose $\lambda$ to diverge at a rate such that asymptotically the problem is non-trivial, i.e., the power of the oracle procedure is not $0$ (impossible to detect any non-nulls) nor $1$ (trivial to detect all non-nulls). To set up this regime, we note that $\mathrm{lfdr}(\boldz)$ in~\eqref{eq:localfdr_smyth} satisfies the following:
\begin{equation}
\inf_{\boldz \in \RR^K} \mathrm{lfdr}(\boldz) = \frac{\pi_0}{\pi_0 + (1-\pi_0) (1+\lambda)^{(\nu_0+\nu)/2}}.
\label{eq:lfdr_infimum}
\end{equation}
For this reason, we will choose $\lambda$ to diverge at a rate such that the above infimum converges to $1/(1+c_\star)$, for some $c_\star \in (0,\infty)$, i.e., we will choose $\lambda$ such that
\begin{equation}
\pi_1\lambda^{(\nu + \nu_0)/2} \to c_{\star}\,  \text{ for some }\, c_\star \in (0,\infty).
\label{eq:lambda_asymptotics}
\end{equation}
Below we analyze the FDR and power of various procedures in this asymptotic regime. Given any procedure $\mathcal{D}_n : \RR^{K \times n} \to \{0,1\}^n$ that takes the data $\boldZ_{1:n}$ as input and outputs a vector of rejections, we define its FDR and power as
\begin{equation}
\thinmuskip=2mu \medmuskip=2mu \thickmuskip=3mu
\mathrm{FDR}_n[\mathcal{D}_n] := \EE{ \frac{\sum_{i=1}^n \ind\!\cb{\mathcal{D}_{n,i} = 1,\, \mu_i = 0}}{1 \lor \sum_{i=1}^n \ind\cb{\mathcal{D}_{n,i} = 1}} },\;
\mathrm{Pow}_n[\mathcal{D}_n] := \EE{ \frac{\sum_{i=1}^n \ind\!\cb{\mathcal{D}_{n,i} = 1,\, \mu_i \neq 0}}{1 \lor \sum_{i=1}^n \ind\cb{\mu_i \neq 0}} },
\label{eq:fdr_power_defs}
\end{equation}
where we make the dependence on $n$ explicit in the notation. 

To gain some intuition for the asymptotic power, we first consider the oracle local false discovery rate procedure $\mathrm{LFDR}^{\mathrm{OR}}(\alpha)$, which knows the hyperparameters in~\eqref{eq:lonnsted_prior}. By virtue of~\eqref{eq:localfdr_smyth}, this is equivalent to thresholding the absolute moderated t-statistic. 
Thus the oracle rejects whenever $|T(\boldZ_i;\nu_0,s_0)|\geq s_{\alpha,n}$, where $s_{\alpha,n}$ is the largest solution of the following equation (setting $s_{\alpha,n} = \infty$ if no solution exists):
\begin{equation}    
 \mathrm{Fdr}_n(s_{\alpha,n})\stackrel{!}{=}\alpha,\, \text{ where }\;  \mathrm{Fdr}_n(s) := \frac{\pi_0\bar F_0(s)}
{\pi_0\bar F_0(s)+\pi_1\bar F_1(s)},
\label{eq:mfdr_equilibrium_rewritten}
\end{equation}
and where \smash{$
\bar F_0(s):=\bar F_{t,\nu+\nu_0}(s)
$}
and
\smash{$
\bar F_1(s):=\bar F_{t,\nu+\nu_0}(s/\{1+\lambda\}^{1/2})
$} are the survival functions of the moderated t-statistic under the null and alternative, respectively.

Notice that since $\pi_1 \to 0$, it must be the case that $\bar{F}_0(s_{\alpha,n})$ is of order $\pi_1$ or smaller. Then, a convenient way to parameterize thresholds is via $\tau>0$ mapped through the function,  
$$
\zeta_n(\tau)
:=
\p{\frac{L_{\nu+\nu_0}}{\tau\pi_1}}^{1/(\nu+\nu_0)},\, \text{ where }\, L_{\nu+\nu_0}:=\lim_{s\to\infty}s^{\nu+\nu_0}\bar F_{t,\nu+\nu_0}(s) = \frac{(\nu+\nu_0)^{\frac{\nu+\nu_0}{2}-1}
\Gamma\bigl(\tfrac{\nu+\nu_0+1}{2}\bigr)}
{\sqrt{\pi}\,\Gamma\bigl(\tfrac{\nu+\nu_0}{2}\bigr)}.
$$
Lemma~\ref{lemm:sparse_limits} collects the following limits: for any fixed $\tau>0$, we have $\bar F_0(\zeta_n(\tau))/\pi_1 \to \tau$ as $n\to\infty$. Moreover, $\bar{F}_1(\zeta_n(\tau)) \to \Psi_1(\tau)$, where
$\Psi_1(\tau):=\bar F_{t,\nu+\nu_0}\p{\{L_{\nu+\nu_0}/(\tau c_\star)\}^{1/(\nu+\nu_0)}}.$
These definitions suggest that the threshold $s_{\alpha,n}$ in~\eqref{eq:mfdr_equilibrium_rewritten} can be parameterized as $s_{\alpha,n} \approx \zeta_n(\tau_\alpha)$ for some $\tau_\alpha>0$ that solves
\begin{equation} 
\mathrm{Fdr}_{\infty}(\tau_{\alpha}) \stackrel{!}{=} \alpha,\, \text{ where }\; \mathrm{Fdr}_\infty(\tau):=\frac{\tau}{\tau+\Psi_1(\tau)}.
\label{eq:mfdr_equilibrium_asymptotic}
\end{equation}
The following result shows that indeed the solution of equation~\eqref{eq:mfdr_equilibrium_asymptotic} and the threshold function $\zeta_n(\cdot)$ govern the asymptotic behavior of $s_{\alpha,n}$.
\begin{prop}
Fix $\alpha \in (0,1)$ and suppose that the asymptotic regime described in this section holds.
If $c_\star>(1/\alpha)-1$, then there exists a unique $\tau_\alpha>0$ that solves~\eqref{eq:mfdr_equilibrium_asymptotic}. Moreover, for all large enough $n$, there exists a unique $s_{\alpha,n}$ that solves~\eqref{eq:mfdr_equilibrium_rewritten}, and $s_{\alpha,n}/\zeta_n(\tau_\alpha)\to1$ as $n\to\infty$.
\label{prop:asymptotic_threshold}
\end{prop}
To interpret the condition on $c_{\star}$, note that (using~\eqref{eq:lfdr_infimum}),
$$
c_{\star} > \frac{1}{\alpha} -1 \;\;\; \Longleftrightarrow \;\;\;\lim_{n \to \infty} \inf_{\boldz \in \RR^K} \mathrm{lfdr}(\boldz) \equiv \frac{1}{1+c_{\star}}  < \alpha. 
$$
Our main theorem compares the asymptotic FDR and power of the following procedures: 
\begin{enumerate}[label=(\roman*), wide]
\item $\mathrm{LFDR}^{\mathrm{OR}}(\alpha)$: the oracle q-value procedure defined in~\eqref{eq:mfdr_equilibrium_rewritten}.
\item $\mathrm{cBH}(\alpha,\mathcal{H}_{\mathrm{orth}})$: BH at level $\alpha$ applied to the compound orthogonal rotation p-values, using the absolute moderated t-statistic with $\hat{\nu}_0$ and $\hat{s}_0^2$ defined in~\eqref{eq:nu_0_quantile} and~\eqref{eq:s_0_median}.
\item $\mathrm{DDR}(\alpha, \tau_n, \mathcal{H}_{\mathrm{orth}})$: as above but with the DDR modification in Procedure~\ref{proc:compound_ddr} and BH applied at level $\alpha$ and with $\tau_n$-censoring (where our theory will take $\tau_n \to 0$ slowly).
\item $\mathrm{cBH}(\alpha,\mathcal{H}_{\mathrm{symm}})$: BH at level $\alpha$ applied to the compound sign flip p-values, using the score in (ii).
\item $\mathrm{SeqStep+}(\alpha)$: Procedure~\ref{proc:seqstep_multi} with the order-two group $\{\mathbf{1}, H_{\mathrm{half}}\}$ using the score in (ii);
\item $t\textnormal{-}\mathrm{BH}(\alpha)$: BH at level $\alpha$ applied to the usual two-sided t-test p-values.
\item $\mathrm{sBH}(\alpha,\mathcal{H}_{\mathrm{symm}})$: BH at level $\alpha$ applied to the separate sign flip p-values, using the score in (ii).\footnote{ Note that by Proposition~\ref{prop:sep_pvalue_limma},  $\mathrm{sBH}(\alpha,\mathcal{H}_{\mathrm{orth}})$ is identical to $t\textnormal{-}\mathrm{BH}(\alpha)$, so we do not include it in the list.
}
\end{enumerate}

\begin{theo}
\label{theo:sparse_regime}
Consider the asymptotic regime described above and fix $\alpha \in (0,1)$. The following hold:
\begin{enumerate}[label=(\roman*)]
\item If $c_\star > (1/\alpha)-1$, then
$\mathrm{FDR}_n[\mathrm{LFDR}^{\mathrm{OR}}(\alpha)] \to \alpha$ and
$\mathrm{Pow}_n[\mathrm{LFDR}^{\mathrm{OR}}(\alpha)] \to 2\Psi_1(\tau_{\alpha})$.
\item If $c_\star > (1/\alpha)-1$, then
$\mathrm{FDR}_n[\mathrm{cBH}(\alpha, \mathcal{H}_{\mathrm{orth}})] \to \alpha$ and
$\mathrm{Pow}_n[\mathrm{cBH}(\alpha, \mathcal{H}_{\mathrm{orth}})] \to 2\Psi_1(\tau_{\alpha})$.
\item If $c_\star > (1/\alpha)-1$ and $\tau_n \to 0$, $\pi_1 / \tau_n \to 0$ as $n \to \infty$, then
$\mathrm{FDR}_n[\mathrm{DDR}(\alpha, \tau_n, \mathcal{H}_{\mathrm{orth}})] \to \alpha$ and
$\mathrm{Pow}_n[\mathrm{DDR}(\alpha, \tau_n, \mathcal{H}_{\mathrm{orth}})] \to 2\Psi_1(\tau_{\alpha})$.
\item If $\alpha>2^{-\nu}$ and $c_\star > (1/\alpha_{\nu})-1$, where $\alpha_\nu:=(\alpha-2^{-\nu})/(1-2^{-\nu})$, then
$\mathrm{FDR}_n[\mathrm{cBH}(\alpha, \mathcal{H}_{\mathrm{symm}})] \to \alpha_{\nu}$ and
$\mathrm{Pow}_n[\mathrm{cBH}(\alpha, \mathcal{H}_{\mathrm{symm}})] \to 2\Psi_1(\tau_{\alpha_{\nu}})$.
\item If $c_\star > (1/\alpha)-1$, then
$\mathrm{FDR}_n[\mathrm{SeqStep+}(\alpha)] \to \alpha$ and
$\mathrm{Pow}_n[\mathrm{SeqStep+}(\alpha)] \to 2\Psi_1(\tau_{\alpha})$.
\item For any $c_{\star} \in (0,\infty)$, $\mathrm{FDR}_n[t\textnormal{-}\mathrm{BH}(\alpha)] \to \alpha$ and $\mathrm{Pow}_n[t\textnormal{-}\mathrm{BH}(\alpha)] \to 0$.
\item For any $c_{\star} \in (0,\infty)$, $\mathrm{FDR}_n[\mathrm{sBH}(\alpha, \mathcal{H}_{\mathrm{symm}})] \to 0$ and
$\mathrm{Pow}_n[\mathrm{sBH}(\alpha, \mathcal{H}_{\mathrm{symm}})] \to 0$.
\end{enumerate}
\end{theo}

Let us discuss this result. First, compound BH with orthogonal rotations asymptotically matches the power of the oracle local false discovery rate procedure. Its FDR converges to $\alpha$ (and so in this asymptotic regime, the $1.93$ factor in Theorem~\ref{theo:bh} is not binding). Meanwhile, compound BH with sign flips pays a price which can be attributed to the identity group element and its negative. In particular, both $\mathbf{1}$ and $-\mathbf{1}$ leave the absolute moderated t-statistic unchanged, so that
\begin{equation}
    \begin{aligned}
&\pcompfun(\boldz;\widehat{S}(\cdot),\mathcal{O}_{1:n},\mathcal{H}_{\mathrm{symm}}) \\
&\quad= \frac{2^{-\nu}}{n}\sum_{i=1}^n\ind\cb{\widehat{S}(\boldZ_i)\geq\widehat{S}(\boldz)}
+\frac{1}{n2^K}\sum_{i=1}^n\sum_{H\in\mathcal{H}_{\mathrm{symm}}\setminus\{\mathbf{1},-\mathbf{1}\}}\ind\cb{\widehat{S}(H\boldZ_i)\geq\widehat{S}(\boldz)}.
    \end{aligned}
\label{eq:compound_pvalue_symm_identity}
\end{equation}
Thus every score exceeding $\widehat{S}(\boldz)$ contributes $2^{-\nu}/n$ to the compound p-value, including non-nulls. The condition $\alpha>2^{-\nu}$ in part (iv) is also consistent with Proposition~\ref{prop:no_discoveries}, since here $1/(\kappa+1)=2^{-\nu}$ after reducing the group size by half.\footnote{Since $\widehat S(\boldz)=\widehat S(-\boldz)$, we may
restrict to sign flips whose first coordinate is $+1$, without changing
the p-values. This subgroup has $2^\nu$ elements.}

Let us explain the asymptotic implication of~\eqref{eq:compound_pvalue_symm_identity} informally.
Write $\bar{F}$ for the right-tailed distribution function of $\widehat{S}(\boldZ_i)$ and $\bar{F}_0$ for the right-tailed null distribution function.   By~\eqref{eq:compound_pvalue_symm_identity}, we have that compound BH is using an inflated estimate of $\bar{F}_0$: assuming that \smash{$\widehat{S}(H\boldZ_i)$} behaves like a null score for  all $H \notin \cb{\pm \mathbf{1}}$,\footnote{This is exactly correct only for null $i$.}  this estimate is approximately $2^{-\nu}\bar{F}+(1-2^{-\nu})\bar{F}_0$. Thus the alternatives contaminate the estimated null distribution, and BH rejects any hypothesis with $\widehat{S}(\boldZ_i) \geq  s$ for $s$ that approximately satisfies:
$$
\frac{2^{-\nu} \bar{F}(s) + (1-2^{-\nu}) \bar{F}_0(s)}{\bar{F}(s)} \approx \alpha, \,\text{ and, by rearranging, }\, \frac{\bar{F}_0(s)}{\bar{F}(s)} \approx \frac{\alpha - 2^{-\nu}}{1-2^{-\nu}}= \alpha_{\nu}.
$$
These heuristics can be made precise in the asymptotics whereby $\mathrm{cBH}(\alpha,\mathcal{H}_{\mathrm{symm}})$ has asymptotic FDR equal to $\alpha_{\nu} < \alpha$ and matches the power of the oracle procedure run at level $\alpha_{\nu}$.
Orthogonal rotations do not pay this price, since the two transformations have Haar measure zero in $O(K)$.  
(On the other hand, sign flip invariance is a weaker assumption than orthogonal rotation invariance, so that we may still prefer the former in practice.)

Turning to the other baselines, $\mathrm{sBH}(\alpha, \mathcal{H}_{\mathrm{symm}})$ has asymptotically negligible power. This is purely because of discreteness (the smallest possible separate p-value is $2^{-\nu}$), which implies that for large enough $n$, BH makes exactly $0$ discoveries with high probability; this is established as part of the proof. Meanwhile $t\textnormal{-}\mathrm{BH}(\alpha)$ has also asymptotically negligible power, but for a different reason (note that here $\mathrm{FDR}_n$ converges to $\alpha$, so the probability of making at least one discovery is asymptotically at least $\alpha$): the ranking implied by the two-sided t-statistics is not sufficiently powerful for detecting non-nulls in our regime. In fact, one can show that t-tests with BH would only have non-trivial power if we substantially strengthened the signals by requiring 
$\pi_1 \lambda^{\nu/2} \to c_{\star}$ instead of $\pi_1\lambda^{(\nu + \nu_0)/2} \to c_{\star}$ in~\eqref{eq:lambda_asymptotics}. This shows the benefit of the additional $\nu_0$ degrees of freedom obtained through variance moderation.
Finally, in Supplement~\ref{subsec:sens_limma}, we construct an oracle version of SENS~\citep{tian2025conformalized} that uses the true local false discovery rate conditional on the summary statistic used by SENS and show it also has asymptotically negligible power in our regime.

Theorem~\ref{theo:sparse_regime} illustrates the main claim of our paper: compound randomization can retain the power gains from learning a score across hypotheses, while maintaining validity under the assumed null invariance. In the asymptotic regime of this section, the gains can be substantial, with compound BH attaining non-negligible power (matching the oracle for orthogonal rotation) even though separate randomization and ordinary t-tests have asymptotically zero power. The result also highlights the benefits of compound over separate randomization: both use the same (nearly optimal) score function, but separate randomization does not retain the ranking implied by this score function.

\section{Numerical results}
\label{sec:numerical_results}

We illustrate the performance of our proposed methods in numerical experiments. Throughout, we compare six of the procedures from Theorem~\ref{theo:sparse_regime} (labeled ``Local fdr oracle'', ``Comp. $\mathcal{H}_{\mathrm{orth}}$'', ``Comp. $\mathcal{H}_{\mathrm{symm}}$'', ``Limma SeqStep+'', ``t-test'', and ``Sep. $\mathcal{H}_{\mathrm{symm}}$'', respectively, in the figures). Recall from Proposition~\ref{prop:sep_pvalue_limma} that separate orthogonal rotation p-values coincide with ordinary t-test p-values, so both are represented by ``t-test''. We also consider the plug-in limma method of~\citet{smyth2004linear} (``Limma''). Finally, we compare to the Gaussian and General variants of SENS from~\citet{tian2025conformalized} (``SENS Gaussian'' and ``SENS General''). Both control the FDR under sign flip invariance, but differ in the score function used. In Supplement~\ref{subsec:compound_variants}, we also compare compound BH at level $\alpha$, as used here, with BH at level $\alpha/1.93$ and the DDR modification, for both orthogonal rotations and sign flips, both of which ensure finite-sample FDR control at level $\alpha$.

We generate $n=5{,}000$ independent hypotheses, with $(\mu_i,\sigma_i^2)$ drawn according to~\eqref{eq:lonnsted_prior}, and observe
$$
Z_{ij} = \mu_i + \sigma_i\varepsilon_{ij},\;\;\; j=1,\ldots,K,
$$
where the errors $\varepsilon_{ij}$ are independent of the parameters and of one another, with mean zero and variance one. Throughout, we take $\lambda=10$, $\nu_0=10$, $s_0^2=1$, and nominal FDR level $\alpha=0.1$.  We report estimates of the FDR and power defined in~\eqref{eq:fdr_power_defs} averaged over $500$ Monte Carlo replicates. We consider four settings:
\begin{itemize}[nosep]
\item Gaussian noise, varying $K$: $\varepsilon_{ij} \sim \mathrm{N}(0,1)$  with $\pi_1=0.025$ and $K\in\{3,5,\ldots,13\}$;
\item Gaussian noise, varying $\pi_1$: $\varepsilon_{ij} \sim \mathrm{N}(0,1)$ with $K=5$ and $\pi_1\in\{0.005,0.015,\ldots,0.045\}$;
\item Uniform noise, varying $K$: $\varepsilon_{ij} \sim \mathrm{Un}[-\sqrt{3},\sqrt{3}]$ with $\pi_1=0.025$ and  $K\in\{3,5,\ldots,13\}$;
\item Laplace noise, varying $K$: $\varepsilon_{ij} \sim \mathrm{Laplace}(0,1/\sqrt{2})$ with $\pi_1=0.025$ and $K\in\{3,5,\ldots,13\}$.
\end{itemize}

In Figure~\ref{fig:main_gaussian_AB}, we show the results for the first two settings with Gaussian noise, so that both~\eqref{eq:normal_samples} and~\eqref{eq:lonnsted_prior} hold. In other words, this is the setting studied in Section~\ref{sec:power_analysis}. All methods approximately control the FDR at the nominal level. The local fdr oracle and limma have nearly identical power, higher than the other methods. Compound BH with orthogonal rotations has higher power than compound BH with sign flips, as in Theorem~\ref{theo:sparse_regime}. Ordinary t-tests, SENS and separate BH with sign flips have lower power. In Panel A, we see that the power of all methods increases with $K$. Notice that for $K=3$, $\alpha < 2^{-\nu}$, and compound BH with sign flips has no power, but it has non-negligible power for $K\geq 5$.  Separate sign flips only start to make discoveries at $K=13$, when the smallest p-value is $2^{-12} \approx 0.00024$. Among data-driven methods with finite-sample guarantees, compound BH with orthogonal rotations and Limma SeqStep+ have non-negligible power for $K=3$, with Limma SeqStep+ having higher power. For $K\geq5$, compound BH with orthogonal rotations has similar or higher power than Limma SeqStep+. The gap between compound BH with orthogonal rotations and sign flips becomes smaller as $K$ increases (which is consistent with the asymptotic results in Theorem~\ref{theo:sparse_regime}). Meanwhile, in Panel B, we see that compound BH with orthogonal rotations gets closer to the oracle as $\pi_1$ decreases, consistent with the sparse asymptotic regime of Section~\ref{sec:power_analysis}. Limma SeqStep+ is close to the oracle for the largest $\pi_1$, but has much lower power at $\pi_1=0.005$. This is consistent with the cost of the $+1$ correction discussed in Section~\ref{subsec:seqstep}. Also, for $K=5$ we have $\alpha_\nu=(0.1-2^{-4})/(1-2^{-4})=0.04$, close to the empirical FDR of compound BH with sign flips at $\pi_1=0.025$.

\begin{figure}
\centering
\includegraphics[width=1\textwidth]{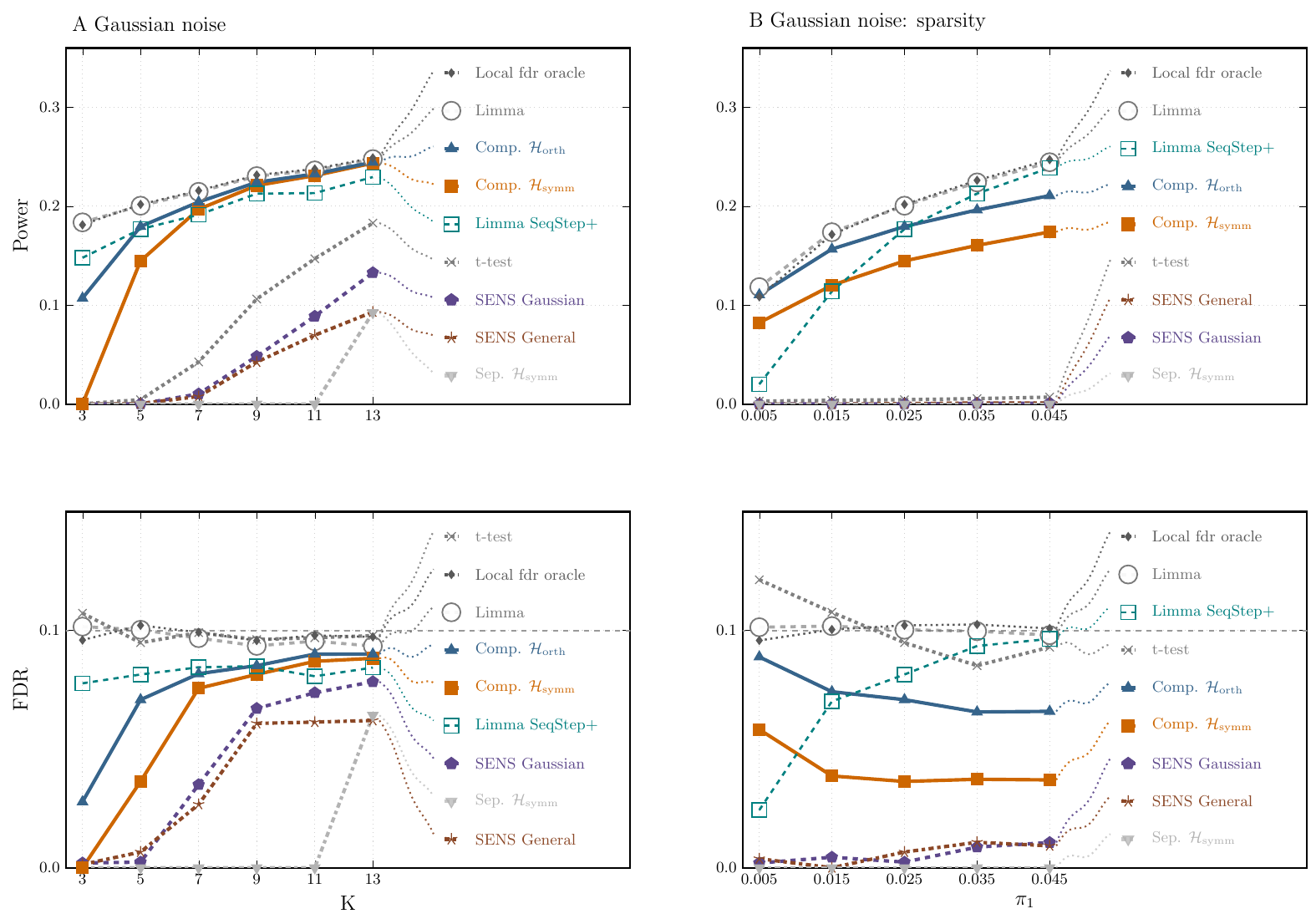}
\caption{Power (top) and FDR (bottom) under Gaussian noise. Panel A varies the number of replicates $K$ with $\pi_1=0.025$; Panel B varies the non-null proportion $\pi_1$ with $K=5$. The nominal FDR is set to $\alpha=0.1$. Our compound randomization procedures (``Comp. $\mathcal{H}_{\mathrm{orth}}$'' and ``Comp. $\mathcal{H}_{\mathrm{symm}}$'') have substantially higher power than separate randomization, ordinary t-tests, and SENS, with orthogonal rotations retaining power even at $K=3$. Their power approaches the oracle as $K$ increases (Panel A), and the gap between compound orthogonal rotations and the oracle narrows as $\pi_1$ decreases (Panel B). Limma SeqStep+ also retains power at $K=3$, but loses power for the smallest $\pi_1$.}
\label{fig:main_gaussian_AB}
\end{figure}

In Figure~\ref{fig:main_robustness_CD}, we show the results for the last two settings, with uniform and Laplace noise. Both distributions have variance one (the second parameter of the Laplace distribution is its scale), so the conditional variance of $Z_{ij}$ remains $\sigma_i^2$. The null distributions are invariant under sign flips, but not under orthogonal rotations. We apply all methods as in Figure~\ref{fig:main_gaussian_AB}. In particular, the local fdr oracle still uses the Gaussian formulas with the true prior hyperparameters, so it is no longer a correctly calibrated oracle (and is labeled ``Gaussian oracle'').
In Panel C, with uniform noise, we see that ordinary t-tests, limma and the Gaussian oracle exceed the nominal FDR level. Compound BH with orthogonal rotations also exceeds the nominal level for $K\geq7$, with FDR around $0.16$ for $K\geq9$. Compound BH with sign flips stays below the nominal level throughout, and has substantially higher power than separate BH with sign flips and SENS for $K\geq5$. Limma SeqStep+ also stays below the nominal level, with higher power than compound sign flips for $K\leq7$ and lower power for $K\geq9$.
Thus the higher power of the Gaussian methods in this panel comes with a loss of FDR control. Recall the inclusions in~\eqref{eq:null_inclusions}. The uniform null belongs to $\mathcal{P}_i(\mathcal{H}_{\mathrm{symm}})$ but not to $\mathcal{P}_i(\mathcal{H}_{\mathrm{orth}})$ or $\mathcal{P}_i^{\mathrm{N}}$. Here, compound rotations incur much less FDR inflation than the Gaussian methods.
Meanwhile, in Panel D, all methods approximately control the FDR empirically, with ordinary t-tests, limma and compound BH with orthogonal rotations being conservative. Limma SeqStep+ has the highest power among data-driven methods for $K\leq7$, while compound BH with sign flips has slightly higher power for $K\geq9$. Both exceed the Gaussian oracle for $K\geq5$. Thus the learned score can remain useful when the working model is incorrect: sign flip calibration retains power and validity.

Overall, the experiments support the main claim of our paper: compound randomization retains the power gains from borrowing information across hypotheses, which separate randomization with the same score does not retain. These gains persist even when the working model used to learn the score is incorrect. Compound sign flips retain their validity under the weaker symmetry assumption, while compound rotations can exceed the nominal FDR level when rotation invariance does not hold.

\begin{figure}
\centering
\includegraphics[width=1\textwidth]{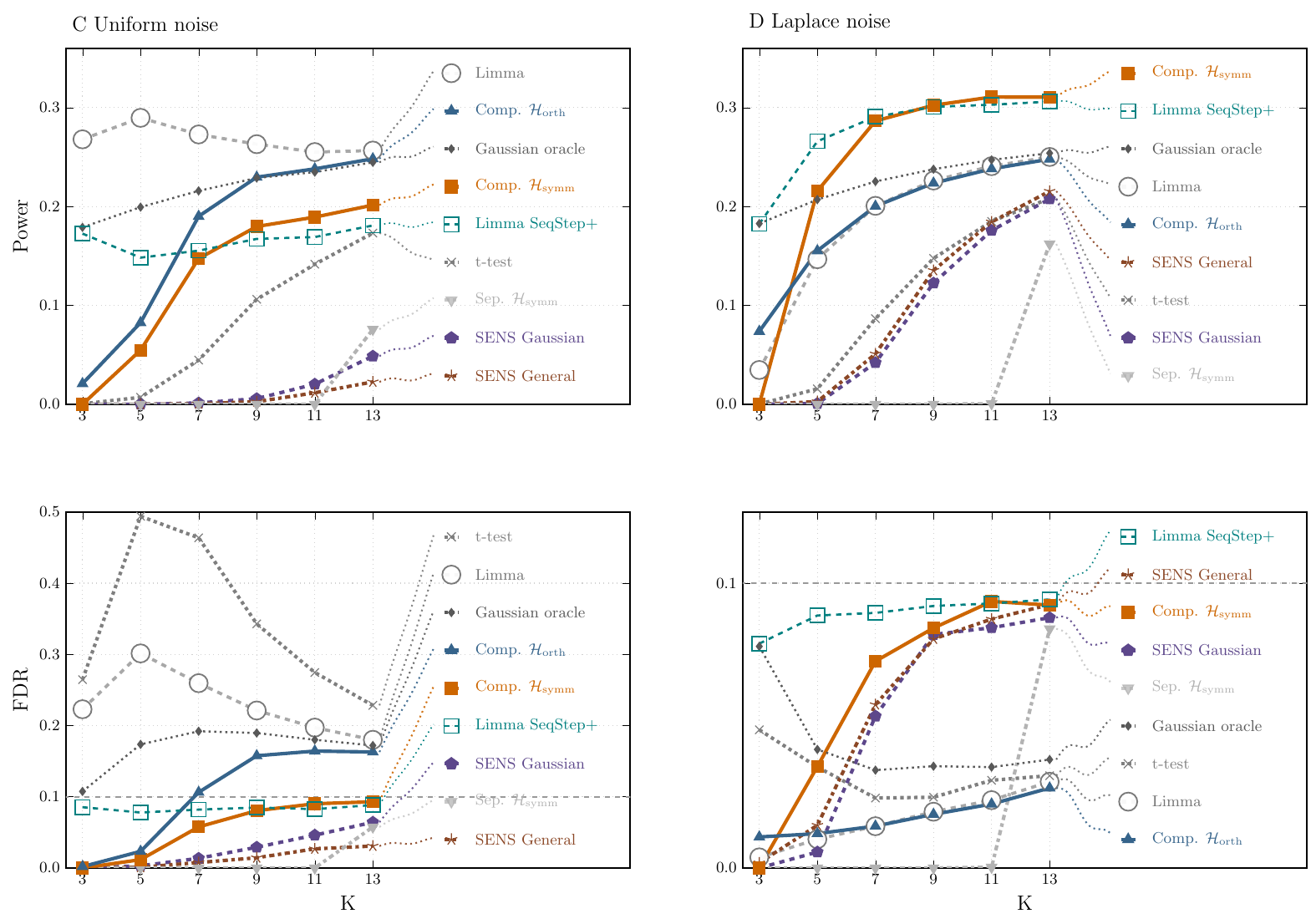}
\caption{Power (top) and FDR (bottom) under uniform noise (Panel C) and Laplace noise (Panel D). We vary the number of replicates $K$ with $\pi_1=0.025$ and nominal FDR level $\alpha=0.1$; the FDR axes have different scales. The curve labeled ``Gaussian oracle'' uses Gaussian formulas with the true prior hyperparameters. In Panel C, several Gaussian-based methods fail to control the FDR. Our compound sign flip procedure (``Comp. $\mathcal{H}_{\mathrm{symm}}$'') and Limma SeqStep+ stay below the nominal level, with higher power than separate sign flips and SENS for $K\geq5$. Our compound rotation procedure (``Comp. $\mathcal{H}_{\mathrm{orth}}$'') has higher power for $K\geq7$, but exceeds the nominal FDR level. In Panel D, Limma SeqStep+ has the highest power among data-driven methods for $K\leq7$, and compound sign flips for $K\geq9$; both stay below the nominal FDR level. Compound rotations also retain substantial power and have conservative FDR.}
\label{fig:main_robustness_CD}
\end{figure}

The simulations have shown the benefit of compound randomization in situations where limma~\citep{smyth2004linear} is relevant: the regime of small $K$ where it is important to share information about the variances. In other words, this is an example of a setting where a good score function can be learned using shared properties of the composite null distributions. In Supplement~\ref{sec:stylized_example}, we also consider a stylized setting where compound randomization can be beneficial even for larger $K$.

\section{Other settings for compound randomization}
\label{sec:other_settings}
So far in this paper we have focused on one-sample testing. Here we show that the proposed framework is much broader by demonstrating how it can be applied to two-sample testing (Section~\ref{subsec:twosample}) and for testing of linear model coefficients (Section~\ref{subsec:linearmodel}).

\subsection{Two-sample testing}
\label{subsec:twosample}

A common task in high-throughput biology is testing for differential expression, i.e., whether the mean expression of each gene differs between two groups, e.g., treatment and control. Suppose that the $K$ observations for each unit come from two groups $A$ and $B$, with sizes $K_A, K_B \geq 1$ and $K_A + K_B = K > 2$. For $j=1,\ldots,K$ and $i=1,\ldots,n$, we consider the Gaussian model
\begin{equation}
Z_{ij} = \theta_i + \delta_i \ind\cb{j \in B} + \varepsilon_{ij},\;\; \varepsilon_{ij} \simindep \mathrm{N}(0, \sigma_i^2).
\label{eq:normal_samples_two_sample}
\end{equation}
Here $\theta_i$ is the mean in group $A$, $\delta_i$ is the difference in means between groups $B$ and $A$, and $\sigma_i^2$ is the common variance within unit $i$, all unknown. The null hypotheses are $H_{0,i}: \delta_i = 0$ for $i=1,\ldots,n$. Writing $\boldZ_i = (Z_{i1},\ldots, Z_{iK})^\top$, the $i$-th null hypothesis takes the form:
$$
\mathcal{P}_i^{\mathrm{N}} := \cb{\mathbb P \in \mathcal{P}\, :\, \boldZ_i \sim \mathrm{N}(\theta_i\mathbf{1},\, \sigma_i^2 I_K) \text{ for some } \theta_i \in \RR,\; \sigma_i^2>0}.
$$
As in the one-sample setting, we embed this null in broader hypotheses $\mathcal{P}_i(\mathcal{H})$ for two choices of the group.

\paragraph{Permutations.}
We take
$$
\mathcal{H}_{\mathrm{perm}} := \mathfrak{S}_K,
$$
the group of permutations of $\{1,\ldots,K\}$, which acts on $\boldz$ by permuting its coordinates. In this case, $\mathcal{P}_i(\mathcal{H}_{\mathrm{perm}})$ is the hypothesis that the $K$ observations are exchangeable. We may identify the orbit $\mathcal{O}_i$ with the empirical distribution $K^{-1}\sum_{j=1}^K\delta_{Z_{ij}}$.

\paragraph{Orthogonal rotations.}
We may instead take
$$
\mathcal{H}_{\mathrm{orth},\mathbf{1}}:=\cb{H\in O(K):H\mathbf{1}=\mathbf{1}},
$$
which fixes the common mean and rotates the centered observations. Here the orbit is determined as $\mathcal{O}_i = (\hat\mu(\boldZ_i),\hat{\sigma}^2(\boldZ_i))$, that is, the sample mean and the sample variance of $\boldZ_i$ computed as in~\eqref{eq:summary_statistics}, ignoring which group each observation belongs to.

We have that,
$$
\mathcal{P}_i^{\mathrm{N}}\subseteq\mathcal{P}_i(\mathcal{H}_{\mathrm{orth},\mathbf{1}})\subseteq\mathcal{P}_i(\mathcal{H}_{\mathrm{perm}}).
$$
We discuss how our methods apply to this setting in more detail in Supplement~\ref{sec:twosample_supp}.

\subsection{Testing for a coefficient in a linear model}
\label{subsec:linearmodel}
Another common question in genomics, which nests both the one- and two-sample problems,
is whether a regression coefficient is equal to $0$. Consider
\begin{equation}
Z_{ij} = W_j\delta_i + X_{j}^\top \theta_i + \varepsilon_{ij},\;\; \varepsilon_{ij} \simiid \mathrm{N}(0, \sigma_i^2),\;\; \text{ for } j=1,\ldots,K.
\label{eq:normal_samples_linear_model}
\end{equation}
Here $W_j\in\RR$ is the covariate of interest and $X_j\in\RR^p$ contains the nuisance covariates for observation $j$, all treated as fixed. Let $W:=(W_1,\ldots,W_K)^\top$ and $X:=(X_1,\ldots,X_K)^\top\in\RR^{K\times p}$, so that the $j$-th row of $X$ is $X_j^\top$.\footnote{
In principle $W$ and $X$ could depend on $i$, that is we could replace $W$ by \smash{$W^i=(W_1^i,\ldots,W_K^i)$} and $X$ by \smash{$X^i = (X_1^i,\ldots,X_K^i)$}. We omit this dependence for notational simplicity and because for applications in high-throughput biology it is very common to have a shared design across units.
} We assume that $(W,X)$ has full column rank and $K>p+1$. We observe independent units and wish to test $\delta_i=0$. We can write the $i$-th null hypothesis as\footnote{The difference to the fixed-$X$ knockoff setting of~\citet{barber2015controlling} is as follows: therein one is faced with a single regression task and seeks to test whether each of many coefficients are equal to zero. Meanwhile, here we are faced with $n$ regression tasks pertaining to distinct units, and test one coefficient per unit.}
$$
\mathcal{P}_i^{\mathrm{N}}:=\cb{\mathbb P\in\mathcal{P}\,:\,\boldZ_i\sim\mathrm{N}(X\theta_i,\sigma_i^2 I_K)\text{ for some }\theta_i\in\RR^p,\;\sigma_i^2>0}.
$$

\paragraph{Orthogonal rotations.}
We embed $\mathcal{P}_i^{\mathrm{N}} \subseteq \mathcal{P}_i(\mathcal{H}_{\mathrm{orth},X})$, where
\begin{equation}
\mathcal{H}_{\mathrm{orth},X} := \cb{H \in O(K) : HX = X}
\label{eq:rotations_within}
\end{equation}
is the subgroup of orthogonal transformations that fix the nuisance design. Writing
$P_X = I_K - UU^\top$ for the projection onto the column space of $X$, where the columns of $U$ form an orthonormal basis
of its orthogonal complement, every
element of $\mathcal{H}_{\mathrm{orth},X}$ can be written as $H = P_X + UQU^\top$ for a 
$Q \in O(K-p)$, so that
$$
H\boldZ_i = P_X\boldZ_i + UQ\,U^\top\boldZ_i.
$$
Thus $H$ leaves the fitted part unchanged and rotates the residual coordinates $U^\top\boldZ_i$.  We may identify the orbit $\mathcal{O}_i$ with $(P_X\boldZ_i,\Norm{U^\top\boldZ_i}_2)$. This construction is used in the rotation tests of~\citet{langsrud2005rotation, wu2010roast, perry2010arotation, solari2014rotationbased}. (In the context of multiple testing, these authors are interested in this construction for dealing with dependence across units, rather than for learning a score function across units and using it to construct compound p-values.)
We discuss how our methods apply to this setting in more detail in Supplement~\ref{sec:linear_model_supp}.

\paragraph{Stratified permutations.} A further approach that makes weaker assumptions is via the stratified permutations of~\citet{dhaultfoeuille2024robust}.  Identifying permutations with their permutation matrices,\footnote{
For $H \in \mathcal{H}_{\mathrm{perm}} \equiv \mathfrak{S}_K$, $X \in \RR^{K \times p}$, then $HX$ is the matrix whose rows have been permuted by $H$.
} we consider the subgroup that fixes $X$:
\begin{equation}
\mathcal{H}_{\mathrm{perm},X}:=\cb{H\in\mathcal{H}_{\mathrm{perm}}:HX=X}.
\label{eq:permutations_within}
\end{equation}
Notice that $\mathcal{P}_i^{\mathrm{N}} \subseteq \mathcal{P}_i(\mathcal{H}_{\mathrm{orth},X}) \subseteq \mathcal{P}_i(\mathcal{H}_{\mathrm{perm},X}).$
To describe these permutations more explicitly, let $x_1,\ldots,x_L$ be the distinct rows of $X$, and let $I_\ell:=\{j:X_j=x_\ell\}$ be the corresponding strata. Then $H$ permutes observations within each $I_\ell$, with no permutations across strata, so that
\smash{$
\mathcal{H}_{\mathrm{perm},X}\cong\prod_{\ell=1}^L\mathfrak{S}_{|I_\ell|}$}. We may identify the orbit $\mathcal{O}_i$ with the collection of empirical distributions \smash{$(|I_\ell|^{-1}\sum_{j\in I_\ell}\delta_{Z_{ij}})_{\ell=1}^L$}. 

Interpreted in terms of~\eqref{eq:normal_samples_linear_model}, the null $\mathcal{P}_i(\mathcal{H}_{\mathrm{perm},X})$ requires that $\delta_i=0$ and that the error terms within each stratum $\{\varepsilon_{ij} \,:\, j \in I_{\ell}\}$ be exchangeable. In particular, we have that
$\mathcal{P}_i^{\mathrm{N}} \subseteq \mathcal{P}_i^{\mathrm{exch}}\subseteq\mathcal{P}_i(\mathcal{H}_{\mathrm{perm},X})$, where
$$
\mathcal{P}_i^{\mathrm{exch}}:=\cb{\mathbb P\in\mathcal{P}:\boldZ_i-X\theta_i\text{ has exchangeable coordinates for some }\theta_i\in\RR^p}.
$$
A disadvantage of using the permutation group in~\eqref{eq:permutations_within}
is that it requires repeated rows in $X$ for the group to be non-trivial and so it essentially requires discrete entries. If such repeated rows are not available, then we can still construct randomization tests for 
$\mathcal{P}_i^{\mathrm{exch}}$ via the approaches of~\citet{lei2021assumptionfree} and~\citet{pouliot2026exact}.

\section{Real data analyses}

\subsection{Differential proteoform abundance after cytokine treatment}
\label{subsec:yves}

We reanalyze the proteomics study of~\citet{ives2025characterization}, who measured proteoform abundances in matched control and cytokine-treated human pancreatic islets from six donors. Treatment consisted of a 24-hour exposure to interferon-$\gamma$ and interleukin-1$\beta$. 
We retain the $n=286$ proteoforms observed in both conditions (control and treatment) for all six donors. For each proteoform $i$, we take $\boldZ_i\in\RR^6$ to be the vector of treated-minus-control pairwise differences of $\log_2$-transformed intensities. This sets us in the one-sample setting.

We seek to control FDR at both  $\alpha=0.05$ and $\alpha=0.1$.
We use the moderated t-statistic score with $(\hat{\nu}_0,\hat{s}_0^2)=(3.863,0.474)$ estimated by~\eqref{eq:nu_0_quantile}--\eqref{eq:s_0_median}; limma gives $(3.882,0.511)$. For SeqStep+, we run Procedure~\ref{proc:seqstep_multi} separately with the group $\{\mathbf{1},H\}$ for all $H\in\{\pm1\}^6$ with $\#\{j:H_j=+1\}=3$, including $H_{\mathrm{half}}$ from~\eqref{eq:h_half}. We count $H$ and $-H$ only once, since they give identical results, yielding ten choices of groups $\{\mathbf{1},H\}$.

\begin{table}
\centering
\caption{Number of discoveries in the proteomics study. Every entry equal to $1$ refers to the same beta-2 microglobulin (B2M) proteoform. For SeqStep+, parentheses give the number of groups (out of 10) yielding each result. }
\label{tab:ives-discoveries}
\small
\setlength{\tabcolsep}{9pt}
\begin{tabular}{@{}lcc@{}}
\toprule
Method & $\alpha=0.05$ & $\alpha=0.1$ \\
\midrule
BH + t-test & $0$ & $0$ \\
BH + separate sign flips & $0$ & $0$ \\
BH + Limma & $1$ & $1$ \\
Compound BH (or DDR with $\tau=\alpha/10$) with rotations & $1$ & $1$ \\
Compound BH (or DDR with $\tau=\alpha/10$) with sign flips & $0$ & $1$ \\
Compound BH at $\alpha/1.93$ with rotations & $0$ & $1$ \\
Compound BH at $\alpha/1.93$ with sign flips & $0$ & $0$\\
Selective SeqStep+ & $0$ (10) & $0$ (9), $24$ (1) \\
\bottomrule
\end{tabular}
\end{table}

\begin{table}
\centering
\caption{Paired differences and unadjusted p-values for the B2M (amino acids 21--119) and GCG (amino acids 92--128) proteoforms. }
\label{tab:ives-measurements}
\small
\setlength{\tabcolsep}{3.3pt}
\begin{tabular}{@{}lrrrrrrrrrr@{}}
\toprule
& \multicolumn{6}{c}{Donor} & \multicolumn{2}{c}{Separate p-values} & \multicolumn{2}{c}{Compound p-values} \\
\cmidrule(lr){2-7}\cmidrule(lr){8-9}\cmidrule(l){10-11}
& 1 & 2 & 3 & 4 & 5 & 6 & t-test & Sign flips & Rotations & Sign flips \\
\midrule
B2M & $2.15$ & $1.70$ & $1.08$ & $2.05$ & $2.48$ & $0.94$ & $9.8\times10^{-4}$ & $\begin{gathered}1/32\\[-2pt]\approx0.031\end{gathered}$ & $1.6\times10^{-4}$ & $\begin{gathered}2/(286\cdot32)\\[-2pt]\approx2.2\times10^{-4}\end{gathered}$ \\
GCG & $2.72$ & $1.64$ & $-1.60$ & $-1.19$ & $3.33$ & $-1.81$ & $0.61$ & $\begin{gathered}20/32\\[-2pt]\approx0.63\end{gathered}$ & $0.51$ & $\begin{gathered}4736/(286\cdot32)\\[-2pt]\approx0.52\end{gathered}$ \\
\bottomrule
\end{tabular}
\end{table}

We report the number of rejections of different methods in Table~\ref{tab:ives-discoveries}. The t-test makes no discoveries. Both limma and compound rotations (with both BH and DDR) each select one proteoform at both $\alpha=0.05$ and $\alpha=0.1$: the same beta-2 microglobulin (B2M) proteoform. Compound sign flips also reject B2M at $\alpha=0.1$. Meanwhile, SeqStep+  makes $0$ discoveries for $\alpha=0.05$ and it makes $0$ discoveries for $9$ of the $10$ groups at $\alpha=0.1$, while it makes $24$ discoveries for the remaining group. We view this as an important disadvantage of SeqStep+ that the number of discoveries can be so sensitive to the arbitrary choice of $H$ (equivalently, to the arbitrary labeling of six donors).

The first row of Table~\ref{tab:ives-measurements} shows $\boldZ_i$ 
for the B2M proteoform, along with the separate and compound p-values. Note that here all six $Z_{ij}$ are positive, which suggests a potentially large positive $\mu_i$. The second row shows another proteoform, namely a fragment of proglucagon (GCG) comprising amino acids 92--128. Here, half the $Z_{ij}$ are positive and half negative, and so GCG is consistent with the null and receives large p-values. The reason we show GCG is that it explains the compound sign flip p-value for B2M: in~\eqref{eq:compound_pvalue_fun_sign_flip}, $\widehat{S}(\boldZ_{\mathrm{B2M}})$ is only equaled or exceeded by itself and by GCG under the sign pattern $H=(+1,+1,-1,-1,+1,-1)$, which makes all entries of $H\boldZ_{\mathrm{GCG}}$ positive. This yields the compound sign flip p-value of $2/(286\cdot32)$ (where we divide by $32$ instead of $2^6=64$ since we are not counting $H$ and $-H$ separately). In the same table we see that the compound rotation p-value for B2M is slightly smaller than its compound sign flip p-value, which one could partially attribute to the latter's self-counting in~\eqref{eq:compound_pvalue_symm_identity}.

\subsection{Differential methylation after T cell activation}
\label{subsec:methylation}
Our next study pertains to DNA methylation in human T cells~\citep{zhang2013genomewide} and was reanalyzed using limma by~\citet{maksimovic2017crosspackage}. The data consist of ten samples from three human donors (M28, M29, M30), with measurements on naive T cells (naive), resting regulatory T cells (rest.\ Treg), and their activated counterparts (act.\ naive and act.\ Treg). Following the preprocessing of~\citet{maksimovic2017crosspackage}, we retain $n=439{,}918$ CpG sites and use the ``methylation M-values'' as the outcomes.

For the $i$-th CpG site, we model its methylation via the additive ANOVA model with
\begin{equation}
\label{eq:overparameterized}
Z_{ij} = \alpha_i(\mathrm{M}(j)) + \beta_i(\mathrm{C}(j)) + \varepsilon_{ij},\;\;j=1,\ldots,10,
\end{equation}
where $\mathrm{M}(j) \in \{\mathrm{M}28, \mathrm{M}29, \mathrm{M}30\}$ denotes the human donor of the $j$-th sample, and $\mathrm{C}(j) \in \{\text{naive, rest.\ Treg, act.\ naive, act.\ Treg}\}$ denotes the cell type and status; and we model both by fixed effects.\footnote{The $7$ parameters $\{\alpha_i(m)\}_m$ and $\{\beta_i(c)\}_c$ are not identified, but contrasts thereof are.
} We are interested in testing whether contrasts of the form $\beta_i(c) - \beta_i(c')$ for $c \neq c'$ are equal to $0$ and we first consider this contrast for $c=\text{``rest.\ Treg''}$ and $c'=\text{``naive''}$. To this end, we  reparameterize~\eqref{eq:overparameterized} as in~\eqref{eq:normal_samples_linear_model} such that $\delta_i = \beta_i(c) - \beta_i(c')$. All ten samples are used for each contrast.
Here $W_j$ is an indicator whether the sample is from resting Treg cells, and $X_j$ contains an intercept and indicators for M29, M30 (where M29 and M30 are two of the three human donors), activated naive, and activated Treg. Figure~\ref{fig:design_swaps_methylation} shows the resulting design matrix $(W,X)$ in green.

\begin{figure}[!tbp]
\centering
\resizebox{\linewidth}{!}{%
  \input{figures/design_swaps_methylation.tikz}%
}
\caption{Design matrix  $(W,X)$ and within-donor swaps for four  contrasts. Act.\ and Rest.\ denote activated and resting cells; $\mathbf{1}$ is the intercept column. The displayed partition uses the resting Treg indicator as $W$. Arrows identify the pairs of samples that can be exchanged for each contrast.}
\label{fig:design_swaps_methylation}
\end{figure}

We seek to test if $\delta_i=0$ for each CpG $i$. We follow our compound randomization approaches for linear models described in Section~\ref{subsec:linearmodel} (with more details in Supplement~\ref{sec:linear_model_supp} about our scores).
We consider two choices for the invariance group: $\mathcal{H}_{\mathrm{orth},X}$ in~\eqref{eq:rotations_within} and $\mathcal{H}_{\mathrm{perm},X}$ in~\eqref{eq:permutations_within}. The latter is intuitive to describe: the naive and resting Treg samples from donor M28 have identical $X_j$, as do those from donor M30. Thus exchanging either pair leaves $X$ unchanged. We may swap neither pair, either one, or both, giving the four elements of $\mathcal{H}_{\mathrm{perm},X}$. See the left part of Figure~\ref{fig:design_swaps_methylation} for an illustration.

We are interested in three more contrasts, namely: act.\ naive vs.\ naive, act.\ Treg vs.\ rest.\ Treg, and act.\ Treg vs.\ act.\ naive. For each, we reparameterize the same model so that $\delta_i$ is the contrast of interest, with corresponding covariates $W_j$ and $X_j$. The right part of Figure~\ref{fig:design_swaps_methylation} shows the pairs that we may exchange: those with identical $X_j$ under the corresponding reparameterization. For activated naive vs.\ naive, we may swap the pair within each of the three donors, giving $2^3=8$ permutations, but we only use $4$ (because flipping all three donor pairs only has the effect of changing the sign of the estimated coefficient, and so $\widehat{S}(H\boldZ_i) = \widehat{S}(H'\boldZ_i)$ for $H$ and $H'$ that differ by simultaneously flipping all of M28, M29 and M30). For activated Treg vs.\ resting Treg, only M30 has both samples, giving $2$ permutations. For activated Treg vs.\ activated naive, we may swap the pairs within M29 and M30, giving $4$ permutations.

\input{tables/methylation.tex}

We show results in Table~\ref{tab:methylation-discoveries}, showing both the number of discoveries of different methods at $\alpha=0.05$ and the parameters of the estimated limma prior in~\eqref{eq:limma_variance_prior}.
For the resting Treg vs.\ naive contrast, the t-test and separate rotations make only $4$ discoveries. Limma and compound BH with rotations make $3{,}023$ and $3{,}033$, respectively. Variance moderation thus increases the number of discoveries substantially.
For the other three contrasts, compound rotations make substantially more discoveries than limma. One possible explanation is miscalibration of limma's reference distribution $t_{\hat\nu_0+\nu}$ when the working model does not hold. The simulations of~\citet{ignatiadis2025empirical} show that such misspecification can make limma either conservative or anticonservative. Here the larger discovery counts with compound rotations are consistent with conservative calibration by limma. Compound rotations avoid relying on this t-distribution by calibrating the score against the rotated data. The DDR modification (Procedure~\ref{proc:compound_ddr}, with $\tau=\alpha/10$) changes the discovery counts by less than $2\%$ for every contrast.

BH with separate or compound permutations makes no discoveries due to only having two or four effective permutations.
With so few permutations available, this is a natural setting for SeqStep+, as discussed in Section~\ref{subsec:seqstep}. We use Procedure~\ref{proc:seqstep_multi} with the full permutation group for each contrast, except for activated naive vs.\ naive, where we fix the M28 pair and only allow swaps within M29 and M30. SeqStep+ makes the most discoveries for every contrast, including $88$ for Treg activation, where compound rotations make $12$ and limma makes none.
\section{Discussion}
\label{sec:conclusion}

In this paper, we have developed procedures that borrow information
across units to construct powerful scores, while providing finite sample
FDR guarantees under independence across units and null group invariance.
For the problem of testing means with few replicates, we use the orbits
to learn the variance distribution and construct moderated t-statistics.
Under limma's working model, our compound orthogonal rotation and SeqStep+ procedures asymptotically match the power of an oracle
local false discovery rate procedure in the sparse regime we study.
Meanwhile, the finite sample FDR guarantees do not require this working
model to hold.

Our proposal is related to several other approaches that borrow
information about the variances. \citet{lu2016variance} and~\citet{
ignatiadis2025empirical} generalize limma by constructing
empirical Bayes pivots for nonparametric choices of the variance
distribution in lieu of~\eqref{eq:limma_variance_prior}.
\citet{ignatiadis2025empirical} call these ``empirical partially Bayes
p-values,'' following terminology of~\citet{cox1975note}: only a prior
on the nuisance parameters $\sigma_i^2$ is posited, but not on the
primary parameters $\mu_i$.
These approaches are plug-in and only have asymptotic FDR control.
Interestingly, \citet{ignatiadis2025empirical} show that their method
still controls FDR asymptotically when the variances are treated as
fixed: in this frequentist regime, their empirical partially Bayes
p-values are asymptotic versions of the compound p-values in
Definition~\ref{defi:compound_pvalues} (see
\citet[Section~5.2]{ignatiadis2025asymptotic} for a formal definition
and treatment of asymptotic compound p-values).
Thus compound p-values also arise in this approach, although through
a different argument.
Other authors, e.g., \citet{lu2019empirical, zheng2021mixtwice,
seo2025empirical}, have sought plug-in approximations to local fdr
oracles under hierarchical models for $(\mu_i,\sigma_i^2)$ different
from~\eqref{eq:lonnsted_prior}.

A further class of approaches provides finite sample FDR control and
shares our idea of using information in the orbits. Under the Gaussian
model~\eqref{eq:normal_samples}, the t-test p-value $P_i^t$ is independent
of $\hat{\tau}_i^2$ when $\mu_i=0$, and $\hat{\tau}_i^2$ identifies the
orbit under $\mathcal{H}_{\mathrm{orth}}$. One can use
$\hat{\tau}_1^2,\ldots,\hat{\tau}_n^2$ to construct nonnegative compound
e-values~\citep{ignatiadis2025asymptotic}
$E_i=E_i(\hat{\tau}_1^2,\ldots,\hat{\tau}_n^2)$ satisfying
$\EEInline{\sum_{i:\mu_i=0}E_i}\leq n$, and then apply
ep-BH~\citep{ignatiadis2024evalues}, that is, BH applied to
$Q_i:=P_i^t/E_i$. The $Q_i$ are compound p-values, and the resulting
procedure has finite sample FDR control and often more power than BH
with t-tests alone. References pursuing variants of this approach
include~\citet{westfall2004weighted, finos2007fdr,
bourgon2010independent, guo2017analysis, ignatiadis2021covariate,
ignatiadis2026tiny}. These methods use the information in the orbits
to weight the t-test p-values. Here we instead use it to learn a score
function, which we calibrate against group-transformed scores pooled
across units. For our proposed procedures, we also establish the
asymptotic power guarantee in Theorem~\ref{theo:sparse_regime}.

More broadly, the idea of careful masking to construct powerful scores
while retaining finite sample FDR control guarantees has by now a rich literature attached to it; some examples include~\citet{habiger2014compound, lei2018adapt, yurko2020selective,
ignatiadis2021covariate, yang2021bonus, marandon2024adaptive,
tian2025conformalized, chakraborty2026power}.
Our interpretation of this literature is that the emphasis has been
on learning properties of the alternative distribution, or how it
differs from the null. Here we show that even learning properties of
the null distributions, when these are composite, can be very
informative. Indeed, limma is one of the workhorses of modern genomics,
and it borrows information across units precisely by learning the
variance distribution in~\eqref{eq:limma_variance_prior}. Our results
show how to retain the resulting power gains while providing finite
sample FDR guarantees that do not require this variance model to hold.

\paragraph{Reproducibility.} Code for the simulations and data analyses, developed with assistance from GPT-6 Astra, is available at the following Github repository:\\
\url{https://github.com/nignatiadis/compound-randomization-paper}

\paragraph{Acknowledgements.}  We thank Jierui Zhu for helpful comments. This work was completed in part with resources provided by the University of Chicago's Research Computing Center. N.I. gratefully acknowledges support from the U.S. National Science Foundation (DMS-2443410).
The authors acknowledge grant ANR-23-CE40-0018-01 (BACKUP) of the French National Research Agency ANR.

\bibliographystyle{abbrvnat}
\bibliography{biblio}

\appendix

\setcounter{prop}{0}

\renewcommand{\theprop}{S\arabic{prop}}
\renewcommand{\theHprop}{S\arabic{prop}}
\renewcommand{\thelemm}{S\arabic{lemm}}

\section{Reformulations of proposed procedures}

\subsection{Compound BH with randomization}
\label{subsec:compound_threshold}
\begin{mdframed}[nobreak=true]
\renewcommand{\theprocedurevariant}{\getrefnumber{proc:compound_pvalues}\ensuremath{'}}
\refstepcounter{procedurevariant}\label{proc:compound_threshold}
\textbf{Procedure \theprocedurevariant\ (Compound BH: threshold formulation).}

\begin{enumerate}
\item Use $\mathcal{O}_{1:n}$ to learn a score function $\widehat{S}(\cdot)$ that takes the $\boldZ_i$ as input.
\item For each $s$, let
$$\widehat{\mathrm{FDR}}(s) := \frac{ \sum_{k=1}^n \int_{\mathcal{H}} \ind\cb{ \widehat{S}(H \boldZ_k) \geq s} \dd\mu_{\mathcal{H}}(H)}
{1 \lor \sum_{k=1}^n \ind\cb{ \widehat{S}(\boldZ_k) \geq s}},$$  
and compute the threshold $\hat{s} := \inf\Big\{s \in \{ \widehat{S}(\boldZ_1),\ldots,\widehat{S}(\boldZ_n)\}\,:\, \widehat{\mathrm{FDR}}(s) \leq \alpha\Big\}$.
\item Reject all $i$ with $\widehat{S}(\boldZ_i) \geq \hat{s}$.
\end{enumerate}
\end{mdframed}

\subsection{SeqStep+ with group of order two}
\label{subsec:seqstep_order_two}

\begin{mdframed}[nobreak=true]
\renewcommand{\theprocedurevariant}{\getrefnumber{proc:seqstep_multi}\ensuremath{'}}
\refstepcounter{procedurevariant}\label{proc:seqstep_order_two}
\textbf{Procedure \theprocedurevariant\ (Selective SeqStep+ with a learned score and group of order $2$).}\begin{enumerate}
\item Use $\mathcal{O}_{1:n}$ to learn a nonnegative score function $\widehat{S}(\cdot)$ that takes the $\boldZ_i$ as input.
\item Compute the symmetrized scores
$
\tilde{S}_i:= (\widehat{S}(\boldZ_i)\lor\widehat{S}(H\boldZ_i))\cdot \mathrm{sign}(\widehat{S}(\boldZ_i)-\widehat{S}(H\boldZ_i)),
$
with the convention $\mathrm{sign}(0)=0$.
For each $s>0$, let
\begin{equation}
\widehat{\mathrm{FDR}}(s):=
\frac{1+\sum_{i=1}^n\ind\cb{\tilde{S}_i\leq-s}}
{1\lor\sum_{i=1}^n\ind\cb{\tilde{S}_i\geq s}},
\label{eq:fdr_seqstep}
\end{equation}
and compute the threshold $\hat{s}:=\inf\!\big\{s\in\big\{|\tilde{S}_1|,\ldots, |\tilde{S}_n|\big\}\setminus \{0\}\,:\,\widehat{\mathrm{FDR}}(s)\leq\alpha\big\}$.
\item Reject all $i$ with $\tilde{S}_i\geq\hat{s}$.
\end{enumerate}
\end{mdframed}

\section{More on Selective SeqStep+ with groups of order \texorpdfstring{$2$}{2}}

\subsection{RESS and SENS as special cases of Procedure~\ref{proc:seqstep_order_two}}
\label{subsec:ress_and_SENS}
Our goal in this supplement is to show that we can interpret 
RESS (Reflection via Sample Splitting) by~\citet{zou2020new} and SENS (self-calibrated empirical null samples) by~\citet{tian2025conformalized} as a special case of the order-two subgroup framing of Procedure~\ref{proc:seqstep_order_two}. For both procedures, we explain the correspondence for even $K \geq 4$. We spell out this connection because both RESS and SENS are already powerful; and the generalization in
Procedure~\ref{proc:seqstep_order_two} allows to further improve these methods.

For notation convenience, we interpret the summary statistic functions in~\eqref{eq:summary_statistics} and~\eqref{eq:standard_t} as also incorporating the length of $\boldz$. More specifically, we partition any $\boldz \in \RR^K$ as
$$
\boldz^\top = (\boldz^A, \boldz^B)^\top,\;\; \boldz^A, \boldz^B \in \RR^{K/2}.
$$
Our notation convention above means that, for instance, we have:
$$
\hat{\mu}(\boldz) = \frac{1}{K}\sum_{j=1}^K z_j,\;\;\; \hat{\mu}(\boldz^A) = \frac{2}{K} \sum_{j=1}^{K/2} z_{j}, \;\;\; \hat{\mu}(\boldz^B) = \frac{2}{K}\sum_{j=(K/2)+1}^K z_j. 
$$
We also recall the notation $H_{\mathrm{half}}$ from~\eqref{eq:h_half}.

\paragraph{RESS (Reflection via Sample Splitting).} RESS works as follows: it computes $T_i^A = T(\boldZ_i^A)$ and $T_i^B = T(\boldZ_i^B)$ and then lets $W_i = T_i^A T_i^B$. It then computes
\begin{equation}
\hat{w} := \inf\cb{ w\in\{|W_1|,\ldots,|W_n|\}\,:\, \frac{1+\sum_{i=1}^n \ind\{ W_i \leq -w\} }{ 1 \lor \sum_{i=1}^n \ind\{W_i \geq w\}} \leq \alpha},
\label{eq:seqstep_hatw}
\end{equation}
and rejects all hypotheses with $W_i \geq \hat{w}$. 

We can cast it as a special case of Procedure~\ref{proc:seqstep_order_two} with $H=H_{\mathrm{half}}$ by using the score
$$
\widehat{S}(\boldz) = \p{ T(\boldz^A)\cdot T(\boldz^B)}_+,
$$
where $(a)_+ := a \lor 0$ is the positive part. 

To see the equivalence, notice that:
$$
\widehat{S}(H_{\mathrm{half}}\boldz) = \p{- T(\boldz^A)\cdot T(\boldz^B)}_+ = \p{T(\boldz^A)\cdot T(\boldz^B)}_-,
$$
where $(a)_-:= (-a)\lor 0$ is the negative part. Hence
$
|\tilde{S}_i| = |W_i|
$ and $\widehat{S}(Z_i) > \widehat{S}(H_{\mathrm{half}}Z_i)$ if and only if $W_i > 0$. 

We emphasize here that the score of RESS is fixed (i.e., not data-driven).

\paragraph{SENS (self-calibrated empirical null samples).}

SENS works as follows.  It first defines
$$
X_i :=\frac{\hat{\mu}(\boldZ_i^A)+\hat{\mu}(\boldZ_i^B)}{2} = \hat{\mu}(\boldZ_i),\;\;  X_i^0:=\frac{\hat{\mu}(\boldZ_i^A)-\hat{\mu}(\boldZ_i^B)}{2},
$$
and $V_i:=V(\boldZ_i)$, where
\begin{equation}
V^2(\boldz) := \frac{\hat{\sigma}^2(\boldz^A) + \hat{\sigma}^2(\boldz^B)}{2K}.
\label{eq:sens_V}
\end{equation}
They then let
\begin{equation}
T_i^{\mathrm{SENS}}:= \Phi^{-1}(F_{t,K-2}(X_i/ V_i)),\;\; T_i^{\mathrm{SENS},0} := \Phi^{-1}(F_{t,K-2}(X_i^0/ V_i)),
\label{eq:sens_t}
\end{equation}
where $F_{t,K-2}$ is the cumulative distribution function of the t-distribution with $K-2$ degrees of freedom and $\Phi$ is the standard normal distribution function. They then use the unordered pairs $\{T_i,T_i^0\}$, $i\in[n]$, to learn a score function $\hat{g}$.\footnote{The authors propose two options for $\hat{g}$, leading to two versions of SENS: SENS-General and SENS-Gaussian. The former is motivated by a general symmetric noise distribution, and the latter for Gaussian noise.
} Finally, they take
\begin{equation}
U_i := \hat{g}(T_i^{\mathrm{SENS}}) ,\;\;\, U_i^0 := \hat{g}(T_i^{\mathrm{SENS},0}),\;\; W_i = \mathrm{sign}(U_i^0 - U_i)\cdot \cb{\exp( -U_i)\lor \exp(-U_i^0)},
\label{eq:hat_sens}
\end{equation}
and apply SeqStep+, i.e., they use~\eqref{eq:seqstep_hatw} (on the $W_i$ defined here), and reject all hypotheses with $W_i \geq \hat{w}$. 

We can cast SENS as a special case of Procedure~\ref{proc:seqstep_order_two} with $H=H_{\mathrm{half}}$ by using the score
$$
\widehat{S}(\boldz) = \exp\Bigl[-\hat{g}\Bigl\{\Phi^{-1}\Bigl(F_{t,K-2}\bigl\{\hat{\mu}(\boldz)/V(\boldz)\bigr\}\Bigr)\Bigr\}\Bigr].
$$
To see the equivalence, notice that:
$$
X_i = \hat{\mu}(\boldZ_i),\;\; X_i^0 = \hat{\mu}(H_{\mathrm{half}}\boldZ_i),\;\; V_i = V(H_{\mathrm{half}}\boldZ_i) = V(\boldZ_i).$$
Moreover $\hat{g}$ is indeed learned from the orbits since  the unordered pairs $\{T_i,T_i^0\}$ are measurable 
with respect to the orbits $\{\boldZ_i, H_{\mathrm{half}}\boldZ_i\}$.

\subsection{SENS-limma}
\label{subsec:sens_limma}

The goal of this section is to demonstrate that the general formulation in Procedure~\ref{proc:seqstep_order_two} provides the right abstraction for deriving our SeqStep+ method that asymptotically matches the performance of oracle limma. Here it's important that we are computing $\widehat{S}(\cdot)$ as we would usually do and do not try to build invariance properties into it directly.

To make this point, we consider an oracle variant of SENS that we call SENS-limma and show that it is nearly powerless in our setting. We emphasize that this procedure was not proposed by the authors of SENS (they are mostly motivated to learn properties of the alternatives).

Our starting point is based on the following observation by the SENS authors: they argue heuristically that a powerful choice of $\hat{g}(\cdot)$ in~\eqref{eq:hat_sens} should approximate $\mathbb P[\mu_i =0 \mid T_i^{\mathrm{SENS}}=\cdot]$, i.e., the local false discovery rate conditional on the summary statistic $T_i^{\mathrm{SENS}}$ defined in~\eqref{eq:sens_t}; see~\citet[Section~3.2]{tian2025conformalized}.

The oracle SENS-limma method uses exactly $\hat{g}(\cdot) = \PP{\mu_i =0 \mid T_i^{\mathrm{SENS}}=\cdot}$ under our data-generation model given by~\eqref{eq:normal_samples} and~\eqref{eq:lonnsted_prior}. We write $\mathrm{SENS\textnormal{-}limma}(\alpha)$ for the resulting procedure using~\eqref{eq:hat_sens} and~\eqref{eq:seqstep_hatw} at level $\alpha$.

We have the following result that complements Theorem~\ref{theo:sparse_regime}.
\begin{prop}
\label{prop:sparse_regime_sens}
Consider the asymptotic regime of Theorem~\ref{theo:sparse_regime}, with even $K\geq4$, and fix $\alpha\in(0,1)$. Then
$$
\mathrm{Pow}_n[\mathrm{SENS\textnormal{-}limma}(\alpha)] \to 0
\;\text{ as }n\to\infty.
$$
\end{prop}

\section{Details for other settings for compound randomization}

\subsection{Details for two-sample tests}
\label{sec:twosample_supp}

\subsubsection{Existing approaches}
We start by providing a similar summary of existing approaches for the two-sample problem as we did in Section~\ref{sec:existing} for the one-sample problem. We keep it shorter however and only discuss the standard t-test and limma.

We continue using the summary notation in~\eqref{eq:summary_statistics} that we used for the one-sample problem. 

Given $\boldz \in \RR^K$, write $\boldz^A$ and $\boldz^B$ for the subvectors indexed by $A$ and $B$, and $\hat{\mu}(\boldz^A)$ and $\hat{\mu}(\boldz^B)$ for their sample means. We define the difference in sample means and the pooled sample variance as
\begin{equation}
\hat{\delta}(\boldz) := \hat{\mu}(\boldz^B) - \hat{\mu}(\boldz^A),\;\;\;
\hat{\sigma}_{\mathrm{pool}}^2(\boldz) := \frac{1}{K-2}\cb{\sum_{j \in A} (z_j - \hat{\mu}(\boldz^A))^2 + \sum_{j \in B} (z_j - \hat{\mu}(\boldz^B))^2}.
\label{eq:summary_stats_twosample}
\end{equation}
For the above quantities, we were careful not to use conflicting notation with the one-sample case. However, below we will overload some notation, for instance, we will write $\nu = K-2$. 

Note that under~\eqref{eq:normal_samples_two_sample}, complete and sufficient statistics for $(\theta_i, \delta_i, \sigma_i^2)$ are given by $\hat{\mu}(\boldZ_i)$, $\hat{\delta}(\boldZ_i)$ and $\hat{\sigma}_{\mathrm{pool}}^2(\boldZ_i)$.

Effectively both the t-test and limma in this setting can be interpreted as discarding $\hat{\mu}(\boldZ_i)$ and working only with $ \hat{\delta}(\boldZ_i)$ and $\hat{\sigma}_{\mathrm{pool}}^2(\boldZ_i)$. At that point, one can proceed almost verbatim as in the one-sample case.

\paragraph{Standard t-tests with BH.}  Here, the textbook approach is to define the equal variance two-sample t-statistic function for $\boldz \in \mathbb R^K$ as
\begin{equation}
T(\boldz) :=  \sqrt{\frac{K_AK_B}{K_A+K_B}}\cdot \frac{\hat{\delta}(\boldz)}{\hat{\sigma}_{\mathrm{pool}}(\boldz)},
\label{eq:standard_t_two}
\end{equation}
and then compute p-values as $P_i^{t} := 2 \bar{F}_{t,\nu}(|T(\boldZ_i)|)$. Finally, one applies the Benjamini-Hochberg procedure to $P_1^{t},\ldots,P_n^{t}$ at level $\alpha$. This procedure controls the FDR at level $\alpha$ under~\eqref{eq:normal_samples_two_sample}.

\paragraph{Limma with BH.}  Suppose one is willing to continue assuming~\eqref{eq:limma_variance_prior}. Then, one can define in analogy to~\eqref{eq:moderated_stats} (and again overloading notation),
\begin{equation}
\hat{\sigma}^2(\boldz; \nu_0, s_0) := \frac{  \nu_0 s_0^2 + \nu \hat{\sigma}_{\mathrm{pool}}^2(\boldz)}{\nu_0 + \nu},\;\;\; T(\boldz; \nu_0, s_0) := \sqrt{\frac{K_AK_B}{K_A+K_B}}\cdot \frac{ \hat{\delta}(\boldz)}{\hat{\sigma}(\boldz; \nu_0, s_0)}.
\label{eq:moderated_stats_two_sample}
\end{equation}
If~\eqref{eq:normal_samples_two_sample} with $\delta_i=0$ holds for some $i$, and~\eqref{eq:limma_variance_prior} also holds, then $T(\boldZ_i; \nu_0, s_0) \sim t_{\nu + \nu_0}$. For oracle limma, one can compute p-values as $P_i^{\mathrm{limma}} := 2 \bar{F}_{t,\nu+\nu_0}(|T(\boldZ_i; \nu_0, s_0)|)$ and then apply BH. For data-driven limma, one first estimates the hyperparameter $\nu_0$ and $s_0^2$.

\subsubsection{Implementing compound randomization}

Recall from Section~\ref{subsec:twosample}, that for the two-sample testing case we consider two choices of group: $\mathcal{H}_{\mathrm{perm}}$ or $\mathcal{H}_{\mathrm{orth},\mathbf{1}}$. For both of these, $\hat{\sigma}^2(\boldZ_i)$ defined in~\eqref{eq:summary_statistics} is measurable with respect to the $i$-th orbit $\mathcal{O}_i$. Notice that this is the sample variance ignoring group membership, and so different from $\hat{\sigma}^2_{\mathrm{pool}}(\boldZ_i)$ defined in~\eqref{eq:summary_stats_twosample}. Nonetheless under the null we have that (compare to~\eqref{eq:hat_tau_i_bn}):
$$
\hat{\sigma}_i^2 \mid (\sigma_i^2, \delta_i=0) \; \sim \frac{\sigma_i^2}{K-1} \chi^2_{K-1}.
$$
This suggests that (compare to~\eqref{eq:hlimma_suff}), we can compute
$$
(\hat{\nu}_0(\mathcal{O}_{1:n}), \hat{s}_0^2(\mathcal{O}_{1:n})) := h^{\mathrm{limma}}(\hat{\sigma}_1^2,\ldots,\hat{\sigma}_n^2; K-1),
$$
with $\hat{\sigma}_i^2 = \hat{\sigma}^2(\boldZ_i)$. Then we can use the score:
$$
\widehat{S}(\boldz) := \abs{ T(\boldz; \hat{\nu}_0(\mathcal{O}_{1:n}), \hat{s}_0(\mathcal{O}_{1:n}))}.
$$
For $\mathcal{H}_{\mathrm{orth},\mathbf{1}}$, we can again compute the compound randomization p-values analytically. Notice that
$$
(K-1)\hat{\sigma}^2(\boldz) = \nu\hat{\sigma}_{\mathrm{pool}}^2(\boldz) + \frac{K_AK_B}{K_A+K_B}\hat{\delta}^2(\boldz).
$$
The same calculation as in Proposition~\ref{prop:orthogonal_rotation_cpvalue}, now rotating the centered observations, gives
$$
\pcompfun(\boldz; \widehat{S}(\cdot), \mathcal{O}_{1:n}, \mathcal{H}_{\mathrm{orth},\mathbf{1}}) = \frac{1}{n} \sum_{i=1}^n \bar{F}_{B,1/2,\nu/2}\cb{\frac{\widehat{S}^2(\boldz)/\{\nu+\hat{\nu}_0\}}{1+\widehat{S}^2(\boldz)/\{\nu+\hat{\nu}_0\}}\p{1+\frac{\hat{\nu}_0\hat{s}_0^2}{(K-1)\hat{\sigma}_i^2}}}.
$$
For both $\mathcal{H} \in \{\mathcal{H}_{\mathrm{perm}}, \mathcal{H}_{\mathrm{orth},\mathbf{1}}\}$ and any $\boldz \in \mathcal{O}_i$, we have that:
$$
\psepfun(\boldz; \widehat{S}(\cdot), \mathcal{O}_i, \mathcal{H}) = \psepfun(\boldz; \abs{T(\cdot;0,1)}, \mathcal{O}_i, \mathcal{H}).
$$
Thus, for permutations, separate randomization gives the usual two-sample permutation p-value, regardless of the estimates $\hat{\nu}_0$ and $\hat{s}_0^2$.

Moreover, as in Proposition~\ref{prop:sep_pvalue_limma}, separate randomization with $\mathcal{H}_{\mathrm{orth},\mathbf{1}}$ returns the usual two-sided two-sample t-test p-value:
$$
P_i^{\mathrm{sep}} = \psepfun(\boldZ_i; \widehat{S}(\cdot), \mathcal{O}_i, \mathcal{H}_{\mathrm{orth},\mathbf{1}}) = 2\bar{F}_{t,\nu}(|T(\boldZ_i)|) = P_i^t\;\text{ almost surely}.
$$
Indeed, on each orbit, both the absolute moderated t-statistic and the absolute ordinary t-statistic are increasing functions of $\hat{\delta}^2(\boldz)$.

\begin{rema}[Limma-trend]
We note that beyond $\hat{\sigma}^2(\boldZ)$, $\hat{\mu}(\boldZ)$ is also measurable with respect to orbits $\mathcal{O}_i$ (for both choices of $\mathcal{H})$. This is important if one wants to implement compound randomization based on limma-trend p-values~\citep{law2014voom, nandy2026how}. We do not pursue the details here.
\end{rema}

\subsection{Details for testing a coefficient in a linear model}
\label{sec:linear_model_supp}

We continue from~\eqref{eq:normal_samples_linear_model} and consider only $\mathcal{H}_{\mathrm{orth},X}$. As in Supplement~\ref{sec:twosample_supp}, we overload notation below, now writing $\nu=K-p-1$. Given $\boldz \in \RR^K$, the joint least squares estimates are
$$
(\hat{\delta}(\boldz),\hat{\theta}(\boldz)) := \argmin_{\delta\in\RR,\;\theta\in\RR^p} \Norm{\boldz-W\delta-X\theta}_2^2.
$$
We also define the residual variance:
$$
\hat{\sigma}_{\mathrm{res}}^2(\boldz) := \frac{1}{\nu}\Norm{\boldz-W\hat{\delta}(\boldz)-X\hat{\theta}(\boldz)}_2^2.
$$
Under~\eqref{eq:normal_samples_linear_model}, $\hat{\delta}(\boldZ_i)$, $\hat{\theta}(\boldZ_i)$ and $\hat{\sigma}_{\mathrm{res}}^2(\boldZ_i)$ are sufficient for $(\delta_i,\theta_i,\sigma_i^2)$. Again, we can interpret both the standard t-test and limma as discarding $\hat{\theta}(\boldZ_i)$ and working only with the coefficient estimate $\hat{\delta}(\boldZ_i)$ and the residual variance $\hat{\sigma}_{\mathrm{res}}^2(\boldZ_i)$.

In analogy to~\eqref{eq:moderated_stats_two_sample}, we define
\begin{equation}
\hat{\sigma}^2(\boldz; \nu_0, s_0) := \frac{\nu_0s_0^2+\nu\hat{\sigma}_{\mathrm{res}}^2(\boldz)}{\nu_0+\nu},\;\;\;
T(\boldz; \nu_0, s_0) := \Norm{U^\top W}_2\cdot\frac{\hat{\delta}(\boldz)}{\hat{\sigma}(\boldz; \nu_0, s_0)}.
\label{eq:moderated_stats_linear_model}
\end{equation}
We also write $T(\boldz):=T(\boldz;0,1)$ for the usual regression t-statistic.

For compound randomization, notice that $\Norm{U^\top\boldZ_i}_2$ is measurable with respect to the orbit $\mathcal{O}_i$, and under the Gaussian null,
$$
\frac{1}{K-p}\Norm{U^\top\boldZ_i}_2^2 \mid (\sigma_i^2,\delta_i=0) \sim \frac{\sigma_i^2}{K-p}\chi^2_{K-p}.
$$
This suggests computing
$$
(\hat{\nu}_0(\mathcal{O}_{1:n}), \hat{s}_0^2(\mathcal{O}_{1:n})) := h^{\mathrm{limma}}\p{\frac{\Norm{U^\top\boldZ_1}_2^2}{K-p},\ldots,\frac{\Norm{U^\top\boldZ_n}_2^2}{K-p}; K-p}.
$$
We then use the same score $\widehat{S}(\cdot)$ as in Supplement~\ref{sec:twosample_supp}, with $T$ now given by~\eqref{eq:moderated_stats_linear_model}. The computational shortcut there continues to hold, with $\mathcal{H}_{\mathrm{orth},\mathbf{1}}$ replaced by $\mathcal{H}_{\mathrm{orth},X}$ and $(K-1)\hat{\sigma}_i^2$ replaced by $\Norm{U^\top\boldZ_i}_2^2$, since
$$
\Norm{U^\top\boldZ_i}_2^2 = \nu\hat{\sigma}_{\mathrm{res}}^2(\boldZ_i) \,+\, \Norm{U^\top W}_2^2\hat{\delta}^2(\boldZ_i).
$$
Finally, separate randomization again returns the usual two-sided t-test p-value for the coefficient:
$$
P_i^{\mathrm{sep}} = \psepfun(\boldZ_i; \widehat{S}(\cdot), \mathcal{O}_i, \mathcal{H}_{\mathrm{orth},X}) = 2\bar{F}_{t,\nu}(|T(\boldZ_i)|)\;\text{ almost surely}.
$$

\section{Further simulation results}

\subsection{Compound randomization variants}
\label{subsec:compound_variants}

In the experiments of the main text, compound BH at level $\alpha$ controls the FDR empirically whenever the corresponding null invariance assumption holds. This is consistent with asymptotic FDR control at level $\alpha$ for BH with compound p-values~\citep{armstrong2022false,ignatiadis2025empirical}. However, the finite-sample guarantee in Theorem~\ref{theo:bh} is only $1.93\alpha$. We therefore compare two modifications that give finite-sample FDR control at level $\alpha$: BH at level $\alpha/1.93$, using the bound of~\citet{barber2026false}, and DDR (Procedure~\ref{proc:compound_ddr}), based on Theorem~\ref{theo:ddr}. For orthogonal rotations, Theorem~\ref{theo:sparse_regime}(iii) also shows that DDR asymptotically matches the power of the local fdr oracle in the sparse regime, with $\tau_n\to0$ and $\pi_1/\tau_n\to0$.

In Figures~\ref{fig:supp_gaussian_compound} and~\ref{fig:supp_robustness_compound}, we compare these two modifications with compound BH at level $\alpha$, for each of $\mathcal{H}_{\mathrm{orth}}$ and $\mathcal{H}_{\mathrm{symm}}$. We take $\tau=\alpha/10$ for DDR and use the same simulation settings and learned scores as in Figures~\ref{fig:main_gaussian_AB} and~\ref{fig:main_robustness_CD}.

DDR generally has power close to compound BH at level $\alpha$, while BH at level $\alpha/1.93$ loses more power. The difference is particularly pronounced for sign flips at $K=5$: since $\alpha/1.93<2^{-4}$, compound BH at this level makes no discoveries, whereas DDR retains non-negligible power. Under uniform noise, DDR with orthogonal rotations still exceeds the nominal FDR level for $K\geq9$, since rotation invariance does not hold. All three sign flip variants stay below the nominal level in these experiments.

\begin{figure}
\centering
\includegraphics[width=1\textwidth]{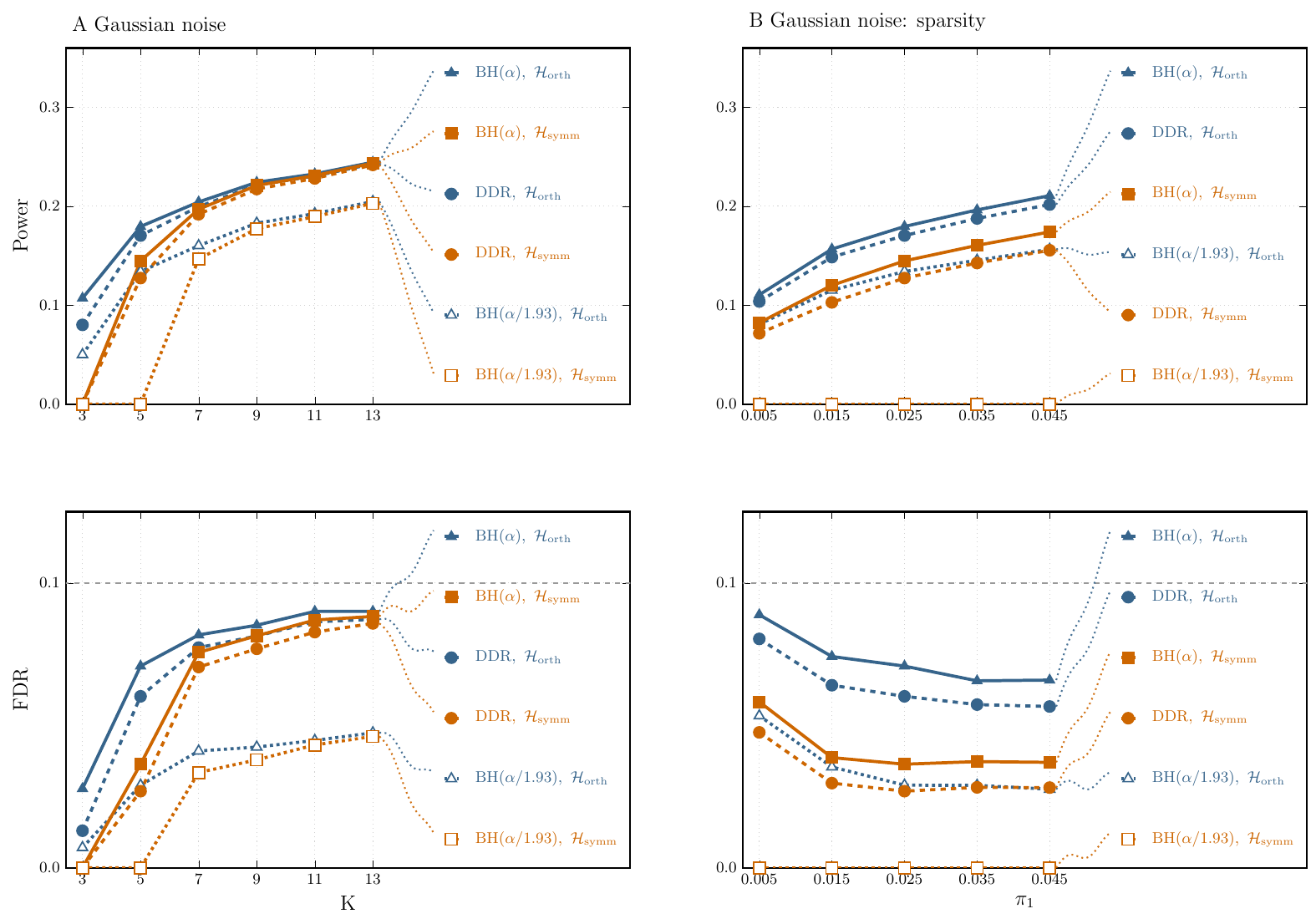}
\caption{Power (top) and FDR (bottom) for compound BH variants under Gaussian noise, in the settings of Figure~\ref{fig:main_gaussian_AB}. For each group, we compare BH at levels $\alpha$ and $\alpha/1.93$ with DDR at level $\alpha$ and $\tau=\alpha/10$, where $\alpha=0.1$.}
\label{fig:supp_gaussian_compound}
\end{figure}

\begin{figure}
\centering
\includegraphics[width=1\textwidth]{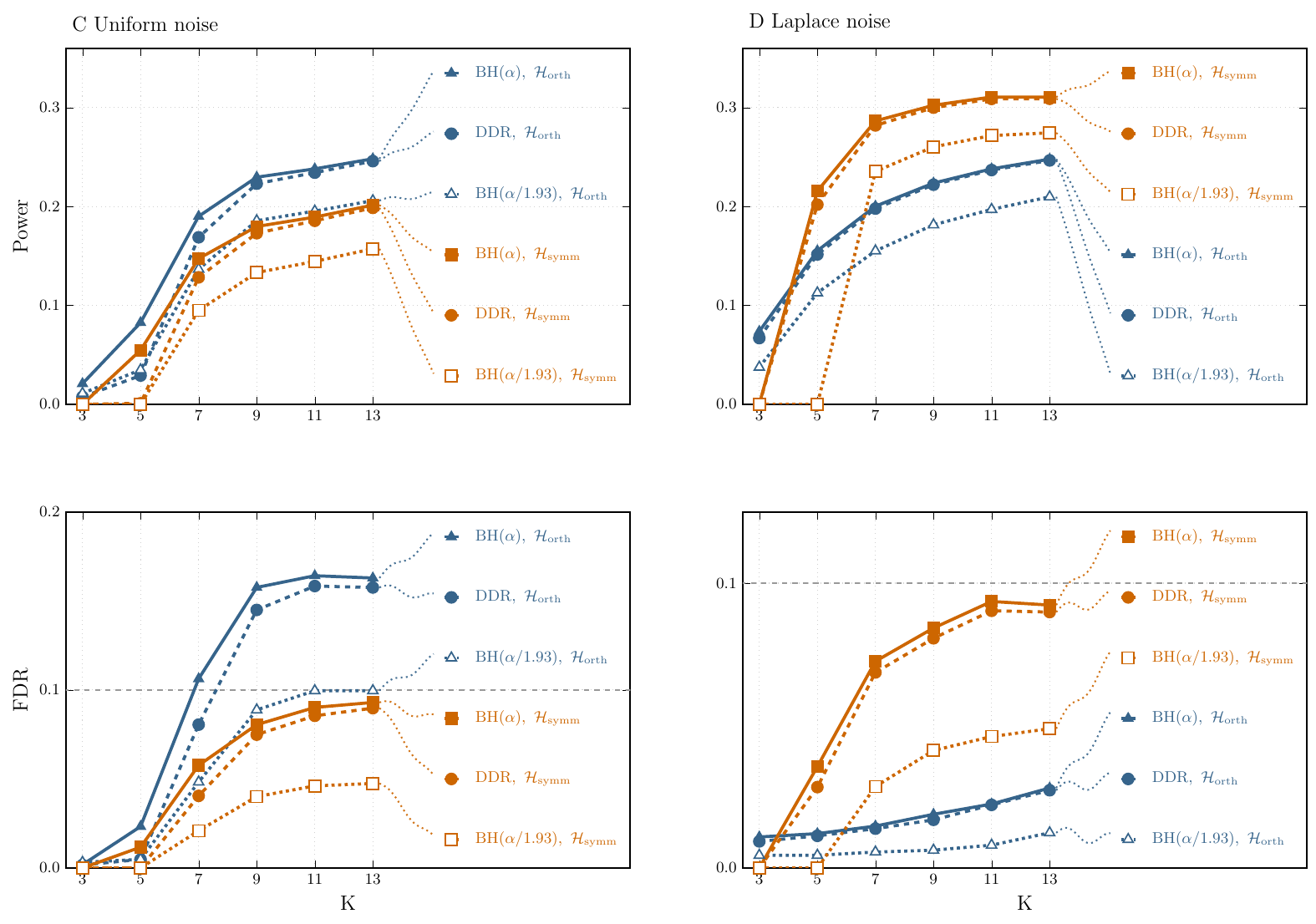}
\caption{Power (top) and FDR (bottom) for the same compound BH variants as in Figure~\ref{fig:supp_gaussian_compound}, under uniform noise (Panel C) and Laplace noise (Panel D), in the settings of Figure~\ref{fig:main_robustness_CD}. The FDR axes have different scales.}
\label{fig:supp_robustness_compound}
\end{figure}

\subsection{A further simulation with larger \texorpdfstring{$K$}{K}}
\label{sec:stylized_example}

We now consider a stylized setting with $K=16$ in which the score $\widehat{S}(\boldz)=|\hat\mu(\boldz)|$ gives the same ranking as the local false discovery rate. This lets us examine the benefit of preserving this ranking through compound randomization in a setting where separate randomization has less power not because of resolution but because it does not retain the ranking.

In our simulation, we generate samples as in~\eqref{eq:normal_samples}, that is,
$$
Z_{ij}\mid\mu_i,\sigma_i\sim\mathrm{N}(\mu_i,\sigma_i^2),\qquad j=1,\ldots,K.
$$
Parameters are generated as follows (for a hyperparameter $a>0$ and a constant $c>0$):
$$
\begin{aligned}
\sigma_i/\sqrt K &\sim 0.85\delta_{0.5}+0.15\delta_{2.5},\\
H_i\mid\sigma_i &\sim \mathrm{Bernoulli}(1-\pi_0(\sigma_i)),
  \;\; \text{where }\pi_0(\sigma):=\frac{1}{1+c\exp(a^2\sigma^2/(2K))},\\
\mu_i\mid H_i,\sigma_i &\sim
  \begin{cases}
  \delta_0, & H_i=0,\\
  \tfrac12\delta_{a\sigma_i^2/K}+\tfrac12\delta_{-a\sigma_i^2/K}, & H_i=1.
  \end{cases}
\end{aligned}
$$
This construction is chosen specifically so that the local false discovery rate is decreasing in $|\hat\mu_i| = |\hat\mu(\boldZ_i)|$, with
$$
\mathrm{lfdr}(\boldZ_i):=\PP{H_i=0\mid\boldZ_i}
=\frac{1}{1+c\cosh(a\hat\mu(\boldZ_i))},
$$
and where $\hat{\mu}$ is defined in~\eqref{eq:summary_statistics}.
In our simulation, we take $n=200$, $K=16$, $a=1.2$, and choose $c$ in each case so that $\EE{1-\pi_0(\sigma_i)}=\pi_1$ with 
$
\pi_1\in\{0.005,0.01,0.02,\ldots,0.05\}.
$

We compare BH applied to compound orthogonal rotation p-values (``Comp. $\mathcal{H}_{\mathrm{orth}}$''), compound sign flip p-values (``Comp. $\mathcal{H}_{\mathrm{symm}}$''), separate sign flip p-values (``Sep. $\mathcal{H}_{\mathrm{symm}}$''), and ordinary t-test p-values (``t-test''). The randomization procedures all use $\widehat{S}(\boldz)=|\hat\mu(\boldz)|$. We also include the local false discovery rate oracle (``Oracle''), as described in Section~\ref{sec:existing}. Throughout, we use nominal FDR level $\alpha=0.1$ and report estimates of FDR and power as in~\eqref{eq:fdr_power_defs}, averaged over $1{,}000$ Monte Carlo replicates.

\begin{figure}
\centering
\includegraphics[width=0.7\linewidth]{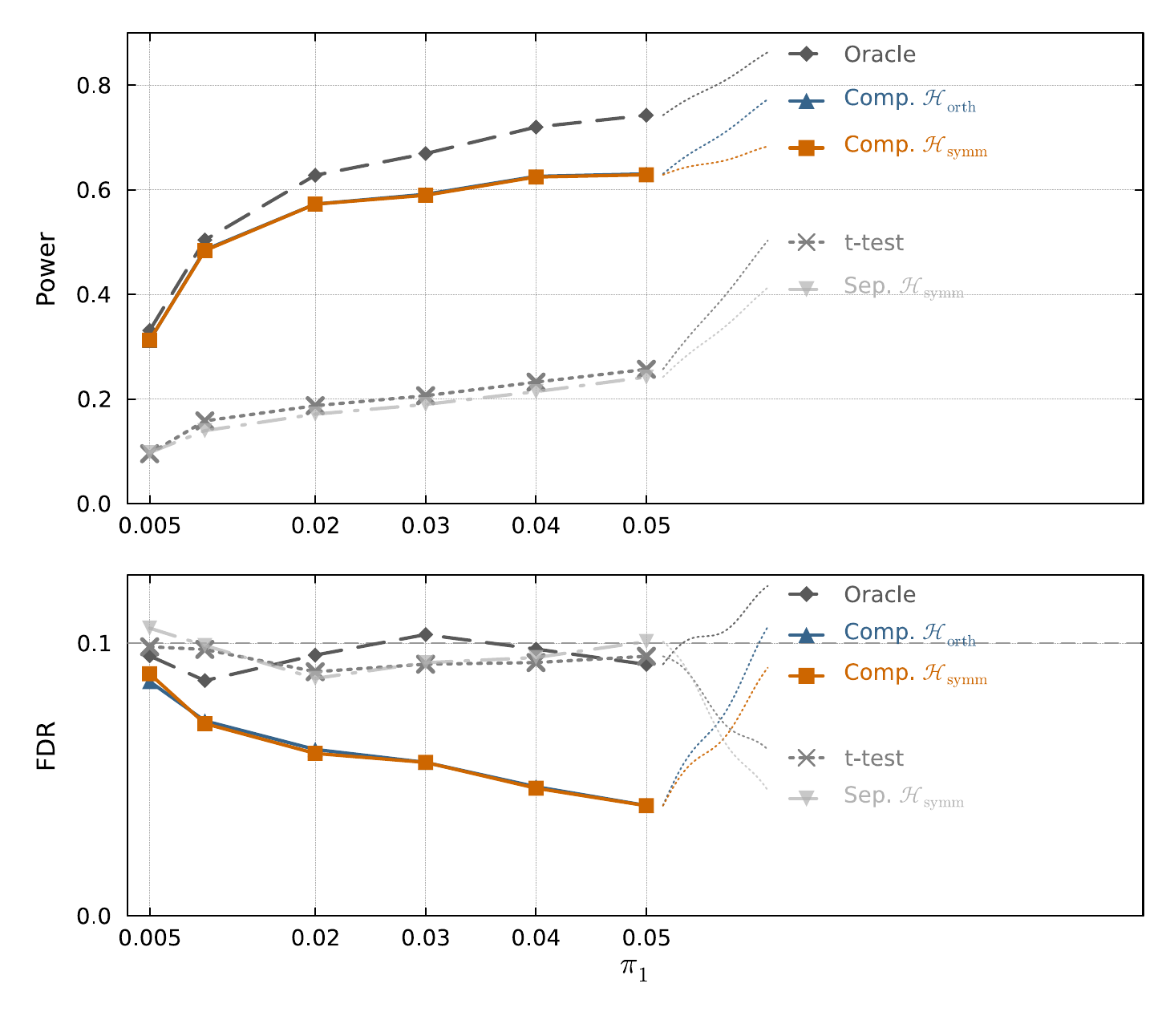}
\caption{Power (top) and FDR (bottom) with $n=200$, $K=16$, and nominal FDR level $\alpha=0.1$, as the non-null proportion $\pi_1$ varies. Our compound procedures (``Comp. $\mathcal{H}_{\mathrm{orth}}$'' and ``Comp. $\mathcal{H}_{\mathrm{symm}}$'') have nearly identical performance and substantially higher power than separate sign flips and ordinary t-tests. Separate sign flips closely track t-tests here. The main message is that compound randomization p-values have the benefit of preserving the ranking across hypotheses implied by the used score function, while separate randomization can rerank hypotheses.
}
\label{fig:stylized_lfdr}
\end{figure}

Results are reported in Figure~\ref{fig:stylized_lfdr}. We make the following observations. Here the separate sign flip procedure is not limited by resolution, and its FDR is very near the nominal level throughout. 
Moreover, its FDR and power very closely track those of the t-test. As explained after Proposition~\ref{prop:sep_pvalue_limma}, this is not surprising for larger $K$, see e.g., Example 15.2.4 in~\citet{lehmann2005testing} for a justification of this phenomenon as $K \to \infty$. What is more interesting is the difference in power between the compound and separate procedures with compound procedures having substantially higher power. This is because the compound procedures preserve the ranking implied by the local false discovery rate, i.e., by the score $\abs{\hat{\mu}_i}$, while the separate procedures start with $\abs{\hat{\mu}(\boldz)}$ but after randomization rank by the t-statistic $|T(\boldz)|$, exactly for rotations and approximately for sign flips. As a last observation we note that again the compound procedures increasingly get closer to the oracle as $\pi_1$ decreases.

\section{Facts about the limma model, the t-distribution and the F-distribution}
\label{sec:t_formulas}

For $\nu>0$, we write $t_{\nu}$ for the t-distribution  with $\nu$ degrees of freedom. Its density is given by
\begin{equation}
f_{t,\nu}(x)
=\frac{\Gamma((\nu+1)/2)}{\sqrt{\nu\pi}\,\Gamma(\nu/2)}
\p{1+\frac{x^2}{\nu}}^{-(\nu+1)/2},\;\; x\in\RR.
\label{eq:t_density}
\end{equation}
We also write $\bar{F}_{t,\nu}$ for the corresponding survival function.

Analogously, for $K \in \mathbb N_{\geq 2}$ and $\nu>0$, we write $F_{K,\nu}$ for the F-distribution with $K$ and $\nu$ degrees of freedom.  Its density is given by
\begin{equation}
f_{F,K,\nu}(x)=\frac{\Gamma((K+\nu)/2)}{\Gamma(K/2)\Gamma(\nu/2)}
  \p{\frac{K}{\nu}}^{K/2}
  x^{K/2-1}\p{1+\frac{Kx}{\nu}}^{-(K+\nu)/2},
  \;\; x>0.
\label{eq:f_density}
\end{equation}

\subsection{Moderated t-statistics and local fdr}

The following results are from~\citet{smyth2004linear}. We reprove them here for self-containedness.

\begin{prop}
\label{prop:moderated_two_groups}
Under~\eqref{eq:normal_samples} and~\eqref{eq:lonnsted_prior}, with $K\geq2$ and $0<\pi_0<1$,
\begin{equation}
T(\boldZ_i;\nu_0,s_0)\mid\mu_i=0\sim t_{\nu+\nu_0},
\;\;
T(\boldZ_i;\nu_0,s_0)
\mid\mu_i\neq0 \sim \sqrt{1+\lambda} \cdot t_{\nu+\nu_0}.
\label{eq:t_two_groups}
\end{equation}
\end{prop}

\begin{proof}
We prove the result only given $\mu_i =0$; the alternative argument is analogous. Conditional on $\mu_i=0$, we can write
$$
(\sqrt K\hat{\mu}_i/\sigma_i,\,\nu\hat{\sigma}_i^2/\sigma_i^2,\, \sigma_i^2)\mid (\mu_i=0)\;\, \stackrel{\mathcal{D}}{=} \;\, \p{G,\, V,\,\frac{\nu_0s_0^2}{V_0}},
$$
where $G\sim\mathrm{N}(0,1)$, $V\sim\chi^2_\nu$ and $V_0\sim\chi^2_{\nu_0}$ are jointly independent.  Substituting into~\eqref{eq:moderated_stats} gives
$$
T(\boldZ_i;\nu_0,s_0) \mid (\mu_i=0)
\;\, \stackrel{\mathcal{D}}{=} \;\,\frac{G}{\sqrt{(V_0+V)/(\nu_0+\nu)}}\sim t_{\nu+\nu_0},
$$
since $V_0+V\sim\chi^2_{\nu_0+\nu}$ independently of $G$.
\end{proof}

\begin{prop}
\label{prop:localfdr_formula}
Under~\eqref{eq:normal_samples} and~\eqref{eq:lonnsted_prior}, with $K\geq2$ and $0<\pi_0<1$, the local false discovery rate is
$$
\mathrm{lfdr}(\boldz)
=\left[1+\frac{\pi_1}{\pi_0}(1+\lambda)^{(\nu_0+\nu)/2}
\left\{\frac{\nu_0+\nu+T^2(\boldz;\nu_0,s_0)}
{(1+\lambda)(\nu_0+\nu)+T^2(\boldz;\nu_0,s_0)}\right\}^{(\nu_0+K)/2}\right]^{-1}.
$$
\end{prop}
\begin{proof}
The pair $(\hat{\mu}(\boldZ_i),\hat{\sigma}^2(\boldZ_i)) = (\hat{\mu}_i, \hat{\sigma}_i^2)$ is sufficient for $(\mu_i,\sigma_i^2)$. 
Integrating over the distribution of $\mu_i$ given $\mu_i \neq 0, \sigma_i^2$, we find that,
$$
(\hat{\mu}_i, \hat{\sigma}_i^2) \mid\mu_i\neq0,\sigma_i^2
\,\sim\,\mathrm{N}(0,(1+\lambda)\sigma_i^2/K) \otimes \frac{\sigma_i^2}{\nu}\chi^2_\nu.
$$
Multiplying their densities gives
$$
\begin{aligned}
p(\hat{\mu},\hat{\sigma}^2\mid\mu_i\neq0,\sigma_i^2=\sigma^2)
=\frac{\sqrt K(\nu/2)^{\nu/2}(\hat{\sigma}^2)^{\nu/2-1}}
{\sqrt{2\pi(1+\lambda)}\,\Gamma(\nu/2)}
(\sigma^2)^{-K/2}
\exp\!\cb{-\frac{\nu\hat{\sigma}^2+K\hat{\mu}^2/(1+\lambda)}{2\sigma^2}}.
\end{aligned}
$$
The null density $p(\hat{\mu},\hat{\sigma}^2\mid\mu_i=0,\sigma_i^2=\sigma^2)$ is the same expression with $\lambda=0$. By~\eqref{eq:limma_variance_prior}, the prior density of $\sigma^2$ is
$$
\frac{(\nu_0s_0^2/2)^{\nu_0/2}}{\Gamma(\nu_0/2)}
(\sigma^2)^{-\nu_0/2-1}\exp\!\p{-\frac{\nu_0s_0^2}{2\sigma^2}},\;\; \sigma^2>0.
$$
Thus:
$$
\begin{aligned}
\frac{p(\hat{\mu},\hat{\sigma}^2\mid\mu_i\neq0)}
{p(\hat{\mu},\hat{\sigma}^2\mid\mu_i=0)}
&=\frac{1}{\sqrt{1+\lambda}}
\frac{\displaystyle\int_0^\infty(\sigma^2)^{-(\nu_0+K)/2-1}
\exp\!\cb{-\frac{\nu_0s_0^2+\nu\hat{\sigma}^2+K\hat{\mu}^2/(1+\lambda)}{2\sigma^2}}\dd\sigma^2}
{\displaystyle\int_0^\infty(\sigma^2)^{-(\nu_0+K)/2-1}
\exp\!\cb{-\frac{\nu_0s_0^2+\nu\hat{\sigma}^2+K\hat{\mu}^2}{2\sigma^2}}\dd\sigma^2}\\
&=\frac{1}{\sqrt{1+\lambda}}
\p{\frac{\nu_0s_0^2+\nu\hat{\sigma}^2+K\hat{\mu}^2}
{\nu_0s_0^2+\nu\hat{\sigma}^2+K\hat{\mu}^2/(1+\lambda)}}^{(\nu_0+K)/2}.
\end{aligned}
$$
The last equality uses $\int_0^\infty x^{-a-1}e^{-b/x}\dd x=\Gamma(a)b^{-a}$ for $a,b>0$. By sufficiency and Bayes' rule,
$$
\mathrm{lfdr}(\boldz) = \PP{\mu_i=0 \mid \hat{\mu}_i = \hat{\mu}(\boldz),\; \hat{\sigma}_i^2 = \hat{\sigma}^2(\boldz)}
=\p{1+\frac{\pi_1}{\pi_0}
\frac{p(\hat{\mu},\hat{\sigma}^2\mid\mu_i\neq0)}
{p(\hat{\mu},\hat{\sigma}^2\mid\mu_i=0)}}^{-1}.
$$
Substituting $K\hat{\mu}^2/(\nu_0s_0^2+\nu\hat{\sigma}^2)=T^2(\boldz;\nu_0,s_0)/(\nu_0+\nu)$ gives the desired formula.
\end{proof}

In what follows, we study the null ($t_{\nu+\nu_0}$) and non-null distributions ($\sqrt{1+\lambda}\cdot t_{\nu+\nu_0}$) of the moderated t-statistic $T(\boldZ_i; \nu_0, s_0)$ that were derived in Proposition~\ref{prop:moderated_two_groups} in more detail. Specifically we use short hand notation:
\begin{itemize}
\item $f_0$ and $\bar{F}_0$ for the density, resp. survival function of $t_{\nu+\nu_0}$ (omitting $\nu$, $\nu_0$ from our notation);
\item $f_1$ and $\bar{F}_1$ for the density, resp. survival function of $\sqrt{1+\lambda}\cdot t_{\nu+\nu_0}$ (omitting $\nu$, $\nu_0$ and $\lambda$ from our notation).
\end{itemize}
By virtue of Proposition~\ref{prop:localfdr_formula} and some abuse of notation (wherein we write $\mathrm{lfdr}(\cdot)$ with both arguments $\boldz \in \RR^K$ and $z \in \RR$)
, we may write:
\begin{equation*}
\begin{aligned}
&\mathrm{lfdr}(\boldz) \equiv \mathrm{lfdr}(T(\boldz; \nu_0, s_0)), \\
&\text{where }\,\mathrm{lfdr}(s) := \PP{\mu_i =0 \mid T(\boldZ_i; \nu_0,s_0) = s} = \frac{ \pi_0 f_0(s)}{\pi_0 f_0(s) + (1-\pi_0) f_1(s)}.
\end{aligned}
\end{equation*}
More concretely, Proposition~\ref{prop:localfdr_formula} establishes that $\PP{\mu_i =0 \mid \boldZ_i}$ is measurable with respect to $T(\boldZ_i; \nu_0, s_0)$ and so must be equal to $\PP{\mu_i =0 \mid T(\boldZ_i; \nu_0, s_0)}$.

Let us also recall the notation $\mathrm{Fdr}_n(s)$ for $s \geq 0$ from~\eqref{eq:mfdr_equilibrium_rewritten}:
$$
\mathrm{Fdr}_n(s) \equiv \mathrm{Fdr}(s) := \PP{ \mu_i =0 \,\mid \, |T(\boldZ_i; \nu_0, s_0)| \geq s   } = \frac{\pi_0\bar F_0(s)}
{\pi_0\bar F_0(s)+\pi_1\bar F_1(s)},
$$
and where in this section we omit the dependence on $n$ (which in turn is via dependence on $\pi_0, \pi_1, \lambda$ which can depend on $n$). We have the following monotonicity result.

\begin{lemm}[Monotonicity]
Suppose that $\pi_0 \in (0,1)$ and $\lambda >0$. Then the functions
$$
s \mapsto \frac{f_0(s)}{f_1(s)},\;\;\;\; s \mapsto \frac{\bar{F}_0(s)}{\bar{F}_1(s)},\;\;\;\; s \mapsto \mathrm{lfdr}(s), \;\;\;\; s \mapsto \mathrm{Fdr}(s)
$$
are strictly decreasing on $(0, \infty)$. 
\label{lemm:monotonicity_finite_sample}
\end{lemm}

\begin{proof}
Notice that
$$
\frac{f_1(s)}{f_0(s)}=\frac{1}{\sqrt{1+\lambda}}\p{\frac{1+s^2/(\nu+\nu_0)}{1+s^2/\{(\nu+\nu_0)(1+\lambda)\}}}^{(\nu+\nu_0+1)/2}
$$
is strictly increasing for $s>0$, since $\lambda>0$. Thus,
$$
\frac{\bar{F}_1(s)}{\bar{F}_0(s)}
=\frac{1}{\bar{F}_0(s)}\int_s^\infty\frac{f_1(u)}{f_0(u)}f_0(u)\dd u
>\frac{f_1(s)}{f_0(s)}.
$$
Consequently,
$$
\frac{\dd}{\dd s}\frac{\bar{F}_0(s)}{\bar{F}_1(s)}
=\frac{f_1(s)\bar{F}_0(s)-f_0(s)\bar{F}_1(s)}{\bar{F}_1(s)^2}<0,
$$
as claimed.

The claim for $\mathrm{lfdr}$ and $\mathrm{Fdr}$ follow.
\end{proof}

We conclude with some facts involving $F_{K,\nu_0}$, that is, the $F$ distribution with $K$ and $\nu_0$ degrees of freedom.

\begin{prop}
\label{prop:second_moment_f}
Under~\eqref{eq:normal_samples} and~\eqref{eq:lonnsted_prior}, with $\pi_0>0$, it holds that
\begin{equation}
\frac{\hat\tau_i^2}{s_0^2}\mid\mu_i=0\,\sim\,F_{K,\nu_0}.
\label{eq:second_moment_f}
\end{equation}.
\end{prop}

\begin{proof}
Conditional on $\mu_i=0$, we can write
$$
\p{\frac{\hat\tau_i^2}{\sigma_i^2},\,\frac{\sigma_i^2}{s_0^2}}\mid(\mu_i=0)
\;\,\stackrel{\mathcal D}{=}\;\,\p{\frac{V}{K},\,\frac{\nu_0}{V_0}},
$$
where $V\sim\chi^2_K$ and $V_0\sim\chi^2_{\nu_0}$ are independent. Consequently,
$$
\frac{\hat\tau_i^2}{s_0^2}\mid(\mu_i=0)
\;\,\stackrel{\mathcal D}{=}\;\,\frac{V/K}{V_0/\nu_0}\sim F_{K,\nu_0}.
$$
\end{proof}

\begin{lemm}
Let $Q_\nu(p)$ denote the $p$-quantile of $F_{K,\nu}$, the $F$-distribution with $K$ and $\nu$ degrees of freedom.
For any $0<p_1<p_2<1$, let $R(\nu):=Q_\nu(p_2)/Q_\nu(p_1)$. Then $R(\cdot)$
is strictly decreasing on $(0,\infty)$.\footnote{This proof is due to GPT-5.6 and was checked by the authors.}
\label{lemm:quantile_ratio_monotone}
\end{lemm}

\begin{proof}
The argument follows~\citet{saunders1978quantiles}.
Let $F_\nu$ and $f_\nu$ denote the CDF and density of the
$F_{K,\nu}$ distribution. Fix $\nu_2>\nu_1>0$ and define
$$
H(x):=F_{\nu_2}(cx)-F_{\nu_1}(x),\;\;\text{ where }\; c:=Q_{\nu_2}(p_1)\big / Q_{\nu_1}(p_1).
$$
By definition, $H(0)=0$, $H(Q_{\nu_1}(p_1))=0$, and $\lim_{x\to\infty}H(x)=0$.
We will show that $Q_{\nu_1}(p_1)$ is the unique zero of $H$ in $(0,\infty)$ and that $H(x) >0$ for $x>Q_{\nu_1}(p_1)$. To this end, write
$$
H'(x) = c f_{\nu_2}(cx) - f_{\nu_1}(x)
=
f_{\nu_1}(x)\cb{r(x)-1},
\;\;\;
r(x):=\frac{cf_{\nu_2}(cx)}{f_{\nu_1}(x)}.
$$
Notice that
$$
2\log f_{\nu_1}(x) = (K-2)\log x - (K + \nu_1) \log(Kx + \nu_1) + \mathrm{const},
$$
where $\mathrm{const}$ does not depend on $x$ (but depends on $K,\nu_1$). Hence
$$
\frac{2}{K}(\log f_{\nu_1})'(x) = \frac{K-2}{Kx} - \frac{K+\nu_1}{Kx + \nu_1}.
$$
Differentiation of $\log r$ yields,
$$
\frac{2}{K}(\log r)'(x)
=
\frac{
\nu_2(K+\nu_1)-c\nu_1(K+\nu_2)
+cK(\nu_1-\nu_2)x
}{
(\nu_1+Kx)(\nu_2+cKx)
}.
$$
The denominator is strictly positive and strictly increasing, while 
the numerator is strictly decreasing in $x$. Hence $r$ is
either strictly decreasing, or first strictly increasing and then
strictly decreasing. Moreover,
$$
\begin{aligned}
2\log r(x)
&=(K+\nu_1)\log(Kx+\nu_1)-(K+\nu_2)\log(cKx+\nu_2)+\mathrm{const}\\
&=(\nu_1-\nu_2)\log x+(K+\nu_1)\log(K+\nu_1/x)-(K+\nu_2)\log(cK+\nu_2/x)+\mathrm{const},
\end{aligned}
$$
and since $\nu_2>\nu_1$, we get $\log r(x) \to -\infty$ as $x\to\infty$ and so $r(x)\to0$.
Thus the sign of
$H'$ can only have one of the forms
$$
-,\;\;\;\,\text{or}\;\;\;\, +,-, \;\;\;\,\text{or}\;\;\;\, -,+,-.
$$
The first possibility is incompatible with
the two roots of $H$ we exhibited above. The second is also impossible: in that case $H$
first increases from zero and then decreases to its limit zero, and
therefore $H(x)>0$ for every $x>0$. 

The only remaining
possibility is
$
-,+,-
$, and so $H$ first decreases from zero to a negative local minimum, then
increases, and finally decreases to zero. Since the final segment is
decreasing to zero, its local maximum must be positive. Hence $H$
crosses zero exactly once on its increasing segment. Since
$H(Q_{\nu_1}(p_1))=0$, this unique crossing occurs at $Q_{\nu_1}(p_1)$, and therefore
$$
H(x)>0\quad\text{for }x>Q_{\nu_1}(p_1).
$$
Now $Q_{\nu_1}(p_2)>Q_{\nu_1}(p_1)$, and hence
$
H\p{Q_{\nu_1}(p_2)}>0.
$
Equivalently,
$
F_{\nu_2}\p{cQ_{\nu_1}(p_2)}>p_2.
$
Since $F_{\nu_2}$ is strictly increasing,
$
Q_{\nu_2}(p_2)<cQ_{\nu_1}(p_2).
$
Substituting the definition of $c$ yields
$$
R(\nu_2)=\frac{Q_{\nu_2}(p_2)}{Q_{\nu_2}(p_1)}
<
\frac{Q_{\nu_1}(p_2)}{Q_{\nu_1}(p_1)} = R(\nu_1).
$$
Since $\nu_2>\nu_1$ were arbitrary, the result follows.
\end{proof}

\subsection{Tail behavior of the t-distribution}

We have the following result.
\begin{lemm}
\label{lemm:t_tails}
For each fixed $\nu>0$,
$$
\lim_{x\to\infty}x^\nu\bar F_{t,\nu}(x)
=L_\nu
:=\frac{\nu^{\nu/2-1}\Gamma((\nu+1)/2)}
{\sqrt{\pi}\,\Gamma(\nu/2)}.
$$
\end{lemm}

\begin{proof}
By~\eqref{eq:t_density}, $x^{\nu+1}f_{t,\nu}(x)\to\nu L_\nu$. L'H\^opital's rule gives
$$
\lim_{x\to\infty}\frac{\bar F_{t,\nu}(x)}{x^{-\nu}}
=\lim_{x\to\infty}\frac{f_{t,\nu}(x)}{\nu x^{-\nu-1}}=L_\nu.
$$
\end{proof}
\begin{lemm}
\label{lemm:t_tail_rescaling}
It holds that:
$$ \lim_{a \to 1}\sup_{x\geq0}\abs{\frac{\bar F_{t,\nu}(ax)}{\bar F_{t,\nu}(x)}-1}=0. $$
\end{lemm}

\begin{proof}
We start by showing an inequality which implies the desired conclusion.
For each fixed $\nu>0$ and any $a>0$,
$$
\sup_{x\geq0}\abs{\frac{\bar F_{t,\nu}(ax)}{\bar F_{t,\nu}(x)}-1}
\leq \abs{a^{-\nu}-1}.
$$
For $a\geq1$, the density in~\eqref{eq:t_density} satisfies $af_{t,\nu}(au)\geq a^{-\nu}f_{t,\nu}(u)$. Integrating gives
$$
a^{-\nu}\bar F_{t,\nu}(x)
\leq\int_x^\infty af_{t,\nu}(au)\dd u
=\bar F_{t,\nu}(ax)
\leq\bar F_{t,\nu}(x).
$$
For $0<a<1$, apply this bound with $1/a$ and $ax$ to obtain
$1\leq\bar F_{t,\nu}(ax)/\bar F_{t,\nu}(x)\leq a^{-\nu}$.

\end{proof}

\subsection{Some limits and results in the sparse regime}

We work in the sparse regime of Section~\ref{sec:power_analysis}: as $n\to\infty$,
$$
\pi_1\to0,\;\; n\pi_1\to\infty,\;\; \lambda\to\infty,\;\;
\pi_1\lambda^{(\nu+\nu_0)/2}\to c_\star\in(0,\infty),
$$
with $K$, $\nu_0$, and $s_0^2$ fixed, and $\nu=K-1$. Recall that $\pi_0=1-\pi_1$, $\bar F_0(s)=\bar F_{t,\nu+\nu_0}(s)$, and $\bar F_1(s)=\bar F_0(s/\sqrt{1+\lambda})$. For $\tau>0$, let
$$
\zeta_n(\tau):=\p{\frac{L_{\nu+\nu_0}}{\tau\pi_1}}^{1/(\nu+\nu_0)},\;\;
\Psi_1(\tau):=\bar F_0\p{\cb{\frac{L_{\nu+\nu_0}}{\tau c_\star}}^{1/(\nu+\nu_0)}},\;\;
\mathrm{Fdr}_\infty(\tau):=\frac{\tau}{\tau+\Psi_1(\tau)},
$$
where $L_{\nu+\nu_0}$ is given in Lemma~\ref{lemm:t_tails} (by taking $\nu+\nu_0$ in the statement therein).

\begin{lemm}
\label{lemm:sparse_limits}
For every fixed $\tau>0$, as $n\to\infty$,
$$
\frac{\bar F_0(\zeta_n(\tau))}{\pi_1}\to\tau,\;\;
\bar F_1(\zeta_n(\tau))\to\Psi_1(\tau),\;\;
\mathrm{Fdr}_n(\zeta_n(\tau))\to\mathrm{Fdr}_\infty(\tau).
$$
\end{lemm}

\begin{proof}
Since $\zeta_n(\tau)\to\infty$ and $\pi_1\zeta_n(\tau)^{\nu+\nu_0}=L_{\nu+\nu_0}/\tau$, Lemma~\ref{lemm:t_tails} gives the first limit. Moreover,
$$
\frac{\zeta_n(\tau)}{\sqrt{1+\lambda}}
=\p{\frac{L_{\nu+\nu_0}}{\tau\pi_1(1+\lambda)^{(\nu+\nu_0)/2}}}^{1/(\nu+\nu_0)}
\to\p{\frac{L_{\nu+\nu_0}}{\tau c_\star}}^{1/(\nu+\nu_0)},
$$
so continuity of $\bar F_0$ gives the second limit. The third follows from the first two and $\pi_0\to1$.
\end{proof}

\begin{lemm}
The function $\tau\mapsto\mathrm{Fdr}_\infty(\tau)$ is continuous and strictly increasing on $(0,\infty)$, with
$$
\lim_{\tau\downarrow0}\mathrm{Fdr}_\infty(\tau)=\frac{1}{1+c_\star},\;\;
\lim_{\tau\to\infty}\mathrm{Fdr}_\infty(\tau)=1.
$$
Consequently, $\mathrm{Fdr}_\infty(\tau_\alpha)=\alpha$ has a unique solution $\tau_\alpha>0$ if and only if $1/(1+c_\star)<\alpha<1$.
\label{lemm:monotonicity_infty}
\end{lemm}

\begin{proof}
Write $f_0=f_{t,\nu+\nu_0}$. Notice that, for $u>0$,
$$
u^{\nu+\nu_0+1}f_0(u)
=C\p{\frac{u^2}{1+u^2/(\nu+\nu_0)}}^{(\nu+\nu_0+1)/2},
$$
where $C>0$ is the normalizing constant. The above expression is strictly increasing in $u$. Thus, for $x>0$,
$$
\bar{F}_0(x)=\int_x^\infty f_0(u)\dd u
>f_0(x)\int_x^\infty\p{\frac{x}{u}}^{\nu+\nu_0+1}\dd u
=\frac{xf_0(x)}{\nu+\nu_0}.
$$
Now write $x_\tau:=\{L_{\nu+\nu_0}/(\tau c_\star)\}^{1/(\nu+\nu_0)}$. Since $\Psi_1(\tau)=\bar{F}_0(x_\tau)$ and $\frac{\dd}{\dd\tau}x_\tau=-x_\tau/\{(\nu+\nu_0)\tau\}$, we have
$$
\tau\frac{\dd}{\dd\tau}\Psi_1(\tau)=\frac{x_\tau f_0(x_\tau)}{\nu+\nu_0}<\bar{F}_0(x_\tau)=\Psi_1(\tau).
$$
Consequently,
$$
\frac{\dd}{\dd\tau}\mathrm{Fdr}_\infty(\tau)
=\frac{\Psi_1(\tau)-\tau\frac{\dd}{\dd\tau}\Psi_1(\tau)}{\{\tau+\Psi_1(\tau)\}^2}>0.
$$
Continuity is immediate. The endpoint limits follow since $\Psi_1(\tau)/\tau\to c_\star$ as $\tau\downarrow0$ by Lemma~\ref{lemm:t_tails}, and $\Psi_1(\tau)\to1/2$ as $\tau\to\infty$.
\end{proof}

\section{Proofs}
\label{sec:proofs}

\subsection{Proof of Theorem~\ref{theo:compound-p}}

\begin{proof}
Since $\widehat{S}(\cdot)$ is learned only from $\mathcal{O}_{1:n}$, we can treat it as fixed when conditioning on $\mathcal{O}_{1:n}$.
We first record the conditional distribution of $\boldZ_i$ under the null. For any bounded measurable function $f:\mathcal{Z}\to\RR$, define
$$
Qf(\boldz):=\int_{\mathcal{H}} f(H\boldz)\dd\mu_{\mathcal{H}}(H).
$$
By right-invariance, for every $G\in\mathcal{H}$,
$$
Qf(G\boldz)
=
\int_{\mathcal{H}} f(HG\boldz)\dd\mu_{\mathcal{H}}(H)
=
Qf(\boldz),
$$
so $Qf(\boldZ_i)$ depends on $\boldZ_i$ only through its orbit $\mathcal{O}_i$. If $H_{0,i}$ holds, then for every bounded measurable $g$,
\begin{align*}
\EE{g(\mathcal{O}_i)Qf(\boldZ_i)}
=
\int_{\mathcal{H}}
\EE{g(\mathcal{O}_i)f(H\boldZ_i)}
\dd\mu_{\mathcal{H}}(H)=
\EE{g(\mathcal{O}_i)f(\boldZ_i)},
\end{align*}
where the last equality follows from $\boldZ_i\stackrel{\mathcal D}=H\boldZ_i$ and
$\mathcal{H}\cdot(H\boldZ_i)=\mathcal{H}\cdot\boldZ_i$. Thus
\begin{equation}
\EE{f(\boldZ_i)\mid\mathcal{O}_i}
= Qf(\boldZ_i) = 
\int_{\mathcal{H}} f(H\boldZ_i)\dd\mu_{\mathcal{H}}(H)
\;\text{ almost surely.}
\label{eq:haar_conditional_orbit}
\end{equation}
For $k\in \{1,\ldots,n\}$, let
$$
\bar F_k(s)
:=
\int_{\mathcal{H}}
\ind\cb{\widehat{S}(H\boldZ_k)\geq s}
\dd\mu_{\mathcal{H}}(H),
\;\;
\bar F(s):=\frac{1}{n}\sum_{k=1}^n \bar F_k(s).
$$
The function $\bar F_k$ depends on $\boldZ_k$ only through its orbit. Indeed, for any $G\in\mathcal{H}$ and using right-invariance again,
$$
\int_{\mathcal{H}}
\ind\cb{\widehat{S}(HG\boldZ_k)\geq s}
\dd\mu_{\mathcal{H}}(H)
=
\bar F_k(s).
$$
Consequently, $\bar F_1,\ldots,\bar F_n$ and $\bar F$ are fixed conditionally on $\mathcal{O}_{1:n}$.
For every null $i$, we get
\begin{equation}
\PP{\widehat{S}(\boldZ_i)\geq s\mid\mathcal{O}_{1:n}}
=
\bar F_i(s),
\;\;
\PP{\widehat{S}(\boldZ_i)>s\mid\mathcal{O}_{1:n}}
=
\lim_{u\downarrow s}\bar F_i(u).
\label{eq:conditional_score_survival}
\end{equation}
Fix $t\in(0,1)$ and define
$$
s^*(t):=\inf\cb{s\in\RR:\bar F(s)\leq t}.
$$
By~\eqref{eq:compound_pvalue_fun},
$
\pcomp_i=\bar F(\widehat{S}(\boldZ_i)).
$
Notice that 
$\pcomp_i \leq t$ if and only if $\widehat{S}(\boldZ_i) \geq s^*(t)$ (when $\bar{F}(s^*(t))\leq t$) or if and only if $\widehat{S}(\boldZ_i)> s^*(t)$ (when $\bar{F}(s^*(t))>t$). Thus let us take two cases.
If $\bar F(s^*(t))\leq t$, then
\begin{align*}
\sum_{i:H_i=0}\PP{\pcomp_i\leq t\mid\mathcal{O}_{1:n}}
=
\sum_{i:H_i=0}\bar F_i(s^*(t))
\leq
\sum_{i=1}^n\bar F_i(s^*(t))
=
n\bar F(s^*(t))
\leq nt.
\end{align*}
If instead $\bar F(s^*(t))>t$, then
\begin{align*}
\sum_{i:H_i=0}\PP{\pcomp_i\leq t\mid\mathcal{O}_{1:n}}
=
\sum_{i:H_i=0}\lim_{u\downarrow s^*(t)}\bar F_i(u)\leq
n\lim_{u\downarrow s^*(t)}\bar F(u)
\leq nt,
\end{align*}
where the last inequality follows because $\bar F(u)\leq t$ for every $u>s^*(t)$ by the definition of $s^*(t)$. Thus, conditionally on $\mathcal{O}_{1:n}$,
$$
\sum_{i:H_i=0}\PP{\pcomp_i\leq t\mid\mathcal{O}_{1:n}}
\leq nt
$$
for every $t\in(0,1)$. We have shown that the $\pcomp_i$ are compound p-values conditionally on $\mathcal{O}_{1:n}$. 

Finally, to prove that the $\pcomp_i$, $i\in [n]$ are independent conditionally on $\mathcal{O}_{1:n}$, we use that for any independent variables $(T_i)_{i\in [n]}$, the variables 
$(T_i)_{i\in [n]}$ are also independent conditionally on $(\Psi_i(T_i))_{i\in [n]}$ for any measurable function $\Psi_i(\cdot)$.

The proof of FDR control follows by applying Theorem~\ref{theo:bh} conditionally on $\mathcal{O}_{1:n}$, and then using iterated expectation.

\end{proof}

\subsection{Proof of Theorem~\ref{theo:seqstep_multi}}

\begin{proof}
Condition on $\mathcal{O}_{1:n}$, so that $\widehat{S}(\cdot)$ is fixed. The $\tilde S_i$ are independent by the same argument as in the proof of Theorem~\ref{theo:compound-p}. The separate p-values in~\eqref{eq:sep_pvalue_fun} are valid conditionally on $\mathcal{O}_{1:n}$. For every null $i$, we therefore have
$$
\PP{\tilde S_i>0\mid\mathcal{O}_{1:n}}
=\PP{\psep_i\leq1/(\kappa+1)\mid\mathcal{O}_{1:n}}
\leq\frac{1}{\kappa+1}.
$$
The proof of FDR control now is an application of~\citet[Theorem 3]{barber2015controlling} arguing as in the proof of Proposition~3.3 of~\citet{gimenez2019improving}
\end{proof}

\subsection{Proof of Proposition~\ref{prop:sep_pvalue_limma}}

\begin{proof}
Condition on $\mathcal{O}_{1:n}$, so that $\hat{\nu}_0$ and $\hat{s}_0^2$ are fixed. For both choices of $\mathcal{H}$, notice that $\Norm{\boldz}_2^2 = \Norm{\boldZ_i}_2^2$ for all $\boldz \in \mathcal{O}_i$. Thus,
$$
T^2(\boldz; \hat{\nu}_0, \hat{s}_0) = \frac{(\nu+\hat{\nu}_0)K\hat{\mu}^2(\boldz)}{\hat{\nu}_0\hat{s}_0^2 + \Norm{\boldZ_i}_2^2 - K\hat{\mu}^2(\boldz)},\;\;\;\;
T^2(\boldz; 0, 1) = \frac{\nu K\hat{\mu}^2(\boldz)}{\Norm{\boldZ_i}_2^2 - K\hat{\mu}^2(\boldz)}.
$$
Both are strictly increasing functions of $K\hat{\mu}^2(\boldz)$ on the orbit. Hence they induce the same comparisons in~\eqref{eq:sep_pvalue_fun}, proving the first claim.

For the second claim, we continue to condition on $\mathcal{O}_{1:n}$.
Let $H$ be Haar distributed on $O(K)$. Then  $H\boldZ_i$ is uniform on the sphere of radius $\Norm{\boldZ_i}_2$, and so has the same distribution as $\Norm{\boldZ_i}_2\boldG/\Norm{\boldG}_2$, where $\boldG \sim \mathrm{N}(0,I_K)$. Since $|T(\cdot;0,1)|$ is invariant to rescaling, we thus have, conditionally on $\boldZ_i$,
$$
T(H\boldZ_i;0,1) \mid \boldZ_i \;\stackrel{d}{=} T(\boldG;0,1)\sim t_\nu.
$$
Combining this with the first claim, we conclude.
\end{proof}

\subsection{Proof of Proposition~\ref{prop:orthogonal_rotation_cpvalue}}

\begin{proof}
For fixed $\boldz$ and $i$, let us define:
$$h_i(\boldz) := \int_{O(K)} \ind\cb{ \widehat{S}(H \boldZ_i) \geq \widehat{S}(\boldz)} \dd\mu_{\mathcal{H}}(H).$$
To compute the above, note that we need to integrate the indicator of the following event:
$$
\frac{ (\nu + \hat{\nu}_0)K \hat{\mu}^2(H\boldZ_i) }{ \nu \hat{\sigma}^2(H\boldZ_i) + \hat{\nu}_0 \hat{s}_0^2} \geq \widehat{S}^2(\boldz).
$$
Now note that
$$ \nu \hat{\sigma}^2(H\boldZ_i) = \Norm{H\boldZ_i}_2^2 - K\hat{\mu}^2(H\boldZ_i) = \Norm{\boldZ_i}_2^2  - K\hat{\mu}^2(H\boldZ_i).$$
Thus, by rearranging, we find that:
$$
K\hat{\mu}^2(H\boldZ_i) \geq \frac{\widehat{S}^2(\boldz)/\{\nu + \hat{\nu}_0\}}{ 1 +\widehat{S}^2(\boldz)/\{\nu + \hat{\nu}_0\}}\p{ \Norm{\boldZ_i}^2_2 + \hat{\nu}_0 \hat{s}_0^2}. 
$$
Now write $\hat{\mu}(H\boldZ_i)  = \mathbf{1}^\top H\boldZ_i/K$, then using again that $\Norm{H\boldZ_i}_2 = \Norm{\boldZ_i}_2$,
$$
\frac{ (\mathbf{1}^\top H\boldZ_i/\sqrt{K})^2}{ \Norm{H\boldZ_i}^2_2} \geq \frac{\widehat{S}^2(\boldz)/\{\nu + \hat{\nu}_0\}}{ 1 +\widehat{S}^2(\boldz)/\{\nu + \hat{\nu}_0\}}\p{ 1 + \frac{\hat{\nu}_0 \hat{s}_0^2}{\Norm{\boldZ_i}^2_2}}. 
$$
Thus
$$
h_i(z) = \bar{F}_{B, 1/2, \nu/2}\cb{ \frac{\widehat{S}^2(\boldz)/\{\nu + \hat{\nu}_0\}}{ 1 +\widehat{S}^2(\boldz)/\{\nu + \hat{\nu}_0\}}\p{ 1 + \frac{\hat{\nu}_0 \hat{s}_0^2}{\Norm{\boldZ_i}^2_2}}}.
$$
We conclude by noting that $\pcompfun(\boldz; \widehat{S}(\cdot), \mathcal{O}_{1:n}, \mathcal{H}_{\mathrm{orth}}) = n^{-1}\sum_{i=1}^n h_i(\boldz)$.
\end{proof}

\subsection{Preliminary lemmata}

\begin{lemm}
Let $\hat{F}_1,\ldots,\hat{F}_n:\RR\to[0,1]$ be independent random functions, and suppose that $F_i(t):=\EEInline{\hat{F}_i(t)}$ is continuous for each $i$. Assume either that all $\hat{F}_i(\cdot)$ are nondecreasing and tend to $0$ at $-\infty$ almost surely, or that all are nonincreasing and tend to $0$ at $+\infty$ almost surely. Define
$$
\hat{F} = \frac{1}{n}\sum_{i=1}^n \hat{F}_i,\;\;\;\; F= \frac{1}{n}\sum_{i=1}^n F_i.
$$
Then, for any $\varepsilon, \delta \in (0,1)$
$$
\PP{1-\varepsilon \leq \frac{\hat{F}(t)}{F(t)} \leq 1+\varepsilon  \, \text{ for all }\, t \, \text{ with } F(t) \geq \delta } \geq 
1-
\frac{
4\exp\!\p{-n\varepsilon^2\delta/24}
}{
1-\exp\!\p{-n\varepsilon^3\delta/72}
}.
$$
\label{lemm:gc_multiplicative_compound}
\end{lemm}
\begin{proof}
It suffices to prove the nondecreasing case. Indeed, in the nonincreasing case, the functions $t\mapsto\hat{F}_i(-t)$ are nondecreasing and tend to $0$ as $t\to-\infty$, with expectations $F_i(-t)$. Applying the nondecreasing case to these functions gives the same event in the statement, since $-t$ ranges over all of $\RR$.

Write $\hat{F}_i(\infty):=\lim_{t\to\infty}\hat{F}_i(t)$ and $F(\infty):=\lim_{t\to\infty}F(t)$. Dominated convergence gives $\EEInline{\hat{F}(\infty)}=F(\infty)$. Fix a single $t\in\mathbb R\cup\{\infty\}$ such that $F(t)>0$. Notice that
$$
\Var{\hat{F}(t)} \leq \frac{1}{n^2}\sum_{i=1}^n \EE{\hat{F}_i(t)} \leq \frac{1}{n}F(t).
$$
By Bernstein's inequality we have that:
$$
\PP{\hat{F}(t) \geq (1+\varepsilon) F(t)} \leq \exp\!\p{  - \frac{n^2 \varepsilon^2 F^2(t)/2}{ nF(t) + n \varepsilon F(t)/3}} = \exp\!\p{ -\frac{ n \varepsilon^2 F(t)/2  }{ 1 + \varepsilon/3 }} \leq \exp\!\p{ -\frac{3}{8} n \varepsilon^2 F(t)}.
$$
Analogously we get:
$$
\PP{\hat{F}(t) \leq (1-\varepsilon) F(t)} \leq \exp\!\p{ -\frac{3}{8} n \varepsilon^2 F(t)}.
$$
We may assume that $F(t)\geq\delta$ for some finite $t$, since otherwise the conclusion is immediate. Next we do a peeling-type argument to get simultaneous control. Let us define for $\eta = \varepsilon/3>0$
$
\delta_j := (1+\eta)^j\delta. 
$
Let $J:=\min\{j\geq0:\delta_j>F(\infty)\}$, set $t_j:=\inf\{t:F(t)\geq\delta_j\}$ for $j<J$, and set $t_J:=\infty$, with $\inf\varnothing=\infty$. By continuity, $F(t_j)=\delta_j$ for $j<J$, and $F(t_J)=F(\infty)$. 
Define the event,
$$
A := \bigcap_{j=0}^{J} A_j,\;\text{ where }\; A_j:=\Big\{1-\eta \leq \frac{\hat{F}(t_j)}{F(t_j)} \leq 1+\eta\Big\}.
$$
We will first show that $A$ has high probability:
$$
\begin{aligned}
\PP{A^c}
&\leq 2\sum_{j=0}^{J-1}\exp\!\p{-\frac38n\eta^2\delta_j}
+2\exp\!\p{-\frac38n\eta^2F(\infty)}\\
&\leq4\sum_{j=0}^{\infty}\exp\!\p{-\frac38n\eta^2(1+\eta)^j\delta}
\leq\frac{4\exp\!\p{-\frac38n\eta^2\delta}}{1-\exp\!\p{-\frac38n\eta^3\delta}}.
\end{aligned}
$$
The last step uses the fact that $\sum_{j\geq 0} \exp\{-a (1+\eta)^j\}\leq \sum_{j\geq 0} \exp\{-a (1+j\eta)\}= \exp(-a)/\{1-\exp(-a\eta)\}$ for any $a>0$.

Now we interpolate on the event $A$. Every finite $t$ with $F(t)\geq\delta$ lies in $[t_j,t_{j+1})$ for some $j<J$. Since $F(t_{j+1})\leq(1+\eta)F(t_j)$, including at the endpoint $t_J=\infty$, we have:
$$
\frac{\hat{F}(t)}{F(t)} \leq \frac{\hat{F}(t_{j+1})}{F(t_j)} \leq \frac{(1+\eta)F(t_{j+1})}{F(t_j)} \leq (1+\eta)^2 \leq 1+\varepsilon. 
$$
Analogously, $\hat{F}(t)/F(t) \geq 1-\varepsilon$.
\end{proof}

\begin{lemm}[Multiplicative control of moderated t-statistics]
\label{lemm:moderated_multiplicative}

Fix $\varepsilon \in (0,1)$. Suppose that
$$
 1-\varepsilon \leq \frac{\nu_0 + \nu}{\hat{\nu}_0 +\nu} \leq 1+\varepsilon,\;\;\;\,  1-\varepsilon \leq \frac{\nu_0 + \nu}{\hat{\nu}_0 +\nu} \cdot \frac{\hat{\nu}_0\hat{s}_0^2}{\nu_0 s_0^2} \leq 1+\varepsilon.
$$
Then,
$$
\sup_{z} \abs{\frac{T^2(z; \nu_0, s_0)}{T^2(z; \hat{\nu}_0, \hat{s}_0)} - 1} \leq \varepsilon.
$$
\end{lemm}

\begin{proof}
Write,
$$
\begin{aligned}
\frac{T^2(z; \nu_0, s_0)}{T^2(z; \hat{\nu}_0, \hat{s}_0)} &= \frac{ \hat{\sigma}^2(z; \hat{\nu}_0, \hat{s}_0)}{\hat{\sigma}^2(z; \nu_0, s_0)} \\ 
&= \frac{ (\nu_0 + \nu)(\hat{\nu}_0  \hat{s}_0^2 + \nu \hat{\sigma}^2(z))}{(\hat{\nu}_0+\nu)(\nu_0  s_0^2 + \nu \hat{\sigma}^2(z))} \\ 
& = \frac{\nu_0 + \nu}{\hat{\nu}_0 +\nu}\cdot \frac{\nu \hat{\sigma}^2(z)}{\nu_0s_0^2 + \nu \hat{\sigma}^2(z)} \,+\,  \frac{\nu_0 + \nu}{\hat{\nu}_0 +\nu}\cdot \frac{\hat{\nu}_0\hat{s}_0^2}{\nu_0 s_0^2} \cdot  \frac{\nu_0 s_0^2} {\nu_0s_0^2 + \nu \hat{\sigma}^2(z)}.
\end{aligned}
$$
The conclusion follows.
\end{proof}

\begin{lemm}[Consistency of $\hat{s}_0^2$ and $\hat{\nu}_0$ under sparsity]
Suppose that $n \to \infty$ and $\pi_1 \to 0$. Then, under~\eqref{eq:normal_samples} and~\eqref{eq:lonnsted_prior}, we have that $\hat{\nu}_0 \to \nu_0$ and $\hat{s}_0^2 \to s_0^2$ almost surely, where $\hat{\nu}_0$ and $\hat{s}_0^2$ are defined in~\eqref{eq:nu_0_quantile} and~\eqref{eq:s_0_median}.
\label{lemm:nu_0_s_0_consistency}
\end{lemm}

\begin{proof}
Let us write $F_{K, \nu_0, s_0}(t)$ for the null CDF of $\hat{\tau}_i^2$. We also let $\widehat{F}(t)$ denote the empirical CDF of $\hat{\tau}_1^2,\ldots,\hat{\tau}_n^2$. Notice that for any $t>0$,
$$
\abs{\EEInline{\widehat{F}(t)} - F_{K,\nu_0, s_0^2}(t)} \leq 2 \pi_1.
$$
Combining the above with the Glivenko-Cantelli theorem we thus have that,
$$
\sup_{t > 0} \abs{\widehat{F}(t) - F_{K,\nu_0, s_0}(t)} \to 0 \text{ almost surely.}
$$
From the above and the fact that $F_{K, \nu_0, s_0}(\cdot)$ is continuous and strictly increasing, we conclude that $\widehat{Q}_n(p) \to \smash{Q_{K, \nu_0, s_0}(p)}$ almost surely for any $p \in (0,1)$, where $Q_{K, \nu_0, s_0}(p)$ is the $p$-quantile of $F_{K, \nu_0, s_0}(\cdot)$. In particular, we find that
$$
\frac{ \widehat{Q}_n(3/4)}{\widehat{Q}_n(1/4)} \to R(\nu_0),
$$
where $R(\cdot)$ is defined as in Lemma~\ref{lemm:quantile_ratio_monotone} for $p_1=1/4$ and $p_2=3/4$. Since $R(\cdot)$ is strictly decreasing and continuous, we conclude that $\hat{\nu}_0 \to \nu_0$ almost surely. By continuity of $Q_{K, \nu_0, s_0}(1/2)$ in $\nu_0$, we also conclude that $\hat{s}_0^2 \to s_0^2$ almost surely.
\end{proof}

\subsection{Proof of Proposition~\ref{prop:asymptotic_threshold}}

\begin{proof}
Existence and uniqueness of $\tau_\alpha$ follow from Lemma~\ref{lemm:monotonicity_infty}, since $c_\star>(1/\alpha)-1$.
Existence and uniqueness of $s_{\alpha,n}$ for all large enough $n$ follows from a similar argument, using the monotonicity of $\mathrm{Fdr}_n(\cdot)$, also via Lemma~\ref{lemm:monotonicity_finite_sample}.

Now let us turn to the statement about the limiting behaviour of $s_{\alpha,n}$.
Suppose first that:
$$
\liminf_{n \to \infty} s_{\alpha,n}/\zeta_n(\tau_\alpha) < 1.
$$
Then there exists $\varepsilon>0$ and a subsequence $n_k$ such that $s_{\alpha,n_k} < (1-\varepsilon)\zeta_{n_k}(\tau_{\alpha})$ for all $k$. Also notice that 
$(1-\varepsilon)\zeta_{n_k}(\tau_{\alpha}) = \zeta_{n_k}(\tau_{\alpha}')$ for $\tau_\alpha' :=\tau_\alpha/ (1-\varepsilon)^{\nu+\nu_0}$.
Again by Lemma~\ref{lemm:monotonicity_finite_sample}, we have that $\mathrm{Fdr}_n(s_{\alpha,n_k}) \geq \mathrm{Fdr}_n(\zeta_{n_k}(\tau_{\alpha}'))$ and
$$
 \lim_{k \to \infty }\mathrm{Fdr}_{n_k}(\zeta_{n_k}(\tau_{\alpha}')) = \mathrm{Fdr}_{\infty}(\tau_\alpha') > \alpha.
$$
This is a contradiction, since $\mathrm{Fdr}_n(s_{\alpha,n_k}) = \alpha$. We conclude that $\liminf_{n \to \infty} s_{\alpha,n}/\zeta_n(\tau_\alpha) \geq 1$.
We can argue analogously to show that $\limsup_{n \to \infty} s_{\alpha,n}/\zeta_n(\tau_\alpha) \leq 1$. Thus we conclude that $s_{\alpha,n}/\zeta_n(\tau_\alpha) \to 1$ as $n \to \infty$.

\end{proof}

\subsection{Proof of Theorem~\ref{theo:sparse_regime}}

We split the proof into several parts, corresponding to the different claims in the theorem.

\begin{proof}[Proof for oracle local false discovery rate procedure]
Throughout this section, we write
$$
S^{\star}(\boldz) := \abs{T(\boldz; \nu_0, s_0)}.
$$
where $T(\boldz; \nu_0, s_0)$ is the moderated $t$-statistic given by \eqref{eq:moderated_stats}.
As our first result we note the following:
\begin{equation}
\EE{\p{\frac{\sum_{i=1}^n \ind\{\mu_i \neq 0\}}{n\pi_1} - 1}^2}  \to 0\, \text{ as }\, n \to \infty.
\label{eq:pi1_convergence}
\end{equation}
To see this, note that
$$\EE{\p{\frac{\sum_{i=1}^n \ind\{\mu_i \neq 0\}}{n\pi_1} - 1}} =0,\;\;\, \Var{\frac{\sum_{i=1}^n \ind\{\mu_i \neq 0\}}{n\pi_1}} = \frac{1-\pi_1}{n\pi_1} \to 0 \text{ as } n \to \infty,
$$
since $n \pi_1 \to \infty$ by our asymptotic regime. Next, let 
$$
D_n :=\frac{1}{n}\sum_{i=1}^n
\ind\cb{S^{\star}(Z_i)\geq s_{\alpha,n}}.
$$
Notice that 
$$ 
\EE{D_n} = 2\pi_0 \bar F_0(s_{\alpha,n}) + 2\pi_1 \bar F_1(s_{\alpha,n}) = 2\pi_1\{\tau_{\alpha} + \Psi_1(\tau_{\alpha})\}\{1+o(1)\}\, \text{ as }\, n \to \infty.
$$
Applying Lemma~\ref{lemm:gc_multiplicative_compound} to $t\mapsto\ind\{S^*(\boldZ_i)\geq t\}$ at $t=s_{\alpha,n}$, with $\delta = c \pi_1$ for sufficiently small $c>0$ and $\varepsilon = (n \pi_1)^{-1/4}$ (so that $n \varepsilon^2\delta, n \varepsilon^3 \delta \to \infty$), we find that:
$$
\frac{D_n}{\EE{D_n}} \stackrel{\mathbb P}{\to} 1 \text{ as } n \to \infty,\; \text{ and so},\,\;\frac{D_n}{2\pi_1\{\tau_{\alpha} + \Psi_1(\tau_{\alpha})\}} \stackrel{\mathbb P}{\to} 1\, \text{ as }\, n \to \infty.
$$
We find that also as $n\to\infty$,
$$
\frac{1}{2\pi_1 \tau_{\alpha}  }\cdot \frac{1}{n}\sum_{i=1}^n \ind\cb{S^{\star}(Z_i)\geq s_{\alpha,n},\, \mu_i = 0}\stackrel{\mathbb P}{\to} 1 ,\;\; \frac{1}{2 \pi_1\Psi_1(\tau_{\alpha})}\cdot \frac{1}{n}\sum_{i=1}^n \ind\cb{S^{\star}(Z_i)\geq s_{\alpha,n},\, \mu_i \neq 0}\stackrel{\mathbb P}{\to} 1.
$$
From these two displays we find that the false discovery proportion satisfies:
$$
\mathrm{FDP}_n = \frac{\sum_{i=1}^n \ind\cb{S^{\star}(Z_i)\geq s_{\alpha,n},\, \mu_i = 0}}{1\lor (D_n \cdot n)} \stackrel{\mathbb P}{\to} \frac{\tau_{\alpha}}{\tau_{\alpha} + \Psi_1(\tau_{\alpha})} = \alpha\; \text{ as }\; n \to \infty.
$$
Hence by moving to subsequences and applying the dominated convergence theorem ($\mathrm{FDP}_n \in [0,1]$), we conclude that $\mathrm{FDR}_n \to \alpha$ as $n \to \infty$.

Analogously, for the power, we have that
$$
\frac{\sum_{i=1}^n \ind\cb{S^{\star}(Z_i)\geq s_{\alpha,n},\, \mu_i \neq 0}}{1 \lor \sum_{i=1}^n \ind\cb{\mu_i \neq 0}} \stackrel{\mathbb P}{\to} 2\Psi_1(\tau_{\alpha})\; \text{ as }\; n \to \infty.
$$
Using again the dominated convergence theorem, we find that $\mathrm{Pow}_n \to 2\Psi_1(\tau_{\alpha})$ as $n \to \infty$.
\end{proof}

\begin{proof}[Proof for BH with compound orthogonal rotation p-values]
We first show that $\zeta_n^{-1}(\hat{s}) \stackrel{\mathbb P}{\to} \tau_\alpha$, where $\hat{s}$ is the BH threshold on the learned score $\widehat{S}(\cdot)$. To this end, define
\begin{align}
&N_{0,n}(s) := \frac{1}{n}\sum_{i=1}^n \int_{O(K)} \!\ind\!\cb{\widehat{S}(H\boldZ_i)\geq s,\,\mu_i=0}\dd\mu_{\mathcal{H}}(H),\, &&D_{0,n}(s) := \frac{1}{n}\sum_{i=1}^n \ind\!\cb{\widehat{S}(\boldZ_i)\geq s,\, \mu_i=0},\notag\\
&N_{1,n}(s) := \frac{1}{n}\sum_{i=1}^n \int_{O(K)} \!\ind\!\cb{\widehat{S}(H\boldZ_i)\geq s,\,\mu_i\neq 0}\dd\mu_{\mathcal{H}}(H),\, &&D_{1,n}(s) := \frac{1}{n}\sum_{i=1}^n \ind\!\cb{\widehat{S}(\boldZ_i)\geq s,\, \mu_i\neq0}\notag\\
&N_n(s) := N_{0,n}(s)+N_{1,n}(s),\;\;\;
&&D_n(s) := D_{0,n}(s)+D_{1,n}(s).
\label{eq:compound_counts}
\end{align}
Since the compound p-values are decreasing functions of $\widehat{S}(\boldZ_i)$, BH rejects at the smallest threshold $\hat{s} \in \{\widehat{S}(\boldZ_1),\ldots,\widehat{S}(\boldZ_n)\}$ with $N_n(\hat{s})/D_n(\hat{s})\leq\alpha$. We set $\hat{s}=\infty$ if no such threshold exists.

Let $N_{0,n}^*(s)$, $N^*_{1,n}(s)$, $N_n^*(s)$ and $D_n^*(s)$ be defined as above with $\widehat{S}$ replaced by $S^*(\cdot) = \abs{T(\cdot; \nu_0, s_0)}$ (that is, with true rather than estimated values of $\nu_0$ and $s_0^2$). Fix $\tau>0$ and write $b_n:=\zeta_n(\tau)$. By null rotational invariance,
$$
\EE{N_{0,n}^*(s)}=2\pi_0\bar{F}_0(s),\;\;\;\;
\EE{D_n^*(s)}=2\pi_0\bar{F}_0(s)+2\pi_1\bar{F}_1(s).
$$
On $[0,2b_n]$, both expectations are bounded below by $2\pi_0\bar{F}_0(2b_n)$, since $\bar{F}_0$ is decreasing. Moreover, $\liminf_{n \to \infty}\bar{F}_0(2b_n)/\pi_1>0$ (to see this, recall that $\bar{F}_0(b_n)/\pi_1 \to \tau$) and $\pi_0\to1$, so this lower bound is at least a positive constant times $\pi_1$ for all large enough $n$. Applying Lemma~\ref{lemm:gc_multiplicative_compound} with $\delta=c\pi_1$ for sufficiently small $c>0$ and $\varepsilon=(n\pi_1)^{-1/4}$, we thus find that as $n\to \infty$,
$$
\sup_{0\leq s\leq 2b_n}\abs{\frac{N_{0,n}^*(s)}{2\pi_0\bar{F}_0(s)}-1}\stackrel{\mathbb P}{\to}0,\;\;\;\;\sup_{0\leq s\leq 2b_n}\abs{\frac{D_n^*(s)}{2\pi_0\bar{F}_0(s)+2\pi_1\bar{F}_1(s)}-1}
\stackrel{\mathbb P}{\to}0.
$$
Here we apply the lemma to the independent functions averaged in $N_{0,n}^*(t)$ and $D_n^*(t)$.

Next we show that the contribution of $N_{1,n}^*$ to $N_n^*$ is negligible. Notice that
\begin{equation}
S^*(\boldz)
=\frac{\sqrt{K(\nu+\nu_0)}\abs{\hat{\mu}(\boldz)}}{\sqrt{\nu_0s_0^2+\nu\hat{\sigma}^2(\boldz)}} \leq
\sqrt{\frac{\nu+\nu_0}{\nu}}\abs{T(\boldz;0,1)}.
\label{eq:rotation_bound}
\end{equation}
As in the proof of Proposition~\ref{prop:sep_pvalue_limma}, $T(H\boldZ_i;0,1)$ has a $t_\nu$ distribution conditionally on $\boldZ_i$. Consequently,
$$
N_{1,n}^*(s)\leq\frac{1}{n}\sum_{i=1}^n\ind\{\mu_i\neq0\}\cdot2\bar{F}_{t,\nu}\p{s\sqrt{\nu/(\nu+\nu_0)}}.
$$
By the tail behavior of the t-distribution, there exists a constant $C>0$ such that
$$
\sup_{u \in [0,s]}\cb{\frac{\bar{F}_{t,\nu}(u\sqrt{\nu/(\nu+\nu_0)})}{\bar{F}_0(u)}}\leq C(1+s)^{\nu_0} \; \text{ for all } s\geq 0.$$
Thus, using~\eqref{eq:pi1_convergence} and recalling that $b_n = \zeta_n(\tau)=\Theta(\pi_1^{-1/(\nu+\nu_0)})$ ,
$$
\sup_{0\leq s\leq 2b_n}\frac{N_{1,n}^*(s)}{\bar{F}_0(s)}
=O_{\mathbb P}\p{\pi_1(1+b_n)^{\nu_0}}
=O_{\mathbb P}\p{\pi_1^{\nu/(\nu+\nu_0)}}\stackrel{\mathbb P}{\to}0.
$$
We now turn our attention to $N_n$ and $D_n$. 
By Lemmata~\ref{lemm:moderated_multiplicative} and~\ref{lemm:nu_0_s_0_consistency}, there exists $\varepsilon_n\stackrel{\mathbb P}{\to}0$ such that, with probability tending to one,
$$
(1-\varepsilon_n)S^*(\boldz)\leq\widehat{S}(\boldz)\leq(1+\varepsilon_n)S^*(\boldz)\;\text{ for all }\;\boldz.
$$
Hence $N_n(s)$ lies between $N_n^*(s/(1-\varepsilon_n))$ and $N_n^*(s/(1+\varepsilon_n))$, and the same bounds hold for $D_n(s)$ with $D_n^*$ in place of $N_n^*$. For $j\in\{0,1\}$, Lemma~\ref{lemm:t_tail_rescaling} gives
$$
\sup_{s \geq 0}\abs{\frac{\bar{F}_j(as)}{\bar{F}_j(s)}-1}
\to 0\; \text{ as }\; a \to 1,
$$
uniformly in $\lambda$. (Formally, the lemma established it for $j=0$, but it also works for $j=1$ since $\bar F_1$ is just rescaling the argument of $\bar{F}_0$.)
Using the earlier ratio concentration, we find that
\begin{equation}
\sup_{0\leq s\leq\zeta_n(\tau)}\abs{\frac{N_n(s)}{n^{-1}\lor D_n(s)} - \mathrm{Fdr}_n(s)} \stackrel{\mathbb P}{\to}0
\;\text{ for every fixed }\;\tau>0.
\label{eq:orthogonal_fdr_uniform}
\end{equation}
Take $0<\tau_-<\tau_\alpha<\tau_+$. Using Lemma~\ref{lemm:sparse_limits} and the monotonicity of $\mathrm{Fdr}_{\infty}$ in Lemma~\ref{lemm:monotonicity_infty},
$$
\mathrm{Fdr}_n(\zeta_n(\tau_-))\to\mathrm{Fdr}_\infty(\tau_-)<\alpha,
\;\;\;\;
\mathrm{Fdr}_n(\zeta_n(\tau_+))\to\mathrm{Fdr}_\infty(\tau_+)>\alpha.
$$
Since $\mathrm{Fdr}_n(s)$ is decreasing in $s$, equation~\eqref{eq:orthogonal_fdr_uniform} implies that, with probability tending to one, no $s\leq\zeta_n(\tau_+)$ satisfies $N_n(s)/D_n(s)\leq\alpha$. This means that on the same event, $\hat{s}>\zeta_n(\tau_+)$.
Meanwhile, 
$N_n(\zeta_n(\tau_-))/D_n(\zeta_n(\tau_-))<\alpha$ with probability tending to one. Thus we would like to conclude that $\hat{s}\leq\zeta_n(\tau_-)$, but this is not immediate because $\zeta_n(\tau_-)$ need not be an element of the set of observed scores $\{\widehat{S}(\boldZ_1),\ldots,\widehat{S}(\boldZ_n)\}$.\footnote{
A reader may notice that this argument is more cumbersome than usual arguments for asymptotic FDR control with BH, e.g., in~\citet{storey2004strong}. Indeed, for compound BH with $\mathcal{H}_{\mathrm{orth}}$ we could redefine Procedure~\ref{proc:compound_threshold} to instead use
\begin{equation}
\hat{s}=\inf\left\{s\geq0:\frac{N_n(s)}{n^{-1}\lor D_n(s)}\leq\alpha\right\},
\label{eq:compound_threshold_all_scores}
\end{equation}
in which case we could avoid the next step. The reason that we present this more involved argument is that we also reuse it later for compound BH with $\mathcal{H}_{\mathrm{symm}}$. In that case, $N_n(\cdot)$ is no longer continuous (instead it is left-
continuous with downward jumps), and redefining $\hat{s}$ as in~\eqref{eq:compound_threshold_all_scores} and rejecting all hypotheses with $\widehat{S}(\boldZ_i) \geq \hat{s}$ is no longer equivalent to Procedures~\ref{proc:compound_pvalues}/~\ref{proc:compound_threshold}.
}
Instead we argue as follows. Fix $\tau\in(\tau_-,\tau_\alpha)$. Since $\mathrm{Fdr}_\infty(\tau)<\alpha$, equation~\eqref{eq:orthogonal_fdr_uniform} also 
yields $N_n(\zeta_n(\tau))/D_n(\zeta_n(\tau))\leq\alpha$ with probability tending to one. Moreover, 
$$
\frac{D_n(\zeta_n(\tau))}{\pi_1}\stackrel{\mathbb P}{\to}2\{\tau+\Psi_1(\tau)\},\;\;\;\;
\frac{D_n(\zeta_n(\tau_-))}{\pi_1}\stackrel{\mathbb P}{\to}2\{\tau_-+\Psi_1(\tau_-)\}.
$$
The first limit is strictly larger than the second, so $D_n(\zeta_n(\tau))>D_n(\zeta_n(\tau_-))$ with probability tending to one. Their difference is the number of observed scores in $[\zeta_n(\tau),\zeta_n(\tau_-))$ divided by $n$, so this interval contains an observed score. Moving from $\zeta_n(\tau)$ to the first such score leaves $D_n$ unchanged and can only decrease $N_n$. Thus,
$$
\PP{\zeta_n(\tau_+)\leq\hat{s}\leq\zeta_n(\tau_-)} \to 1.
$$
Since $\zeta_n$ is strictly decreasing, we therefore have
$$
\PP{\tau_-\leq\zeta_n^{-1}(\hat{s})\leq\tau_+}  \to 1
$$
Since we chose $\tau_-,\tau_+$ as arbitrary numbers with $0<\tau_-<\tau_\alpha<\tau_+$, this proves that 
\begin{equation}
\zeta_n^{-1}(\hat{s})\stackrel{\mathbb P}{\to}\tau_\alpha.
\label{eq:tau_alpha_convergence_BH_orthog}
\end{equation}
Now for every fixed $\tau>0$, 
$$
\frac{D_{0,n}(\zeta_n(\tau))}{\pi_1}\stackrel{\mathbb P}{\to}2\tau,\;\;\;\;
\frac{D_{1,n}(\zeta_n(\tau))}{\pi_1}\stackrel{\mathbb P}{\to}2\Psi_1(\tau).
$$
By monotonicity of $D_{0,n}$ and $D_{1,n}$, and using continuity of $\Psi_1$, we conclude that
$$
\frac{D_{0,n}(\hat{s})}{\pi_1}\stackrel{\mathbb P}{\to}2\tau_\alpha,\;\;\;\;
\frac{D_{1,n}(\hat{s})}{\pi_1}\stackrel{\mathbb P}{\to}2\Psi_1(\tau_\alpha).
$$
Together with~\eqref{eq:pi1_convergence}, these limits yield
$$
\frac{nD_{0,n}(\hat{s})}{1\lor (nD_n)(\hat{s})}\stackrel{\mathbb P}{\to}\frac{\tau_\alpha}{\tau_\alpha+\Psi_1(\tau_\alpha)}=\alpha,
\;\;\;\;
\frac{nD_{1,n}(\hat{s})}{1\lor\sum_{i=1}^n\ind\{\mu_i\neq0\}}\stackrel{\mathbb P}{\to}2\Psi_1(\tau_\alpha).
$$
By dominated convergence, we conclude that $\mathrm{FDR}_n\to\alpha$ and $\mathrm{Pow}_n\to2\Psi_1(\tau_\alpha)$ as $n\to\infty$.
\end{proof}

\begin{proof}[Proof for DDR with compound orthogonal rotation p-values]

This proof is similar to the proof for BH with compound orthogonal rotation p-values and will reuse several of the arguments used therein.  Let us introduce some more notation:
$$
\begin{aligned}
N^{\mathrm{DDR}}_{n}(s) := \frac{1}{n}\sum_{i=1}^n \frac{\int_{O(K)} \!\ind\!\cb{\widehat{S}(H\boldZ_i)\geq s}\dd\mu_{\mathcal{H}}(H)}{1-\int_{O(K)} \!\ind\!\cb{\widehat{S}(H\boldZ_i)\geq s_{\tau_n}}\dd\mu_{\mathcal{H}}(H) }.
\end{aligned}
$$
Meanwhile, $N_{n}$, $D_{0,n}, D_{1,n}$, and $D_n$ are defined as in~\eqref{eq:compound_counts}.
By continuity of the rotation survival functions and the definition of $s_{\tau_n}$ in Procedure~\ref{proc:compound_ddr}, we have $N_n^{\mathrm{DDR}}(s_{\tau_n})=\tau_n$.

One can check that DDR will reject at the smallest threshold 
$$ \hat{s} \in \{\widehat{S}(\boldZ_1),\ldots,\widehat{S}(\boldZ_n)\} \cap [s_{\tau_n}, \infty),$$ such that $N^{\mathrm{DDR}}_n(\hat{s})/D_n(\hat{s})\leq\alpha$ (allowing $\hat{s}=\infty$ if no such threshold exists).

The rest of the proof works as follows. We show that the conclusion in~\eqref{eq:tau_alpha_convergence_BH_orthog} for BH that 
\smash{$\zeta_n^{-1}(\hat{s})\stackrel{\mathbb P}{\to}\tau_\alpha$} also holds for DDR. In turn this argument requires two main steps: showing that $N_n^{\mathrm{DDR}}$ can be uniformly controlled by $N_n$, and second that $\PP{\zeta_n(\tau)>s_{\tau_n}}\to1$ for every fixed $\tau>0$ as $n\to\infty$.

First we show that \smash{$s_{\tau_n}\stackrel{\mathbb P}{\to}\infty$}. Indeed, for every fixed $M>0$, the previous proof gives \smash{$N_n(M)\stackrel{\mathbb P}{\to}2\bar{F}_0(M)>0$}. Since $N_n\leq N_n^{\mathrm{DDR}}$ and $\tau_n \to 0$, we have
$$
\PP{s_{\tau_n}\leq M}\leq\PP{N_n(M)\leq\tau_n} \to 0.
$$
Applying the argument in~\eqref{eq:rotation_bound}
to $\widehat{S}(\cdot)$ instead of $S^*(\cdot)$, we find that
$$
\varepsilon_n:=\max_{1\leq i\leq n}\int_{O(K)}\!\ind\!\cb{\widehat{S}(H\boldZ_i)\geq s_{\tau_n}}\dd\mu_{\mathcal{H}}(H)
\leq2\bar{F}_{t,\nu}\p{s_{\tau_n}\sqrt{\frac{\nu}{\nu+\hat\nu_0}}}\stackrel{\mathbb P}{\to}0,
$$
where the convergence follows from $s_{\tau_n}\stackrel{\mathbb P}{\to}\infty$ and $\hat\nu_0\stackrel{\mathbb P}{\to}\nu_0$. Meanwhile,
$$
N_n(s)\leq N_n^{\mathrm{DDR}}(s)\leq\frac{N_n(s)}{1-\varepsilon_n}\;\text{ for all }\;s\geq0.
$$
Combining this bound with~\eqref{eq:orthogonal_fdr_uniform}, we find that
$$
\sup_{0\leq s\leq\zeta_n(\tau)}\abs{\frac{N_n^{\mathrm{DDR}}(s)}{n^{-1}\lor D_n(s)}-\mathrm{Fdr}_n(s)}\stackrel{\mathbb P}{\to}0
\;\text{ for every fixed }\;\tau>0.
$$
Next, fix $\tau>0$. The previous proof gives $N_n(\zeta_n(\tau))/\pi_1\stackrel{\mathbb P}{\to}2\tau$, and hence
$$
\frac{N_n^{\mathrm{DDR}}(\zeta_n(\tau))}{\tau_n}
\leq\frac{\pi_1}{\tau_n}\cdot\frac{N_n(\zeta_n(\tau))}{\pi_1}\cdot\frac{1}{1-\varepsilon_n}
\stackrel{\mathbb P}{\to} 0,
$$
where the last step used the assumption of the theorem that $\tau_n/\pi_1\to\infty$.
Since $N_n^{\mathrm{DDR}}$ is decreasing and $N_n^{\mathrm{DDR}}(s_{\tau_n})=\tau_n$, this proves that $\PP{\zeta_n(\tau)>s_{\tau_n}} \to 1$.

From here on, we can argue almost verbatim as in the proof for BH with compound orthogonal rotation p-values. We omit the details for brevity.

\end{proof}

\begin{proof}[Proof for BH with compound sign flipping p-values]
We first show that $\zeta_n^{-1}(\hat{s})\stackrel{\mathbb P}{\to}\tau_{\alpha_\nu}$. We use the notation of the preceding proof, with the numerator now given by
$$
N_n(s):=\frac{1}{n2^K}\sum_{i=1}^n\sum_{H\in\mathcal{H}_{\mathrm{symm}}}\ind\cb{\widehat{S}(H\boldZ_i)\geq s},
$$
and $N_{0,n}(s)$ and $N_{1,n}(s)$ defined analogously. The denominator terms ($D_n(s)$, $D_{0,n}(s)$, and $D_{1,n}(s)$) are defined as in~\eqref{eq:compound_counts}. We also use the notation $N_n^*(s)$, $N_{0,n}^*(s)$, and $N_{1,n}^*(s)$ and $D^*(s)$ for the same quantities with $\widehat{S}$ replaced by $S^*$.
We let $\hat{s}$ be the smallest observed score satisfying $N_n(\hat{s})/D_n(\hat{s})\leq\alpha$, setting $\hat{s}=\infty$ if there are no rejections. 

Here we also have that $\EE{N_{0,n}^*}=2\pi_0\bar{F}_0(s)$. For $N_{1,n}^*$, we observe that the sign flips $\mathbf{1}$ and $-\mathbf{1}$ leave $S^*(\boldZ_i)$ unchanged. Hence
\begin{equation}
N_{1,n}^*(s)=2^{-\nu}D_{1,n}^*(s)+ R_n^*(s),\; \text{ where }\, R_n^*(s):=\frac{1}{n2^K}\sum_{i:\mu_i\neq0}\sum_{H\notin\{\mathbf{1},-\mathbf{1}\}}\ind\cb{S^*(H\boldZ_i)\geq s},
\label{eq:R_n_star}
\end{equation}
with the inner sum being taken over $\mathcal{H}_{\mathrm{symm}}\setminus\{\mathbf{1},-\mathbf{1}\}$. We also recall the notation $D_{1,n}^*(s):=n^{-1}\sum_{i:\mu_i\neq0}\ind\{S^*(\boldZ_i)\geq s\}$. We will show that, for any fixed $\tau>0$ and $b_n:=\zeta_n(\tau)$,
$$
\sup_{0\leq s\leq2b_n}\frac{R_n^*(s)}{\bar{F}_0(s)}\stackrel{\mathbb P}{\to}0.
$$
Writing $H_j\in\{-1,1\}$ for the sign applied to coordinate $j$, define
$$
\delta_K(H):=\frac{\#\cb{j:H_j=1}-\#\cb{j:H_j=-1}}{K}.
$$
In what follows, we recall that under~\eqref{eq:normal_samples}--\eqref{eq:lonnsted_prior}, $\mu_i/\sigma_i\mid\mu_i\neq0\sim\mathrm{N}(0,\lambda/K)$ and moreover that under our sparse limit, $\lambda \to \infty$ with the distribution of $\sigma_i$ held fixed. These limits imply the following: for any fixed $C>0$ and $\varepsilon>0$,
$$
\PP{|\mu_i|\leq C\mid\mu_i\neq0} \to 0,\;\,
\PP{\Norm{\boldZ_i/\mu_i-\mathbf{1}}_2>\varepsilon\mid\mu_i\neq0} \to 0.
$$
Now for any sign flip $H \in \mathcal{H}_{\mathrm{symm}}$, the above also imply that for any $\varepsilon>0$,
$$
\PP{ \abs{\frac{\hat\mu(H\boldZ_i)}{\mu_i} -\delta_K(H)} > \varepsilon \mid \mu_i \neq 0} \to 0.
$$
Now recall that
$$
S^{*2}(H\boldZ_i)
=\frac{(\nu+\nu_0)K\{\hat{\mu}(H\boldZ_i)/\mu_i\}^2}
{\nu_0s_0^2/\mu_i^2+\Norm{\boldZ_i/\mu_i}_2^2-K\{\hat{\mu}(H\boldZ_i)/\mu_i\}^2}.
$$
Now moreover require that $H\notin\{\mathbf{1},-\mathbf{1}\}$, so that $|\delta_K(H)|<1$. By using the above limits, we find that for any $\varepsilon>0$:
\begin{equation}
\PP{ \abs{S^{*2}(H\boldZ_i) - \frac{(\nu+\nu_0)\delta_K(H)^2}{1-\delta_K(H)^2}} > \varepsilon \mid \mu_i \neq 0} \to 0\, \text{ as } n \to \infty.
\label{eq:s_star_alt_limit}
\end{equation}
Choose 
$$ M:= 1 + \max\!\cb{\sqrt{\frac{(\nu+\nu_0)\delta_K(H)^2}{1-\delta_K(H)^2}}\;:\;H\in\mathcal{H}_{\mathrm{symm}}\setminus\{\mathbf{1},-\mathbf{1}\}  }.
$$
Then $\PP{S^*(H\boldZ_i)\geq M \mid \mu_i \neq 0}\to0$ for each $H \notin \{\mathbf{1}, -\mathbf{1}\}$. Using~\eqref{eq:R_n_star}, we find that,
$$
\frac{\EE{R_n^*(M)}}{\pi_1}
=\frac{1}{2^K}\sum_{H\in\mathcal{H}_{\mathrm{symm}}\setminus\{\mathbf{1},-\mathbf{1}\}}\PP{S^*(H\boldZ_1)\geq M\mid\mu_1\neq0} \to 0,
$$
where we also use the 
fact that $K$ is fixed in our asymptotics. Also, $R_n^*(0)\leq n^{-1}\sum_i\ind\{\mu_i\neq0\}=O_{\mathbb P}(\pi_1)$. By monotonicity, for $n$ large enough such that $2b_n>M$,
$$
\sup_{0\leq s\leq2b_n}\frac{R_n^*(s)}{\bar{F}_0(s)} \leq \sup_{0\leq s\leq M}\frac{R_n^*(s)}{\bar{F}_0(s)} + \sup_{M\leq s\leq 2b_n}\frac{R_n^*(s)}{\bar{F}_0(s)} \leq 
\frac{R_n^*(0)}{\bar{F}_0(M)}+\frac{R_n^*(M)}{\bar{F}_0(2b_n)}
\stackrel{\mathbb P}{\to}0,
$$
since $\liminf_{n \to \infty} \bar{F}_0(2b_n)/\pi_1 >0$. 
Applying Lemma~\ref{lemm:gc_multiplicative_compound} to the independent functions averaged in $N_{0,n}^*(t)$, $D_{1,n}^*(t)$, and $D_n^*(t)$, and using the bound on $R_n^*$, we find that
$$
\sup_{0\leq s\leq2b_n}\abs{\frac{N_n^*(s)}{2\pi_0\bar{F}_0(s)+2^{1-\nu}\pi_1\bar{F}_1(s)}-1}\stackrel{\mathbb P}{\to}0,\;\; \sup_{0\leq s\leq2b_n}\abs{\frac{D_n^*(s)}{2\pi_0\bar{F}_0(s)+2\pi_1\bar{F}_1(s)}-1}\stackrel{\mathbb P}{\to}0.
$$
We now turn to $N_n$ and $D_n$. Arguing as in the proof for compound BH with $\mathcal{H}_{\mathrm{orth}}$, we find that for any $\tau>0$,
$$
\sup_{0\leq s\leq\zeta_n(\tau)}\abs{\frac{N_n(s)}{n^{-1}\lor D_n(s)}-\cb{2^{-\nu}+(1-2^{-\nu})\mathrm{Fdr}_n(s)}}\stackrel{\mathbb P}{\to}0.
$$
Here we used the identity
$$
\frac{\pi_0\bar{F}_0(s)+2^{-\nu}\pi_1\bar{F}_1(s)}{\pi_0\bar{F}_0(s)+\pi_1\bar{F}_1(s)}
=2^{-\nu}+(1-2^{-\nu})\mathrm{Fdr}_n(s).
$$
The limit of the right-hand side as $n \to \infty$ equals $\alpha$ precisely when $\mathrm{Fdr}_\infty(\tau)=\alpha_\nu$.

We omit the remainder of the proof, which is analogous to the proof for compound BH with $\mathcal{H}_{\mathrm{orth}}$.
\end{proof}

\begin{proof}[Proof for SeqStep+]
Write $H=H_{\mathrm{half}}$. Let us temporarily assume that $K$ is even. Then notice the following.\footnote{This argument and its variant for odd $K$ was suggested to us by GPT-6 Astra
} For any $\boldz \in \RR^K$:
$$
\hat{\mu}^2(H\boldz) = \p{\frac{1}{K}\sum_{k=1}^{K}H_k\{z_k - \hat{\mu}(\boldz)\}}^2 \leq  \frac{1}{K}\sum_{k=1}^{K}H_k^2\{z_k - \hat{\mu}(\boldz)\}^2 = \frac{\nu}{K} \hat{\sigma}^2(\boldz),
$$
where we used Jensen's inequality and the facts that $\sum_{k} H_k =0$ and $H_k^2=1$. 
Then applying the above with $H\boldz$ in place of $\boldz$,
$$\hat{\mu}^2(\boldz)  = \hat{\mu}^2(H\cdot H\boldz) \leq \frac{\nu}{K} \hat{\sigma}^2(H\boldz).$$
The implication is the following:
$$
\begin{aligned}
\widehat{S}(\boldz) \cdot \widehat{S}(H\boldz) &= K(\nu+\hat{\nu}_0) \frac{|\hat{\mu}(\boldz) \hat{\mu}(H\boldz)|}{\{\hat{\nu}_0 \hat{s}_0^2 + \nu \hat{\sigma}^2(\boldz)\}^{1/2}\{\hat{\nu}_0 \hat{s}_0^2 + \nu \hat{\sigma}^2(H\boldz)\}^{1/2}} \\ 
&\leq (\nu+\hat{\nu}_0) \frac{\nu^{1/2}\hat{\sigma}(\boldz) \nu^{1/2}\hat{\sigma}(H\boldz)}{\{\hat{\nu}_0 \hat{s}_0^2 + \nu \hat{\sigma}^2(\boldz)\}^{1/2}\{\hat{\nu}_0 \hat{s}_0^2 + \nu \hat{\sigma}^2(H\boldz)\}^{1/2}} \,\leq\, \nu + \hat{\nu}_0.
\end{aligned}
$$
In particular,
$$
\min\cb{\widehat{S}(\boldz),\widehat{S}(H\boldz)}
\leq\cb{\widehat{S}(\boldz)\cdot\widehat{S}(H\boldz)}^{1/2}
\leq\sqrt{\nu+\hat\nu_0}.
$$
Since $\hat{\nu}_0 \stackrel{\mathbb P}{\to} \nu_0$, we find that with $M:= 2\{1+(\nu + \nu_0)^{1/2}\} > 1$, then:
\begin{equation}
\PP{A_n}  \to 1 \text{ as } n \to \infty,\, \text{ where }
A_n  := \cb{\min\{\widehat{S}(\boldz),\widehat{S}(H\boldz)\}<M \text{ for all } \boldz \in \RR^K}
\label{eq:seqstep_event}
\end{equation}
At the end of the proof we will show that in fact~\eqref{eq:seqstep_event} holds also for odd $K$ (not just for even $K$).

So we run the rest of the proof assuming that we are on the event $A_n$. The upshot is as follows. Fix any $s \geq M$. Now suppose that $\widehat{S}(\boldZ_i) \geq s$, then it must be the case that $\widehat{S}(H\boldZ_i) < M \leq s \leq \widehat{S}(\boldZ_i)$ and therefore $\tilde{S}_i = \widehat{S}(\boldZ_i)$. Similarly, if $\widehat{S}(H\boldZ_i) \geq s$, then $\tilde{S}_i = -\widehat{S}(H\boldZ_i)$.

Now recall $\widehat{\mathrm{FDR}}(s)=
(1+\sum_{i=1}^n\ind\{\tilde{S}_i\leq-s\})/
(1\lor\sum_{i=1}^n\ind\{\tilde{S}_i\geq s\})$ from~\eqref{eq:fdr_seqstep}. Then, on $A_n$ it holds that, 
\begin{equation}
\widehat{\mathrm{FDR}}(s)=  \widehat{\mathrm{FDR}^*}(s)\, \text{ for all }\,s \geq M,\; \text{ where }\; \widehat{\mathrm{FDR}^*}(s):=
\frac{1+\sum_{i=1}^n\ind\!\cb{\widehat{S}(H\boldZ_i)\geq s}}
{1\lor\sum_{i=1}^n\ind\!\cb{\widehat{S}(\boldZ_i)\geq s}}.
\label{eq:above_M_simplification}
\end{equation}
In particular, this also means that on $A_n$ and if $\hat{s} \geq M$, then the rejected hypotheses are precisely those with $\widehat{S}(\boldZ_i) \geq \hat{s}$. 

Our strategy is as follows: first we show that  $\hat{s} \geq M$ with probability tending to one, then we can focus on $s \geq M$ where the above simplification holds. 

Let us start by studying $\widehat{\mathrm{FDR}}(s)$ for $s \leq M$. Conditioning on orbits $\mathcal{O}_{1:n}$ and $\boldmu = (\mu_1,\ldots,\mu_n)$, we have that:
$$
\begin{aligned}
&\EE{\sum_{i:\mu_i=0}\ind\{\tilde{S}_i\leq-s\}\,\cond\,\mathcal{O}_{1:n},\,\boldmu}=\EE{\sum_{i:\mu_i=0}\ind\{\tilde{S}_i\geq s\}\,\cond\,\mathcal{O}_{1:n},\,\boldmu}=\frac{1}{2}\sum_{i:\mu_i=0}\ind\{|\tilde{S}_i|\geq s\},
\end{aligned}
$$
where the last step uses that $\tilde{S}_i$ is symmetric about zero under the null. By Bretagnolle's (independent but not iid)   Dvoretzky-Kiefer-Wolfowitz inequality applied conditionally on $\mathcal{O}_{1:n}$ and $\boldmu$, we have for any $\varepsilon>0$ that with $\pi_{0,n} := n^{-1}\sum_{i=1}^n\ind\{\mu_i=0\}$,
$$
\PP{\sup_{s} \abs{\sum_{i:\mu_i=0}\ind\{\tilde{S}_i\geq s\} - \frac{1}{2}\sum_{i:\mu_i=0}\!\!\ind\{|\tilde{S}_i|\geq s\}}  \geq \varepsilon \cdot n \cdot \pi_{0,n} \, \mid \, \mathcal{O}_{1:n},\,\boldmu} \leq 2 e \exp\p{-2\cdot \pi_{0,n}\cdot n\varepsilon^2}.
$$
An analogous statement holds for $\sum_{i:\mu_i=0}\ind\{\tilde{S}_i\leq -s\}$. 
Moreover, for $s\leq M$, we have that:
$$
\frac{1}{2n}\sum_{i:\mu_i=0}\ind\{|\tilde{S}_i|\geq s\}
\geq\frac{1}{2n}\sum_{i:\mu_i=0}\ind\{\widehat{S}(\boldZ_i)\geq M\}
\stackrel{\mathbb P}{\to}\bar{F}_0(M)>0,
$$
where the convergence follows from the proof for compound BH with $\mathcal{H}_{\mathrm{orth}}$. Since also $\pi_{0,n}\stackrel{\mathbb P}{\to}1$, we conclude that
$$
\sup_{0<s\leq M}\abs{
\frac{\sum_{i:\mu_i=0}\ind\{\tilde{S}_i\geq s\}}
{\frac12\sum_{i:\mu_i=0}\ind\{|\tilde{S}_i|\geq s\}}-1}
\stackrel{\mathbb P}{\to}0,\;\;\sup_{0<s\leq M}\abs{
\frac{\sum_{i:\mu_i=0}\ind\{\tilde{S}_i\leq -s\}}
{\frac12\sum_{i:\mu_i=0}\ind\{|\tilde{S}_i|\geq s\}}-1}
\stackrel{\mathbb P}{\to}0.
$$
We arrive at the following result (note that we can ignore the contribution of the $+1$ and also of non-nulls, since their number is $o_{\mathbb P}(n)$ by~\eqref{eq:pi1_convergence}):
$$
\sup_{0<s\leq M}\abs{\widehat{\mathrm{FDR}}(s)-1}\stackrel{\mathbb P}{\to}0,
\;\text{ and so, }\;
\PP{\hat{s}>M} \to 1.
$$
Now let $D_n$ be defined as in~\eqref{eq:compound_counts}, and define
$$
N_n^{\mathrm{Seq}}(s):=\frac{1}{n}+\frac{1}{n}\sum_{i=1}^n\ind\cb{\widehat{S}(H\boldZ_i)\geq s}.
$$
Then $\widehat{\mathrm{FDR}^*}(s)=N_n^{\mathrm{Seq}}(s)/(n^{-1}\lor D_n(s))$.

Next, the proof for compound BH with sign flips in~\eqref{eq:s_star_alt_limit} shows that $\PP{S^*(H\boldZ_i)\geq M/2\mid\mu_i\neq0}\to0$ (because the limit of $S^*(H\boldZ_i)$ in~\eqref{eq:s_star_alt_limit} is $0$ for even $K$ and $\sqrt{(\nu+\nu_0)/(K^2-1)}$ for odd $K$, both below $M/2=1+\{\nu+\nu_0\}^{1/2}$). Hence, for every fixed $\tau>0$,
$$
\sup_{M/2\leq s\leq2\zeta_n(\tau)}
\frac{\sum_{i:\mu_i\neq0}\ind\{S^*(H\boldZ_i)\geq s\}}{n\bar{F}_0(s)}
\leq\frac{\sum_{i:\mu_i\neq0}\ind\{S^*(H\boldZ_i)\geq M/2\}}{n\bar{F}_0(2\zeta_n(\tau))}
\stackrel{\mathbb P}{\to}0.
$$
To see the convergence, the expected numerator is $n\pi_1\PP{S^*(H\boldZ_i)\geq M/2\mid\mu_i\neq0}=o(n\pi_1)$, while $\liminf_{n \to \infty}\bar{F}_0(2\zeta_n(\tau))/\pi_1>0$ by the usual argument, and then we can apply Markov's inequality.

Under the null, $S^*(H\boldZ_i)$ has the same distribution as $S^*(\boldZ_i)$. Thus the same concentration and multiplicative approximation arguments used to prove~\eqref{eq:orthogonal_fdr_uniform} give, for every fixed $\tau>0$,
$$
\sup_{M\leq s\leq\zeta_n(\tau)}\abs{\frac{N_n^{\mathrm{Seq}}(s)}{2\pi_0\bar{F}_0(s)}-1}\stackrel{\mathbb P}{\to}0.
$$
Using our results on $D_n$ in the proof for compound BH with orthogonal rotations, we then get:
$$
\sup_{M\leq s\leq\zeta_n(\tau)}\abs{\widehat{\mathrm{FDR}^*}(s)-\mathrm{Fdr}_n(s)}\stackrel{\mathbb P}{\to}0.
$$
Since \smash{$\widehat{\mathrm{FDR}}(s)=\widehat{\mathrm{FDR}^*}(s)$} for $s\geq M$ on $A_n$, we can repeat the threshold argument for compound BH with $\mathcal{H}_{\mathrm{orth}}$ to conclude that \smash{$\zeta_n^{-1}(\hat{s})\stackrel{\mathbb P}{\to}\tau_\alpha$}. The rest of the proof is then identical to the proof for compound BH with $\mathcal{H}_{\mathrm{orth}}$, where we only recall the comment after~\eqref{eq:above_M_simplification} that on the event $A_n$ and when $\hat{s} \geq M$, then
the rejected hypotheses are precisely those with $\widehat{S}(\boldZ_i) \geq \hat{s}$.

We still have to fill in the argument that~\eqref{eq:seqstep_event} also holds for odd $K$. In this case, we have $\sum_k H_k=1$. Suppose first that $|\hat{\mu}(\boldz)|\leq|\hat{\mu}(H\boldz)|$. Then, by the triangle inequality and Jensen's inequality,
$$
|\hat{\mu}(\boldz)|
\leq|\hat{\mu}(H\boldz)|
=\abs{\frac{\hat{\mu}(\boldz)}{K}+\frac1K\sum_{k=1}^K H_k\{z_k-\hat{\mu}(\boldz)\}}\leq\frac{|\hat{\mu}(\boldz)|}{K}+\sqrt{\frac\nu K}\hat{\sigma}(\boldz).
$$
Rearranging gives $|\hat{\mu}(\boldz)|(1-1/K)\leq \sqrt{\nu/K}\hat{\sigma}(\boldz)$, and hence 
$$\widehat{S}(\boldz) = \sqrt{K(\nu + \hat{\nu}_0)}\frac{ |\hat{\mu}(\boldz)|}{ \{ \hat{\nu}_0 \hat{s}_0^2 + \nu \hat{\sigma}^2(\boldz)\}^{1/2}}
\leq\frac{K}{K-1}\sqrt{\nu+\hat\nu_0}.$$ 
If instead $|\hat{\mu}(\boldz)|>|\hat{\mu}(H\boldz)|$, we apply the same argument to $H\boldz$, using $H(H\boldz)=\boldz$. Thus,
$$
\min\cb{\widehat{S}(\boldz),\widehat{S}(H\boldz)}\leq\frac{K}{K-1}\sqrt{\nu+\hat\nu_0}
\;\text{ for all }\;\boldz\in\RR^K.
$$
Since $\hat\nu_0\stackrel{\mathbb P}{\to}\nu_0$ and recalling the definition of $M$, we conclude that~\eqref{eq:seqstep_event} also holds for odd $K$.

\end{proof}

\begin{proof}[Proof for BH with standard t-tests]

Let us write $\mathrm{Fdr}_{\infty,\nu}$ for the marginal FDR in the two-groups model with standard t-statistics $T_i = T(\boldZ_i) = T(\boldZ_i; 0, 1)$, that is,
$$
\mathrm{Fdr}_{\infty,\nu}(s) := \frac{\pi_0\bar{F}_{0,\nu}(s)}{\pi_0\bar{F}_{0,\nu}(s) + \pi_1\bar{F}_{1,\nu}(s)},\;\;\; s \geq 0,
$$
where $\bar{F}_{0,\nu}(s) = \bar{F}_{t,\nu}(s)$ and $\bar{F}_{1,\nu}(s) = \bar{F}_{t,\nu}\p{s/\{1+\lambda\}^{1/2}}$ are the null and alternative survival functions of $T_i$, respectively. Here the dependence on $n$ through $\pi_1$ and $\lambda$ is implicit.

Under~\eqref{eq:normal_samples}, we thus have for all $i:\mu_i=0$ that $P_i^t \sim \mathrm{Unif}[0,1]$, and moreover all p-values are independent. Thus we have that:
$$\mathrm{FDR}_n[t\textnormal{-}\mathrm{BH}(\alpha)] = \pi_0 \alpha \to \alpha \text{ as } n \to \infty.$$
Now let us turn to the power. Write $f_{0,\nu}$ and $f_{1,\nu}$ for the null and alternative densities of $T_i$. Notice that
\begin{equation}
\frac{f_{1,\nu}(s)}{f_{0,\nu}(s)} = \frac{1}{\sqrt{1+\lambda}}\p{\frac{1+s^2/\nu}{1+s^2/\{\nu(1+\lambda)\}}}^{(\nu+1)/2} \leq (1+\lambda)^{\nu/2},
\label{eq:t_density_ratio}
\end{equation}
that is, $f_{1,\nu}(s) \leq (1+\lambda)^{\nu/2} f_{0,\nu}(s)$ for all $s \geq 0$.
Integrating the above inequality yields $\bar{F}_{1,\nu}(s) \leq (1+\lambda)^{\nu/2}\bar{F}_{0,\nu}(s)$, and so
$$
1 \geq \mathrm{Fdr}_{\infty,\nu}(s) \geq \frac{\pi_0}{\pi_0 + \pi_1(1+\lambda)^{\nu/2}} \to 1
\;\text{ uniformly over }\; s \geq 0,
$$
since $\pi_1(1+\lambda)^{\nu/2} \to 0$ by our asymptotic regime in~\eqref{eq:lambda_asymptotics}.

Next, let us define
$$
D_n(s) := \frac{1}{n}\sum_{i=1}^n \ind\cb{\abs{T_i} \geq s},\;\;\;\;
\EE{D_n(s)} = 2\pi_0\bar{F}_{0,\nu}(s) + 2\pi_1\bar{F}_{1,\nu}(s).
$$
The FDR estimate used by BH at threshold $s$ is $2\bar{F}_{0,\nu}(s)/D_n(s)$, where we set the ratio to $\infty$ if $D_n(s)=0$. Notice that
$$
\frac{2\bar{F}_{0,\nu}(s)}{\EE{D_n(s)}} = \frac{\mathrm{Fdr}_{\infty,\nu}(s)}{\pi_0} \to 1
\;\text{ uniformly over }\; s \geq 0.
$$
Now fix $\eta >0$. 
Applying Lemma~\ref{lemm:gc_multiplicative_compound} to the independent functions averaged in $D_n(t)$, with $\delta = \eta\pi_1$ and $\varepsilon = (n\pi_1)^{-1/4}$, we find that
$$
\sup_{s \geq 0\,:\, \EE{D_n(s)} \geq \eta \pi_1} \abs{\frac{D_n(s)}{\EE{D_n(s)}} - 1} \stackrel{\mathbb P}{\to} 0 \text{ as } n \to \infty,
$$
which also implies that,
$$
\sup_{s \geq 0\,:\, D_n(s) \geq 2\eta \pi_1} \abs{\frac{2\bar{F}_{0,\nu}(s)}{D_n(s)} - 1} \stackrel{\mathbb P}{\to} 0 \text{ as } n \to \infty.
$$
However, BH must reject at a threshold $\hat{s}$ such that $2\bar{F}_{0,\nu}(\hat{s})/D_n(\hat{s}) \leq \alpha$. By the above display:
$$
\PP{D_n(\hat{s}) \geq 2\eta \pi_1} \leq  \PP{ \sup_{s \geq 0\,:\, D_n(s) \geq 2\eta \pi_1} \abs{\frac{2\bar{F}_{0,\nu}(s)}{D_n(s)} - 1} \geq  (1-\alpha)/2} \to 0,
$$
and since $\eta>0$ was arbitrary, we conclude that $D_n(\hat{s})/\pi_1 \stackrel{\mathbb P}{\to} 0$ as $n \to \infty$. 
Recalling~\eqref{eq:pi1_convergence}, we thus find that:
$$
\frac{\sum_{i=1}^n \ind\cb{\abs{T_i} \geq \hat{s},\, \mu_i \neq 0}}{1\lor\sum_{i=1}^n \ind\cb{\mu_i \neq 0} } \leq \frac{D_n(\hat{s})}{n^{-1}\lor \p{n^{-1}\sum_{i=1}^n \ind\cb{\mu_i \neq 0}}}  \stackrel{\mathbb P}{\to} 0 \; \text{ as }\; n \to \infty.
$$
We conclude by dominated convergence that
$$
\mathrm{Pow}_n[t\textnormal{-}\mathrm{BH}(\alpha)] \to 0 \text{ as } n \to \infty.
$$
\end{proof}

\begin{proof}[Proof for BH with separately sign flip randomization]

Let $R_n$ denote the number of rejections of $\mathrm{sBH}(\alpha, \mathcal{H}_{\mathrm{symm}})$.
Here our goal is to show that 
$$
\PP{R_n = 0} \to 1\, \text{ as }\, n \to \infty.
$$
This will imply the results on $\mathrm{FDR}_n$ and $\mathrm{Pow}_n$.
We start by conditioning on all $\mu_i$ and $\sigma_i^2$. By construction, it holds that $\psep_i \in \{2^{-\nu}, 2^{-\nu+1}, \ldots, 1\}$ for all $i$. Conditionally on $\mu_i = 0$, we have that $\psep_i \sim \mathrm{Unif}\{2^{-\nu}, 2^{-\nu+1}, \ldots, 1\}$. Let us define the following:
$$
P_i := \psep_i \ind\{ \mu_i =0\} + 2^{-\nu} \ind\{ \mu_i \neq 0\}.
$$
Since $P_i \leq \psep_i$ for all $i$, if we run BH on the $P_i$ we will reject at least as many hypotheses as if we run BH on the $\psep_i$. Thus it suffices to show that $\PP{R_n' = 0} \to 1$ where $R_n'$ is the number of rejections of BH run on the $P_i$.

Next, for any $t \in \{2 ^{-\nu}, 2^{-\nu+1}, \ldots, 1\}$, we have that as $n \to \infty$,
$$
\begin{aligned}
\frac{n t}{\sum_{i=1}^n \ind\{P_i \leq t\}} = \frac{n t}{\sum_{i=1}^n \ind\{\mu_i = 0, P_i \leq t\} + \sum_{i=1}^n \ind\{\mu_i \neq 0\}}  \stackrel{\mathbb P }{\to} \frac{t}{t + 0} = 1.
\end{aligned}
$$
This means that
$$
\PP{R_n' =0} = \PP{  \frac{n t}{\sum_{i=1}^n \ind\{P_i \leq t\}} > \alpha \;\text{ for all }\; t \in \{2 ^{-\nu}, 2^{-\nu+1}, \ldots, 1\} } \to 1 \text{ as } n \to \infty,
$$
as claimed.
\end{proof}

\subsection{Proof of Proposition~\ref{prop:sparse_regime_sens}}
\begin{proof}
Under~\eqref{eq:normal_samples} and~\eqref{eq:lonnsted_prior}, we have the following distribution results:
$$
\begin{aligned}
&\frac{X_i}{V_i} \mid (\mu_i=0) \sim t_{K-2},\;\; \;\;&&\frac{X_i}{V_i} \mid (\mu_i\neq0) \sim \sqrt{1+\lambda}\cdot t_{K-2},\\
&\frac{X_i^0}{V_i} \mid (\mu_i=0) \sim t_{K-2} \;\;&&\frac{X_i^0}{V_i} \mid (\mu_i\neq0) \sim t_{K-2}.
\end{aligned}
$$
We call $f_{0,K-2}$ and $f_{1,K-2}$ the null and non-null densities of $X_i/V_i$. The same argument as in~\eqref{eq:t_density_ratio} (with $K-2$ in place of $\nu$) gives,
\begin{equation}
\frac{f_{1,K-2}(s)}{f_{0,K-2}(s) }\leq(1+\lambda)^{(K-2)/2}\,\text{ for all }\,s\in\RR.
\label{eq:t_density_ratio_K_minus_2}
\end{equation}
Moreover, $s \mapsto f_{1,K-2}(s)/f_{0,K-2}(s)$ is strictly increasing in $|s|$ by Lemma~\ref{lemm:monotonicity_finite_sample}, and so
$$ 
\PP{ \mu_i=0\, \mid \, X_i/V_i = s} = \frac{\pi_0 f_{0,K-2}(s)}{\pi_0 f_{0,K-2}(s) + \pi_1 f_{1,K-2}(s)}
$$
is strictly decreasing in $|s|$. Since the transformation $s\mapsto\Phi^{-1}(F_{t,K-2}(s))$ in~\eqref{eq:sens_t} is odd and strictly increasing, $\hat{g}(t)$ likewise depends only on $|t|$ and is strictly decreasing in $|t|$. 

Recall the notation $U_i,U_i^0,W_i$ as defined in~\eqref{eq:hat_sens}, with the oracle score $\hat{g}$. We define
$$
D_n(w):=\frac{1}{n}\sum_{i=1}^n\ind\cb{W_i\geq w},\;\;\;\;
N_n(w):=\frac{1}{n}\sum_{i=1}^n\ind\cb{W_i\leq-w}.
$$
For all $w\geq0$, it holds that:
$$
\begin{aligned}
\PP{W_i\geq w\mid\mu_i=0}
&=\PP{W_i\leq-w\mid\mu_i=0} = \frac{1}{2}\PP{|W_i|\geq w\mid\mu_i=0}\\
&=\frac12\PP{\exp(-U_i)\lor\exp(-U_i^0)\geq w\mid\mu_i=0}
\geq\frac12\PP{\exp(-U_i)\geq w\mid\mu_i=0}.
\end{aligned}
$$
By definition, $W_i\geq w$ implies $\exp(-U_i)\geq w$. Rearranging and then integrating~\eqref{eq:t_density_ratio_K_minus_2} with respect to $s$, we get:
$$
\begin{aligned}
\PP{W_i\geq w\mid\mu_i\neq0}
&\leq\PP{\exp(-U_i)\geq w\mid\mu_i\neq0}\\
&=\int_{\RR}\ind\cb{\exp\bigl[-\hat{g}\{\Phi^{-1}(F_{t,K-2}(s))\}\bigr]\geq w}f_{1,K-2}(s)\dd s\\
&\leq(1+\lambda)^{(K-2)/2}\PP{\exp(-U_i)\geq w\mid\mu_i=0}\\
&\leq2(1+\lambda)^{(K-2)/2}\PP{W_i\geq w\mid\mu_i=0}.
\end{aligned}
$$
Now we have: 
$$
\begin{aligned}
\EE{N_n(w)} &= \pi_0\PP{W_i\leq-w\mid\mu_i=0} + \pi_1\PP{W_i\leq-w\mid\mu_i\neq0}\\ 
&\geq \pi_0\PP{W_i\leq-w\mid\mu_i=0} \\ 
&=  \pi_0\PP{W_i\geq w\mid\mu_i=0}.
\end{aligned}
$$
Meanwhile,
$$
\begin{aligned}
\EE{D_n(w)} &= \pi_0\PP{W_i\geq w\mid\mu_i=0} + \pi_1\PP{W_i\geq w\mid\mu_i\neq0}\\
&\leq\cb{\pi_0+2\pi_1(1+\lambda)^{(K-2)/2}}\PP{W_i\geq w\mid\mu_i=0}.
\end{aligned}
$$
Thus
$$
\frac{\EE{N_n(w)}}{\EE{D_n(w)}}
\geq\frac{\pi_0}{\pi_0+2\pi_1(1+\lambda)^{(K-2)/2}},
$$
for all $w\geq0$ with $\EE{D_n(w)}>0$. Consequently,
$$
\p{1-\frac{\EE{N_n(w)}}{\EE{D_n(w)}}}_+
\leq\frac{2\pi_1(1+\lambda)^{(K-2)/2}}{\pi_0+2\pi_1(1+\lambda)^{(K-2)/2}}
\to0\text{ uniformly over these }w,
$$
since $\pi_1(1+\lambda)^{(K-2)/2}\to0$ by~\eqref{eq:lambda_asymptotics}.

Now fix $\eta>0$. Applying Lemma~\ref{lemm:gc_multiplicative_compound} to the independent functions averaged in $D_n(t)$ and $N_n(t)$, as in the proof for BH with standard t-tests in Theorem~\ref{theo:sparse_regime}, we find that
$$
\sup_{w\geq0\,:\,D_n(w)\geq2\eta\pi_1}
\p{1-\frac{n^{-1}+N_n(w)}{n^{-1}\lor D_n(w)}}_+
\stackrel{\mathbb P}{\to}0\text{ as }n\to\infty.
$$
However, the threshold $\hat{w}$ in~\eqref{eq:seqstep_hatw} must have this ratio at most $\alpha$ whenever there are rejections. Thus
$$
\PP{D_n(\hat{w})\geq2\eta\pi_1}\to 0\,\text{ as }\,n\to\infty,
$$
and since $\eta>0$ was arbitrary, we conclude that $D_n(\hat{w})/\pi_1\stackrel{\mathbb P}{\to}0$.

To conclude we note that
$$
\frac{\sum_{i=1}^n\ind\cb{W_i\geq\hat{w},\,\mu_i\neq0}}{1\lor\sum_{i=1}^n\ind\cb{\mu_i\neq0}}
\leq
\frac{D_n(\hat{w})/\pi_1}{\p{1\lor\sum_{i=1}^n\ind\cb{\mu_i\neq0}}/(n\pi_1)}
\stackrel{\mathbb P}{\to}0,
$$
where the denominator on the right converges to $1$ by~\eqref{eq:pi1_convergence}. By dominated convergence, $\mathrm{Pow}_n[\mathrm{SENS\textnormal{-}limma}(\alpha)] \to 0$.
\end{proof}

\end{document}

%% file: figures/design_swaps_methylation.tikz
\begin{tikzpicture}[font=\normalsize,
  swap/.style={<->,>={Stealth[length=1.6mm,width=1.25mm]},line width=.85pt}]
  \matrix (D) [matrix of math nodes,
    nodes={minimum width=9mm,minimum height=7mm,inner sep=0pt},
    column sep=.5mm,row sep=0pt] {
    0&1&0&0&0&0\\
    1&1&0&0&0&0\\
    0&1&0&0&1&0\\[4mm]
    0&1&1&0&0&0\\
    0&1&1&0&1&0\\
    0&1&1&0&0&1\\[4mm]
    0&1&0&1&0&0\\
    1&1&0&1&0&0\\
    0&1&0&1&1&0\\
    0&1&0&1&0&1\\
  };
  \begin{scope}[on background layer]
    \foreach \first/\last in {1/3,4/6,7/10} {
      \fill[restgreen!12] (D-\first-1.north west) rectangle (D-\last-1.south east);
      \fill[black!3] (D-\first-2.north west) rectangle (D-\last-6.south east);
    }
    \foreach \i in {1,2,7,8} {
      \fill[restgreen!7] (D-\i-2.north west) rectangle (D-\i-6.south east);
    }
  \end{scope}
  \foreach \j/\first/\second in {
    1/{Rest.}/{Treg},2/{}/{$\mathbf{1}$},3/{}/{M29},
    4/{}/{M30},5/{Act.}/{naive},6/{Act.}/{Treg}} {
    \node[font=\small,anchor=base]
      at ($(D-1-\j.center)+(0,12mm)$) {\first};
    \node[font=\small,anchor=base]
      at ($(D-1-\j.center)+(0,8mm)$) {\second};
  }
  \node[text=restgreen] at ($(D-1-1.north)+(0,21mm)$) {$W$};
  \node at ($(D-1-2.north west)!0.5!(D-1-6.north east)+(0,21mm)$) {$X$};
  \draw[restgreen!60,line width=.4pt]
    ($(D-1-1.north west)+(0,17mm)$) -- ($(D-1-1.north east)+(0,17mm)$);
  \draw[black!35,line width=.4pt]
    ($(D-1-2.north west)+(0,17mm)$) -- ($(D-1-6.north east)+(0,17mm)$);
  \draw[line width=.55pt] ($(D.north west)+(1mm,0)$) --
    ($(D.north west)+(-1mm,0)$) -- ($(D.south west)+(-1mm,0)$) --
    ($(D.south west)+(1mm,0)$);
  \draw[line width=.55pt] ($(D.north east)+(-1mm,0)$) --
    ($(D.north east)+(1mm,0)$) -- ($(D.south east)+(1mm,0)$) --
    ($(D.south east)+(-1mm,0)$);
  \draw[gray!55] ($(D-1-1.north east)+(.5mm,0)$) --
    ($(D-10-1.south east)+(.5mm,0)$);
  \foreach \i/\label in {1/Naive,2/Rest.\ Treg,3/Act.\ naive,
    4/Naive,5/Act.\ naive,6/Act.\ Treg,
    7/Naive,8/Rest.\ Treg,9/Act.\ naive,10/Act.\ Treg} {
    \node[anchor=east] at ($(D-\i-1.west)+(-17mm,0)$) {\label};
  }
  \node[anchor=east,font=\small,text=black!65]
    at ($(D-2-1.west)+(-36mm,0)$) {M28};
  \node[anchor=east,font=\small,text=black!65]
    at ($(D-5-1.west)+(-36mm,0)$) {M29};
  \node[anchor=east,font=\small,text=black!65]
    at ($(D-8-1.west)!0.5!(D-9-1.west)+(-36mm,0)$) {M30};
  \coordinate (lane3) at ($(D-1-1.west)+(-8.5mm,0)$);
  \foreach \k/\x in {1/.9,2/2.6,4/4.3} {
    \coordinate (lane\k) at ($(D-1-6.east)+(\x,0)$);
  }
  \foreach \i in {1,...,10} {
    \draw[black!32,line width=.4pt] ($(D-\i-1.west)+(-16mm,0)$) --
      ($(D-\i-1.west)+(-3mm,0)$);
    \draw[black!32,line width=.4pt] ($(D-\i-6.east)+(3mm,0)$) --
      ($(D-\i-6.east)+(49mm,0)$);
  }
  \foreach \k/\col/\first/\second in {
    1/naiveblue/{Act.\ naive}/{naive},
    2/tregred/{Act.\ Treg}/{Rest.\ Treg},
    3/restgreen/{Rest.\ Treg}/{naive},
    4/actpurple/{Act.\ Treg}/{Act.\ naive}} {
    \node[text=\col,font=\small,anchor=base]
      at ($(lane\k)+(0,16mm)$) {\first};
    \node[text=\col,font=\small,anchor=base]
      at ($(lane\k)+(0,12mm)$) {vs.};
    \node[text=\col,font=\small,anchor=base]
      at ($(lane\k)+(0,8mm)$) {\second};
    \draw[\col!50,line width=.4pt] ($(lane\k)+(-6.5mm,6mm)$) --
      ($(lane\k)+(6.5mm,6mm)$);
  }
  \foreach \k/\i/\j/\col in {1/1/3/naiveblue,1/4/5/naiveblue,
    1/7/9/naiveblue,2/8/10/tregred,3/1/2/restgreen,3/7/8/restgreen,
    4/5/6/actpurple,4/9/10/actpurple} {
    \draw[swap,\col] (lane\k |- D-\i-1.center)
      .. controls ($(lane\k |- D-\i-1.center)+(3mm,0)$)
              and ($(lane\k |- D-\j-1.center)+(3mm,0)$) ..
      (lane\k |- D-\j-1.center);
    \foreach \row in {\i,\j} {
      \fill[\col] ($(lane\k |- D-\row-1.center)+(-.7mm,0)$) circle (.45mm);
    }
  }
\end{tikzpicture}

%% file: tables/methylation.tex
\begin{table}
\centering
\caption{Prior estimates $(\hat\nu_0,\hat s_0^2)$ and number of discoveries in the methylation study at $\alpha=0.05$, for the four contrasts in Figure~\ref{fig:design_swaps_methylation}.}
\label{tab:methylation-discoveries}
\small
\setlength{\tabcolsep}{4pt}
\begin{tabular}{@{}lcccc@{}}
\toprule
& \textcolor{restgreen}{\shortstack{\strut Rest.\ Treg\\\strut vs.\\\strut naive}}
& \textcolor{naiveblue}{\shortstack{\strut Act.\ naive\\\strut vs.\\\strut naive}}
& \textcolor{tregred}{\shortstack{\strut Act.\ Treg\\\strut vs.\\\strut Rest.\ Treg}}
& \textcolor{actpurple}{\shortstack{\strut Act.\ Treg\\\strut vs.\\\strut Act.\ naive}} \\
\midrule
\# Permutations: effective (full) & $4\ (4)$ & $4\ (8)$ & $2\ (2)$ & $4\ (4)$ \\
\midrule
Limma prior & \multicolumn{4}{c}{$(3.96,0.055)$} \\
Orbit quantile prior & $(3.46,0.062)$ & $(3.66,0.065)$ & $(3.54,0.060)$ & $(3.43,0.059)$ \\
\midrule
BH + t-test & $4$ & $1$ & $0$ & $0$ \\
BH + Limma & $3\,023$ & $627$ & $0$ & $1\,478$ \\
Compound BH with rotations & $3\,033$ & $815$ & $12$ & $1\,723$ \\
DDR with rotations & $2\,989$ & $802$ & $12$ & $1\,690$ \\
BH + (compound/separate) permutations & $0$ & $0$ & $0$ & $0$ \\
SeqStep+ & $3\,095$ & $1\,280$ & $88$ & $2\,367$ \\
\bottomrule
\end{tabular}
\end{table}